\documentclass[a4paper,12pt]{article}

\usepackage[right=1in,left=1in,top=1in,bottom=1in]{geometry}
\usepackage[T1]{fontenc}
\usepackage{amsfonts,amsmath,amssymb,amsthm,mathrsfs,mathtools}
\usepackage[numbers,square]{natbib}
\usepackage{url}
\usepackage{xcolor}
\definecolor{deepnavy}{rgb}{0.04,0.24,0.57}
\definecolor{deepernavy}{rgb}{0.02,0.12,0.285}
\definecolor{pathzero}{RGB}{145,30,55}
\definecolor{pathone}{RGB}{0,100,70}
\usepackage{graphicx}
\makeatletter
\renewcommand{\fnum@figure}{\textbf{\figurename\nobreakspace\thefigure}}
\makeatother
\DeclareRobustCommand{\captionlead}[1]{{\bfseries\boldmath #1}}
\usepackage{sectsty}
\usepackage{titletoc}
\usepackage{setspace}
\usepackage{bbm}
\usepackage{MnSymbol}
\usepackage[all]{xy}
\usepackage{dsfont}
\usepackage{tabularx}
\usepackage{tabularray}
\usepackage{multicol}
\usepackage{enumitem}
\usepackage{booktabs}
\usepackage{float}
\usepackage{microtype}
\usepackage{needspace}
\usepackage{fancyhdr}
\usepackage{quantikz}
\usetikzlibrary{arrows.meta,calc,positioning}
\usepackage{tikz-cd}
\usepackage{pgfplots}
\pgfplotsset{compat=1.18}
\usepackage[colorlinks,hypertexnames=false]{hyperref}
\hypersetup{
  colorlinks=true,
  linkcolor=deepernavy,
  citecolor=deepernavy,
  urlcolor=deepernavy,
  pdftitle={The Quantum Composition Paradox},
  pdfauthor={Jacob Biamonte},
  pdfsubject={Fixed-context Born-kernel composition, reversible state machines, and quantum music},
  pdfkeywords={Born rule, stochastic divisibility, quantum histories, quantum state machines, quantum music, quantum circuits}
}

\sectionfont{\large}
\bibpunct{[}{]}{,}{n}{}{,}
\makeatletter
\g@addto@macro\th@plain{\thm@notefont{\bfseries}}
\g@addto@macro\th@definition{\thm@notefont{\bfseries}}
\makeatother

\newtheorem{theorem}{Theorem}[section]
\newtheorem{proposition}[theorem]{Proposition}
\newtheorem{lemma}[theorem]{Lemma}
\newtheorem{corollary}[theorem]{Corollary}
\theoremstyle{definition}
\newtheorem{definition}[theorem]{Definition}
\newtheorem{example}[theorem]{Example}
\newtheorem{remark}[theorem]{Remark}

\newcommand{\Born}{\mathsf B}
\newcommand{\U}{\mathrm U}
\newcommand{\C}{\mathcal C}
\newcommand{\Mc}{\mathcal M_{\C}}

\newcommand{\Tr}{\operatorname{Tr}}
\newcommand{\rank}{\operatorname{rank}}

\newcommand{\openone}{\mathds{1}}
\newcommand{\ii}{\imath}
\newcommand{\R}{\mathbb R}
\newcommand{\HS}{\operatorname{HS}}
\newcommand{\PromiseBQP}{\mbox{\textup{\textsc{PromiseBQP}}}}
\fancypagestyle{appendixpages}{%
  \fancyhf{}%
  \fancyfoot[C]{\scriptsize\color{black!50}}%
  \fancyfoot[L]{\scriptsize\color{black!50}Appendix to \textit{The Quantum Composition Paradox}}%
  \fancyfoot[R]{\scriptsize\color{black!50}Page \thepage\ of \pageref*{app:last-page}}%
}

\title{\textbf{The Quantum Composition Paradox}}
\author{Jacob Biamonte}
\date{}

\begin{document}
\setcitestyle{numbers,square}
\maketitle
\begin{center}
\small {\'E}TS Montr{\'e}al, Universit{\'e} du Qu{\'e}bec, Montreal, QC, Canada
\end{center}
\vspace{0.5em}

\begingroup
\bfseries
Quantum theory does not generally permit
the probability laws obtained from individual unitary steps by the Born rule
to be sewn into a consistent genealogy; we classify the exceptions
and show that faithful composition can hold from an initial boundary yet fail
after an internal restart.
For finite-dimensional, composition-closed unitary families, universal
composition holds exactly for
unitary monomials, the
phase-dressed permutations whose Born kernels realize reversible deterministic
state machines.  Prescribed sequences evade this obstruction.
We classify all pairs of qubit unitary steps, give a necessary-and-sufficient criterion
for pairs of qutrit unitary steps, and prove that
a boundary-stable unitary sequence on a $d$-dimensional Hilbert space
contains
at most $d$ fully mixing steps, with equality in every prime dimension.  We also
construct arbitrarily long genuinely mixing sequences that compose from their
initial boundary but fail after an internal restart, and
a qutrit-controlled
two-qubit realization with active interference.  We define a
Born--Chapman--Kolmogorov current that vanishes exactly when coherent and
stepwise-checked endpoint laws agree, together with an associated measure of
how many bits the endpoint reveals about the intermediate checking schedule.
This state-machine connection provides a foundation for a quantum theory of music,
in which unitary operations are notes and temporal boundaries are cues.
Musical-transition prediction is proved \PromiseBQP-complete, and the stepwise
patterns that preserve a genealogy specify rules for rhythm, whereas the exceptional failure
of that genealogy from an internal cue
is the quantum music paradox.\\
\par
\endgroup

\noindent Key words: Born rule, stochastic divisibility, quantum histories,
quantum state machines, quantum process theory, quantum music, quantum circuits

\newpage

\section{The composition problem of quantum mechanics}

Quantum theory offers two distinct routes to the same final time.  If
observations remain unchecked, amplitudes combine and interfere; if an
intermediate observation creates a record, probabilities combine
instead.  The double slit makes the resulting disagreement vivid.  Its
experimental lineage runs from Young's interference experiment to J\"onsson's
electron realization~\cite{Young1804,Jonsson1961}; Born's probability rule
and Bohr's complementarity supply its enduring conceptual vocabulary~\cite{Born1926,Bohr1928}.

Feynman eventually said that this phenomenon
``\textit{has in it the heart of
quantum mechanics.  In reality, it contains the only mystery}''
~\cite{FeynmanLectures}.  Here that mystery appears as a failure of
composition: coherent amplitudes compose, but their Born probabilities
generally do not.  Classifying when they do compose poses a quantum--classical
interface challenge and raises a sharp question: does a faithful genealogy
survive restriction to a newly prepared internal boundary,
as it must for a consistent
temporal model that can be restarted from that boundary?

The double slit provides the familiar negative case: stepwise and
endpoint-checked routes assign different probabilities.  This question
narrows the measurement problem.  Quantum mechanics
has long supplied probability laws for successive observations
~\cite{Luders,WignerMeasurement}; in the fixed-context setting used here,
summing over intermediate outcomes gives the stepwise-checked Born probability
kernel.  We first classify
when it equals the endpoint Born kernel of the same evolution without
intermediate checks, and then whether that same equality survives an internal boundary.

\begin{figure}[htbp]
\centering
\begin{minipage}[t]{0.48\textwidth}
\centering
\begin{tikzpicture}
\begin{axis}[
 width=\linewidth,height=4.0cm,
 axis lines=left,
 xmin=-3.2,xmax=3.2,ymin=0,ymax=1.08,
 xtick=\empty,ytick=\empty,
 xlabel={$x$},ylabel={$p_{\mathrm{end}}(x)$},
 xlabel style={at={(axis description cs:1,0)},anchor=north east},
 ylabel style={at={(axis description cs:0,1)},anchor=south east,rotate=-90}]
 \addplot[black,thick,domain=-3.2:3.2,samples=280]
  {exp(-x^2/4.2)*(cos(deg(4.8*x)))^2};
\end{axis}
\end{tikzpicture}

\small Endpoint checked:
$p_{\mathrm{end}}(x)=|\psi_1(x)+\psi_2(x)|^2$.
\end{minipage}\hfill
\begin{minipage}[t]{0.48\textwidth}
\centering
\begin{tikzpicture}
\begin{axis}[
 width=\linewidth,height=4.0cm,
 axis lines=left,
 xmin=-3.2,xmax=3.2,ymin=0,ymax=1.25,
 xtick=\empty,ytick=\empty,
 xlabel={$x$},ylabel={$p_{\mathrm{step}}(x)$},
 xlabel style={at={(axis description cs:1,0)},anchor=north east},
 ylabel style={at={(axis description cs:0,1)},anchor=south east,rotate=-90}]
 \addplot[black,thick,domain=-3.2:3.2,samples=180]
  {0.62*exp(-0.55*(x+0.75)^2)};
 \addplot[black,thick,domain=-3.2:3.2,samples=180]
  {0.62*exp(-0.55*(x-0.75)^2)};
 \addplot[black,thick,densely dashed,domain=-3.2:3.2,samples=180]
  {0.62*exp(-0.55*(x+0.75)^2)+0.62*exp(-0.55*(x-0.75)^2)};
 \node at (axis cs:-1.15,0.34) {$p_1$};
 \node at (axis cs:1.15,0.34) {$p_2$};
 \node at (axis cs:0,1.08) {$p_1+p_2$};
\end{axis}
\end{tikzpicture}

\small Stepwise checked:
$p_{\mathrm{step}}(x)=p_1(x)+p_2(x)$.
\end{minipage}

\medskip
\resizebox{0.96\linewidth}{!}{%
\begin{tikzpicture}[
 >=stealth,
 line width=0.65pt,
 block/.style={draw,minimum width=1.55cm,minimum height=0.58cm,
               align=center,font=\footnotesize},
 route/.style={font=\footnotesize,anchor=east},
 output/.style={font=\footnotesize,anchor=west}
]
 \node[route] at (0,0.72) {Endpoint checked};
 \node[circle,fill=black,inner sep=2.1pt] (esource) at (0.65,0.72) {};
 \node[font=\scriptsize,below] at (esource.south) {source};
 \node[block] (eslits) at (2.65,0.72) {two slits};
 \node[block] (escreen) at (7.75,0.72) {screen};
 \node[output] (eout) at (9.05,0.72) {$p_{\mathrm{end}}(x)$};
 \draw[->] (esource) -- (eslits);
 \draw[->] (eslits) -- (escreen);
 \draw[->] (escreen) -- (eout);

 \node[route] at (0,-0.72) {Stepwise checked};
 \node[circle,fill=black,inner sep=2.1pt] (ssource) at (0.65,-0.72) {};
 \node[font=\scriptsize,below] at (ssource.south) {source};
 \node[block] (sslits) at (2.65,-0.72) {two slits};
 \node[block,minimum width=2.05cm] (scheck) at (5.25,-0.72)
   {which-path check};
 \node[block] (sscreen) at (7.75,-0.72) {screen};
 \node[output] (sout) at (9.05,-0.72) {$p_{\mathrm{step}}(x)$};
 \draw[->] (ssource) -- (sslits);
 \draw[->] (sslits) -- (scheck);
 \draw[->] (scheck) -- (sscreen);
 \draw[->] (sscreen) -- (sout);
\end{tikzpicture}%
}
\caption{\captionlead{The double slit as a quantum--classical composition failure.}  Endpoint-checked
alternatives add as amplitudes and interfere.  Stepwise-checked alternatives
add as probabilities.  The two processes assign different
distributions to the final outcomes.}
\label{fig:double-slit-composition}
\end{figure}
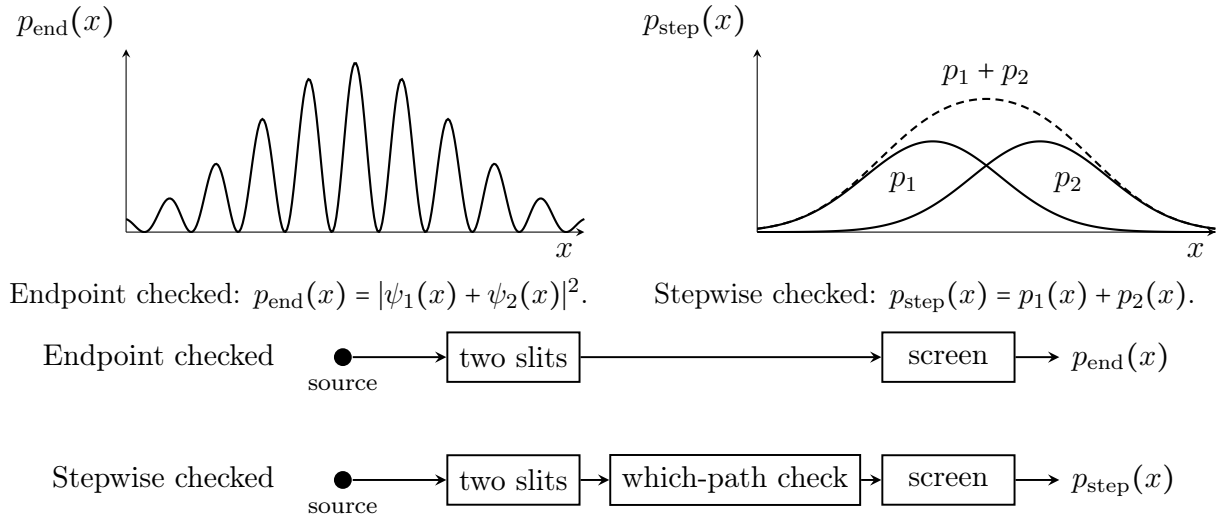

Let $X=\{0,\ldots,d-1\}$ label the basis of a $d$-dimensional Hilbert
space $\mathsf H\cong\mathbb C^d$, and fix the atomic context
$\C=\{P_x=|x\rangle\langle x|:x\in X\}$, with
$\sum_{x\in X}P_x=\openone$.  A $\C$-diagonal initial (entrance) state has
the form $\rho=\sum_{x\in X}q(x)P_x$ for any probability distribution $q$
 over $X$.
For any unitary $U\in\U(d)$, write
$\Born(U)=\Born_{\C}(U)$ for its \emph{Born kernel}, the transition matrix
in this context:
\begin{equation}
 \begin{gathered}
 \Born(U):=U\odot\overline U,\\[-0.2ex]
 \text{\(\langle y|\Born(U)|x\rangle
 =\langle yy|\bigl(U\otimes\overline U\bigr)|xx\rangle
 =|\langle y|U|x\rangle|^2\)}.
 \end{gathered}
 \label{eq:story-born-map}
\end{equation}
Here $\odot$ is elementwise (Hadamard) product and the overbar denotes
entrywise complex conjugation in the fixed context.  This matrix element is
the transition probability
from $|x\rangle$ to $|y\rangle$ under $U$.
Each $\Born(U)$ is unistochastic, hence doubly stochastic, and therefore a
valid classical transition map over the chosen outcomes.  Every kernel
identity below holds for every context-basis
preparation and hence for every probability distribution over that context.
The context $\C$ fixes the basis in which the entrywise conjugation
and Hadamard product are taken.

We use $\circ$ for chronological composition of
both unitaries and their stochastic
kernels, where $V\circ U$ means that $U$ acts
before $V$.

A family of Born kernels assigned to allowed contiguous passages is a
\emph{Born-induced stochastic genealogy} when, at every allowed cut, the
kernel of the whole passage equals the chronological product of the kernels
of its subpassages.

\begingroup
\setlength{\emergencystretch}{2em}
For a unitary $U$ followed by a unitary $V$,
define what we call the signed \emph{Born--Chapman--Kolmogorov current},
or \emph{Born-composition defect},
and its normalized Frobenius norm by
\begin{align}
 \mathcal J_{\C}(V,U)
 &:=\Born_{\C}(V\circ U)
   -\Born_{\C}(V)\circ\Born_{\C}(U),
 \label{eq:born-chapman-current}\\
  0\leq \nu_{\C}(V,U)
 &:=\frac{1}{\sqrt{d-1}}
 \left\|\mathcal J_{\C}(V,U)\right\|_{\mathrm F}
 \leq 1.
 \label{eq:current-size-main}
\end{align}
Here $d\geq2$, and the upper bound is sharp by
Theorem~\ref{thm:bound-main}.
This terminology denotes the
signed departure from the Chapman--Kolmogorov equation which is a faithful
consistency relation for classical transition kernels in stochastic processes
and statistical mechanics~\cite{Chapman,Kolmogorov,BaezBiamonte}.

To quantify the information revealed
from an intermediate measurement, prepare a
uniformly chosen context-basis state.  In this two-step experiment, $X$ is a uniformly sampled,
recorded input basis label, $S$ is an independent fair bit selecting whether
to insert the unread intermediate context check, and $Y$ is the recorded
output label.  Write
$\mathcal I(S;Y\mid X)$ for the conditional
mutual information about the checking
schedule, with $0\leq\mathcal I\leq1$ bit per run in the uniform-input,
two-step case of Proposition~\ref{prop:schedule-information}.
Here both recorded variables $X$ and $Y$ take values in the same
outcome alphabet $X$ associated with $\C=\{P_x:x\in X\}$.

Figure~\ref{fig:quantum-stochastic-composition}
gives equivalent formulations of the quantum--classical consistency
law for a faithful Born-kernel genealogy.
The endpoint
probabilities coincide, whereas the intermediate check may still change the
quantum state.

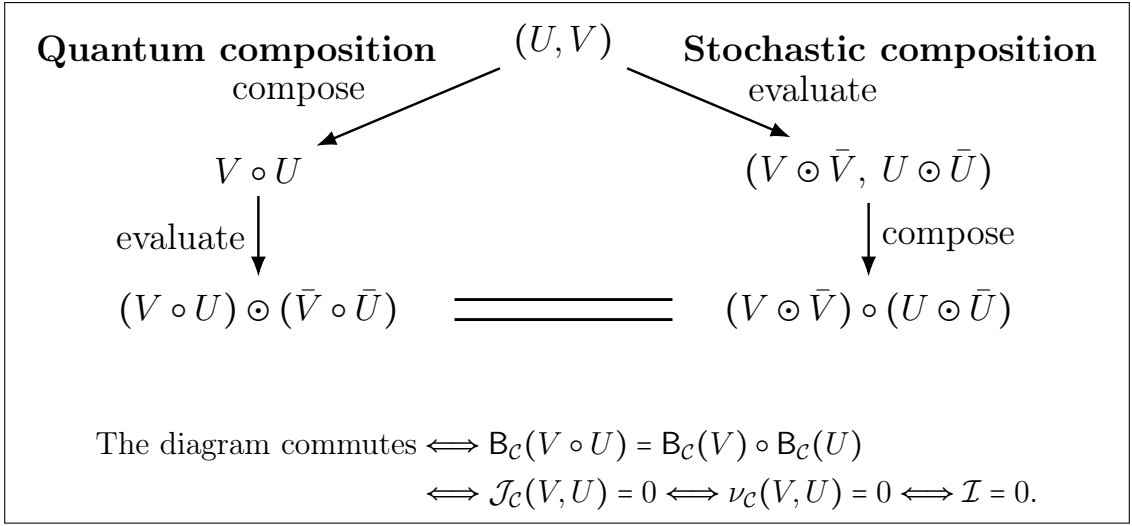
\begin{figure}[H]
\centering
\fbox{\begin{minipage}{0.92\linewidth}
\centering
\resizebox{0.98\linewidth}{!}{%
\begin{tikzpicture}[>=Latex,line width=0.8pt,yscale=0.82]
  \node (start) at (0,2.00) {$(U,V)$};

  \node[font=\bfseries] at (-3.75,1.85) {Quantum composition};
  \node[font=\bfseries] at (3.75,1.85) {Stochastic composition};

  \node (qcomp) at (-3.50,0.20) {$V\circ U$};
  \node (ceval) at (3.50,0.20)
    {$(V\odot\bar V,\;U\odot\bar U)$};

  \node (qeval) at (-3.50,-1.75)
    {$(V\circ U)\odot(\bar V\circ\bar U)$};
  \node (ccomp) at (3.50,-1.75)
    {$(V\odot\bar V)\circ(U\odot\bar U)$};

  \draw[->] (start) -- node[pos=0.66,above left] {compose} (qcomp);
  \draw[->] (start) -- node[pos=0.66,above right] {evaluate} (ceval);
  \draw[->] (qcomp) -- node[left] {evaluate} (qeval);
  \draw[->] (ceval) -- node[right] {compose} (ccomp);

  \draw[line width=0.95pt] (-1.25,-1.61) -- (1.25,-1.61);
  \draw[line width=0.95pt] (-1.25,-1.89) -- (1.25,-1.89);
\end{tikzpicture}%
}
\par\smallskip

\[
 \begin{aligned}
 \text{The diagram commutes}
 &\Longleftrightarrow
 \Born_{\C}(V\circ U)=\Born_{\C}(V)\circ\Born_{\C}(U)\\
 &\Longleftrightarrow
 \mathcal J_{\C}(V,U)=0
 \Longleftrightarrow\nu_{\C}(V,U)=0
 \Longleftrightarrow\mathcal I=0.
 \end{aligned}
\]
\end{minipage}}
\caption[The quantum--classical consistency law for faithful Born-kernel composition.]{%
\captionlead{The quantum--classical consistency law for faithful Born-kernel composition.}
The left route
composes the unitaries before Born evaluation; the right route
composes their Born kernels.}
\label{fig:quantum-stochastic-composition}
\end{figure}
\endgroup

\begingroup
\setlength{\emergencystretch}{2em}
Several related notions must be separated.  Pollock et al.~and Milz
and Modi distinguish operational process-tensor Markovianity from weaker
divisibility tests: agreement at the level of intermediate maps need not
supply a full multitime Markov process~\cite{PollockEtAl,MilzModi}.
Multitime Kolmogorov consistency asks when inserting or removing measurements
preserves the relevant
statistics~\cite{SmirneEtAl2019,GarciaDiazAccessible}
and~\cite{StrasbergGarciaDiaz,SakuldeeTarantoMilz}.
Our equality is the complete specialization to prescribed closed-system
unitary passages.
Barandes develops a formulation of quantum theory in terms of
indivisible stochastic processes.  His framework identifies permutations as
stochastic maps with stochastic inverses and relates generic
Born-kernel noncomposition to interference~\mbox{\cite{BarandesCorrespondence,BarandesTheorem,BarandesProcesses}}.
We classify the exact relevant cases.
The Born-kernel composition law on
a unitary semigroup yields
unitary monomials, whereas prescribed genuinely mixing unitaries can
compose on some passages and fail
on others;
composition need not survive all shifts
inside the preparation boundary.  The framework also
supplies the Born--Chapman--Kolmogorov current and its exact
schedule-information measure, a consistency theorem for reversible state
machines, and an operation-supplied quantum counterpart of Xenakis-style
stochastic music~\cite{Xenakis1992,Ames,Gibson2023}.  In the
circuit
score model, musical-transition prediction is proved to be
\PromiseBQP-complete.
\endgroup

\Needspace{7\baselineskip}
The conditions under which the
diagram in Fig.~\ref{fig:quantum-stochastic-composition} commutes
admit an equivalent path-sum formulation.  For
entrance $x$, exit $y$, and paths $\gamma=(x,z,y)$ through
intermediate context outcomes $z$, write
$\mathcal A(\gamma)=\langle y|V|z\rangle\langle z|U|x\rangle$.  Then
\begin{equation}
 \begin{aligned}
 \langle y|\mathcal J_{\C}(V,U)|x\rangle
 &=\left|\sum_\gamma\mathcal A(\gamma)\right|^2
   -\sum_\gamma|\mathcal A(\gamma)|^2\\
 &=\sum_{\gamma\ne\gamma'}\mathcal A(\gamma)
   \overline{\mathcal A}(\gamma').
 \end{aligned}
 \label{eq:main-path-comparison}
\end{equation}
Here $\gamma\ne\gamma'$ means $z\ne z'$; the
overbar denotes complex
conjugation.

\Needspace{5\baselineskip}

The diagram commutes exactly when this aggregate interference
vanishes for every entrance--exit pair; the normalized current
$\nu_{\C}(V,U)$ quantifies the failure of composition.

Equation~\eqref{eq:main-path-comparison} makes the distinction from
consistent histories explicit.
Weak or
medium decoherence for every context-basis entrance implies endpoint
Born-kernel composition~\cite{Griffiths,DowkerHalliwell,PazZurek,HalliwellReview}.
However, the converse can fail:
diagram commutation requires only aggregate interference
cancellation, and Appendix~\ref{app:not-consistent-histories} gives a multistep
passage that composes while violating even weak decoherence.
The full multistep path-sum derivation is given in
Appendix~\ref{app:discrete-path-integral}.

For a proposed Born-kernel semantics of finite-state quantum control
or a quantum generalization of Xenakis-style
music~\cite{Xenakis1992,Ames,Gibson2023}, vanishing of the
current is the interval-level test that the locally
supplied kernels reproduce the endpoint distribution of the specified coherent
passage.  Boundary stability imposes this test on every required interval.

The argument begins with three results.
Theorem~\ref{thm:semigroup-main} classifies universal Born-kernel composition,
Theorem~\ref{thm:qubit-paradox-main} isolates the boundary-sensitive
separation, and Theorem~\ref{thm:walsh-composition-app} gives a dense
two-qubit realization.  Section~\ref{sec:quantum-music} develops the application
to quantum music: composition identifies admissible rhythms, and
musical-transition prediction is proved \PromiseBQP-complete in the
circuit-universal score model (Definition~\ref{def:qmp-main} and
Theorem~\ref{thm:qmp-promisebqp-main}).

\section{Universal Born-kernel composition}

The standard formulation of quantum theory treats reversible transformations
as forming the unitary group $\U(d)$.  Thus every $U$ comes with an inverse $U^\dagger$, with
$U^\dagger\circ U=U\circ U^\dagger=\openone$.  If the Born images are also to form a
genealogy under the same composition law, then
\begin{equation}
 \Born(U^\dagger\circ U)=\Born(\openone)
 =\Born(U)^\top\circ\Born(U).
 \label{eq:story-inverse-composition}
\end{equation}

Equation~\eqref{eq:story-inverse-composition} holds if and only if
$U$ is a unitary monomial.  We extend this result to universal
composition.
Fix a basis of
distinguishable readout events, and let
$\Gamma\subseteq\U(d)$ be a unitary semigroup of allowed
transformations, meaning that

\[
U,V\in\Gamma\ \Longrightarrow\ V\circ U\in\Gamma.
\]

No closure under adjoints is assumed, as $\Gamma$ may or may not contain
$U^\dagger$ when it contains $U$.
The following theorem classifies universal composition.

\begin{theorem}[%
Classification of Born-consistent unitary semigroups]
\label{thm:semigroup-main}
The identity
\begin{equation}
 \Born(V\circ U)=\Born(V)\circ\Born(U)
 \quad\text{for every }U,V\in\Gamma
 \label{eq:semigroup-law}
\end{equation}
holds if and only if every element of $\Gamma$ is a unitary monomial in
the context basis.
\end{theorem}

The proof is given in
Appendix~\ref{app:semigroup-proof}.  Every such map has the form
$\Delta\circ\Pi$, where
\begin{equation}
 \Delta=\operatorname{diag}
 \left(e^{-\ii\theta_1},e^{-\ii\theta_2},\ldots,e^{-\ii\theta_d}\right),
 \qquad \theta_1,\ldots,\theta_d\in\mathbb R,
\end{equation}
is diagonal unitary and $\Pi$ is a permutation matrix.  The diagonal phases
are gauge symmetries of the Born-kernel map:
\begin{equation}
 \Born_{\C}(\Delta)=\openone,
 \qquad
 \Born_{\C}(\Delta\circ\Pi)=\Born_{\C}(\Pi)=\Pi.
 \label{eq:born-permutation-machine}
\end{equation}
\begingroup
\widowpenalty=10000
  Birkhoff's theorem places this result in its convex
setting: every doubly stochastic matrix is a convex combination of permutation
matrices~\cite{Birkhoff1946}.
At the unitary level, the conclusion becomes sharper:
exact Born-kernel composition collapses the underlying unitary family itself
to unitary monomials.
\par\endgroup

\begin{sloppypar}

The semigroup classification has an immediate finite-state
interpretation: universal Born-kernel composition selects unitary monomials,
whose Born images are reversible deterministic state machines.
The context outcomes supply the machine's finite states, and an output map
assigns observable symbols to those states.  Choosing a musical output alphabet
gives this same compositional structure a musical interpretation.  A compatible
family of their permutation kernels supplies the multitime genealogy; the full
Moore-machine corollary is stated in
Appendix~\ref{app:state-machine-corollary}.
\end{sloppypar}

\section{Composition conditions for prescribed sequences}
\label{sec:three-questions}

\begin{sloppypar}
\noindent

The classification of Born-consistent
unitary semigroups leaves a different problem open: which
prescribed pairs $U,V$ satisfy the composition identity, without requiring it
for every concatenation of those operations?  More generally, which prescribed
forward sequences give the same final outcome probabilities with
and without the intermediate context checks?
\end{sloppypar}

For a prescribed sequence $U_1,\ldots,U_k$, the test
becomes the
boundary matrix
identity
\begin{equation}
 \Born_{\C}(U_k\circ\cdots\circ U_1)
 =\Born_{\C}(U_k)\circ\cdots\circ\Born_{\C}(U_1).
 \label{eq:prescribed-composition-main}
\end{equation}
It compares endpoint statistics of a coherent passage with the statistics
obtained by recording the same fixed context at
its intermediate boundaries.  Taking the record
generally changes the quantum state; the question is whether it changes the
final distribution.  The path-sum identity, its aggregate-interference
criterion, and the formal separation of this theory from consistent histories are given in
Appendices~\ref{app:discrete-path-integral} and
\ref{app:not-consistent-histories}.

We say a Born kernel has \emph{full support} when every matrix entry is
strictly positive.  A full-support kernel is called
\emph{genuinely mixing}.  We give a genuinely mixing witness for which
this equality holds and then show that it can hold from one
initial boundary and fail
as soon as the same steps are restarted from a later one.
\paragraph{%
A genuinely mixing single-qubit Born-kernel genealogy.}
\label{sec:qubit-solutions}

Fix the computational context of one qubit and let
$\sigma_x,\sigma_y,\sigma_z$ be the Pauli
operators.  The current in Eq.~\eqref{eq:born-chapman-current} measures the
failure of a pair of unitary steps to compose over the two context outcomes.
Every column sums to zero: the current moves probability but does not create
it.  A consistent genealogy requires it to vanish at every required cut.

Write
\begin{equation}
 R_\sigma(\theta):=e^{-\ii\theta\sigma/2}.
 \label{eq:pauli-rotation}
\end{equation}
Here $\sigma\in\{\sigma_x,\sigma_y,\sigma_z\}$, and every qubit Born kernel
takes the form
\begin{equation}
 K(\lambda):=\frac12\begin{pmatrix}
 1+\lambda&1-\lambda\\
 1-\lambda&1+\lambda
 \end{pmatrix},
 \qquad -1\leq\lambda\leq1,
 \label{eq:qubit-kernel}
\end{equation}
Thus a qubit step is genuinely mixing exactly when
$-1<\lambda<1$.
These stochastic maps obey
\begin{equation}
 K(\lambda)\circ K(\mu)=K(\lambda\cdot\mu).
 \label{eq:qubit-kernel-product}
\end{equation}
For $\Born(U)=K(\mu)$ and $\Born(V)=K(\lambda)$, this product
equals $\Born(V\circ U)$ if and only if the diagram in
Fig.~\ref{fig:quantum-stochastic-composition} commutes, equivalently if and only if
$\nu_{\C}(V,U)=0$.

For the cross-axis family
$U=R_{\sigma_y}(\theta)$ and $V=R_{\sigma_x}(\phi)$,
\[
\Born(V\circ U)=K(\cos\phi\cdot\cos\theta)
=\Born(V)\circ\Born(U)
\]
for every $\theta,\phi$.  Thus
$\nu_{\C}(R_{\sigma_x}(\phi),R_{\sigma_y}(\theta))=0$ throughout the
full two-parameter family, with $\theta$ and $\phi$ independent.

We now choose from this family a genuinely mixing pair of
single-qubit unitary steps with distinct
individual Born kernels.  Put
\begin{equation}
 \begin{aligned}
 U&=R_{\sigma_y}\!\left(\frac{\pi}{3}\right)
 =\frac{\sqrt3}{2}\openone-\frac{\ii}{2}\sigma_y,\\
 V&=R_{\sigma_x}(\phi)
 =\sqrt{\frac23}\openone-\frac{\ii}{\sqrt3}\sigma_x,\qquad
 \cos\phi=\frac13.
 \end{aligned}
 \label{eq:qubit-distinct-pair}
\end{equation}
The Born kernels are
\begin{equation}
 \begin{aligned}
 \Born(U)&=K(1/2)
 =\begin{pmatrix}\frac34&\frac14\\[2pt]\frac14&\frac34\end{pmatrix},
 \\
 \Born(V)&=
 \;K(1/3)
 =\begin{pmatrix}\frac23&\frac13\\[2pt]\frac13&\frac23\end{pmatrix}.
 \end{aligned}
\label{eq:qubit-distinct-local-kernels}
\end{equation}
Nevertheless, the coherent two-step passage obeys their stochastic
composition law:
\begin{equation}
 \begin{aligned}
 \Born(V\circ U)
 &=\Born(V)\circ\Born(U)\\
 &=
 \;K(1/6)
 =\begin{pmatrix}\frac7{12}&\frac5{12}\\[2pt]\frac5{12}&\frac7{12}\end{pmatrix}.
 \end{aligned}
 \label{eq:qubit-distinct-composition}
\end{equation}
The computational context is the spectral context of
$\sigma_z$, this witness satisfies
$[U,\sigma_z]\ne0$, $[V,\sigma_z]\ne0$, and $[U,V]\ne0$, even though
$\nu_{\C}(V,U)=0$.
For each endpoint, the two contributing path amplitudes are both
nonzero but differ in phase by $\pm\pi/2$, causing the conjugate cross terms in
the Born kernel to cancel exactly.  The equality therefore reflects an
exact quantum cancellation rather than a monomial or one-path special case;
Appendix~\ref{app:qubit} gives the calculation and the continuous
$R_{\sigma_y}(\theta)$--$R_{\sigma_x}(\phi)$ family.

Beginning in $|0\rangle$, let $P_x=|x\rangle\langle x|$ and let the unread
computational-basis projective L\"uders channel be
$\Mc(\rho)=\sum_{x=0}^1P_x\rho P_x$; its outcome is discarded.
Figure~\ref{fig:qubit-two-routes} lines up the two routes at their common
boundaries.

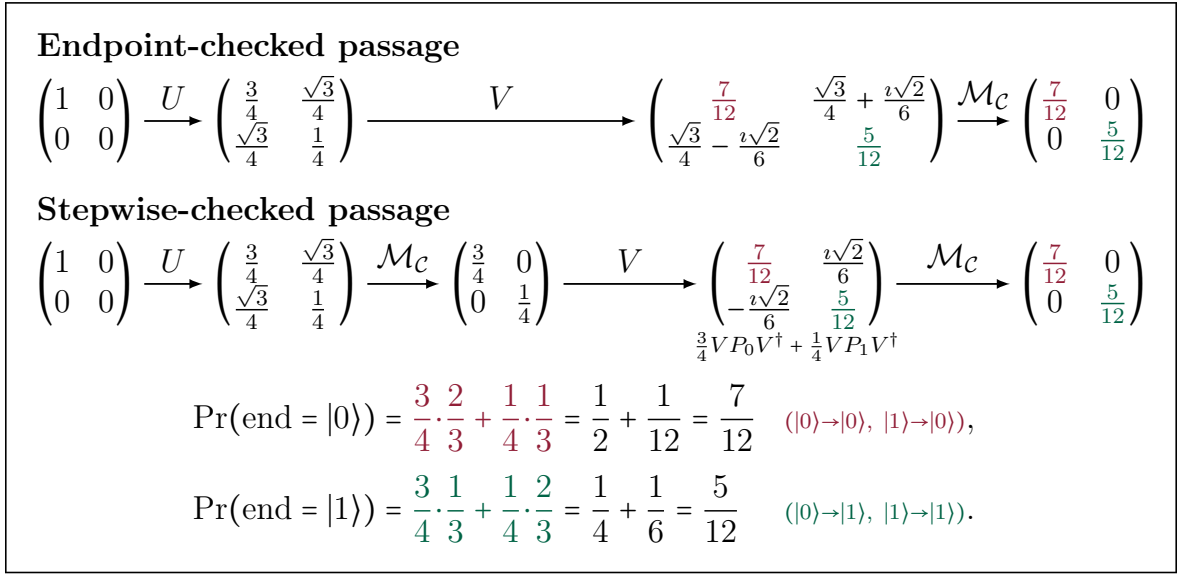
\begin{figure}[!t]
\centering
\begingroup
\setlength{\fboxsep}{7pt}
\setlength{\fboxrule}{0.6pt}
\fbox{%
\resizebox{0.94\textwidth}{!}{%
\begin{tikzpicture}[>=latex,line width=0.65pt]
  \node[anchor=west,font=\bfseries] at (-7.10,2.18)
    {Endpoint-checked passage};
  \node (u0) at (-6.35,1.25)
    {$\begin{pmatrix}1&0\\0&0\end{pmatrix}$};
  \node (u1) at (-3.80,1.25)
    {$\begin{pmatrix}
      \frac34&\frac{\sqrt3}{4}\\
      \frac{\sqrt3}{4}&\frac14
    \end{pmatrix}$};
  \node (u2) at (2.70,1.25)
    {$\begin{pmatrix}
      {\color{pathzero}\frac7{12}}&\frac{\sqrt3}{4}+\frac{\ii\sqrt2}{6}\\
      \frac{\sqrt3}{4}-\frac{\ii\sqrt2}{6}&{\color{pathone}\frac5{12}}
    \end{pmatrix}$};
  \node (u3) at (6.35,1.25)
    {$\begin{pmatrix}
      {\color{pathzero}\frac7{12}}&0\\
      0&{\color{pathone}\frac5{12}}
    \end{pmatrix}$};
  \draw[->] (u0.east) -- node[above] {$U$} (u1.west);
  \draw[->] (u1.east) -- node[above] {$V$} (u2.west);
  \draw[->] (u2.east) -- node[above] {$\Mc$} (u3.west);

  \node[anchor=west,font=\bfseries] at (-7.10,0.10)
    {Stepwise-checked passage};
  \node (c0) at (-6.35,-0.82)
    {$\begin{pmatrix}1&0\\0&0\end{pmatrix}$};
  \node (c1) at (-3.80,-0.82)
    {$\begin{pmatrix}
      \frac34&\frac{\sqrt3}{4}\\
      \frac{\sqrt3}{4}&\frac14
    \end{pmatrix}$};
  \node (cm) at (-1.05,-0.82)
    {$\begin{pmatrix}\frac34&0\\0&\frac14\end{pmatrix}$};
  \node (c2) at (2.70,-0.82)
    {$\begin{pmatrix}
      {\color{pathzero}\frac7{12}}&\frac{\ii\sqrt2}{6}\\
      -\frac{\ii\sqrt2}{6}&{\color{pathone}\frac5{12}}
    \end{pmatrix}$};
  \node[font=\scriptsize] at (2.70,-1.62)
    {$\frac34VP_0V^\dagger+\frac14VP_1V^\dagger$};
  \node (c3) at (6.35,-0.82)
    {$\begin{pmatrix}
      {\color{pathzero}\frac7{12}}&0\\
      0&{\color{pathone}\frac5{12}}
    \end{pmatrix}$};
  \draw[->] (c0.east) -- node[above] {$U$} (c1.west);
  \draw[->] (c1.east) -- node[above] {$\Mc$} (cm.west);
  \draw[->] (cm.east) -- node[above] {$V$} (c2.west);
  \draw[->] (c2.east) -- node[above] {$\Mc$} (c3.west);

  \node at (0,-3.10) {$
    \begin{aligned}
      \Pr(\mathrm{end}=|0\rangle)&=\textcolor{pathzero}{
        \frac34\!\cdot\!\frac23+\frac14\!\cdot\!\frac13}
        =\frac12+\frac1{12}=\frac7{12}
        &&\textcolor{pathzero}{\scriptstyle(|0\rangle\to|0\rangle,\ |1\rangle\to|0\rangle)},\\[5pt]
      \Pr(\mathrm{end}=|1\rangle)&=\textcolor{pathone}{
        \frac34\!\cdot\!\frac13+\frac14\!\cdot\!\frac23}
        =\frac14+\frac16=\frac5{12}
        &&\textcolor{pathone}{\scriptstyle(|0\rangle\to|1\rangle,\ |1\rangle\to|1\rangle)}.
    \end{aligned}$};
\end{tikzpicture}%
}%
}
\endgroup
\caption{\captionlead{Endpoint
agreement for a pair of single-qubit unitary steps $(U,V)$.}  The unread
intermediate projective check changes the quantum state while preserving the
final context probabilities.  The two recorded branches ending in $|0\rangle$
contribute $1/2$ and $1/12$; those ending in $|1\rangle$ contribute $1/4$
and $1/6$.  For these unequal local kernels, the two endpoint
distributions are both $(7/12,5/12)$.}
\label{fig:qubit-two-routes}
\end{figure}

The intermediate check is a non-selective rank-one projective L\"uders
measurement in the computational ($\sigma_z$) context, equivalently
context-basis dephasing~\cite{Luders}.  Its outcome is
discarded and the same qubit then evolves under $V$.  It changes the
state from pure to mixed, with the coherent and stepwise states differing by
$\sqrt3\sigma_x/4$, yet leaves the endpoint marginal $(7/12,5/12)$ unchanged.
Thus two genuinely mixing prescribed steps with distinct Born kernels carry a
genealogy despite this projective measurement disturbance.
Appendices~\ref{app:qubit} and
\ref{app:qutrit} classify all pairs of
qubit unitary steps and all pairs of qutrit unitary steps with this property;
Appendix~\ref{app:qubit} derives the qubit witness from the general criterion.

Let $R$ denote the final qubit before the context measurement.
For the recorded entrance label $X$ and independent fair checking bit $S$,
write $I_c:=I(S;Y\mid X)$ and $I_q:=I(S;R\mid X)$.
The two protocols produce different density operators for $R$, but identical
context probabilities:
\[
 0=I_c\leq I_q\simeq0.2504\ \text{bits}.
\]
Thus the quantum state retains information about the intermediate check
that the prescribed measurement conceals.
Appendix~\ref{app:schedule-information} gives the bound and calculation.

\section[Composition at the quantum--classical divide]{%
\texorpdfstring{%
Composition at the quantum--classical divide}%
{Composition at the quantum--classical divide}}
\label{sec:composition-paradox}

We provide a framework and prove that faithful composition
from the initial boundary need not survive temporal restriction and fresh
preparation at every internal boundary.

Let $U_1,\ldots,U_L$ be a prescribed sequence of qubit unitaries.
For $1\leq a\leq b\leq L$, write
\begin{equation}
 U_{b:a}:=U_b\circ U_{b-1}\circ\cdots\circ U_a.
\label{eq:boundary-indexed-passage-main}
\end{equation}
for the contiguous passage containing precisely the steps
$U_a,\ldots,U_b$, with both endpoint steps included; in particular,
$U_{k:k}=U_k$.  Number the temporal
boundaries $0,\ldots,L$, so $U_k$ carries boundary $k-1$ to boundary $k$;
consequently, $U_{b:a}$ begins at boundary $a-1$ and ends at boundary $b$.
The boundary before $U_1$, boundary $0$, is the \emph{initial
boundary}.  Beginning
at the boundary immediately before $U_a$, namely boundary $a-1$, means
freshly preparing a $\C$-diagonal entrance state there and applying a
contiguous passage $U_{b:a}$; this is an internal boundary when
$2\leq a\leq L$, and taking $b=L$ gives the full prescribed suffix.  It is a
new preparation boundary, not a pause and resumption of the
earlier quantum state.  There are two distinct consistency requirements.

\begin{definition}[Prefix consistency]
A sequence is \emph{prefix-consistent} if
\begin{equation}
 \Born(U_{k:1})=\Born(U_k)\circ\cdots\circ\Born(U_1)
 \quad\text{for every }2\le k\le L.
 \label{eq:prefix}
\end{equation}
It is \emph{boundary-stable}, equivalently
subinterval-consistent, if
\begin{equation}
 \begin{aligned}
  \Born(U_{b:a})
  &=\Born(U_b)\circ\cdots\circ\Born(U_a),\\[-1pt]
  &\hspace{1em}\text{for every }1\le a\le b\le L.
 \end{aligned}
 \label{eq:all-intervals}
\end{equation}
Here ``\textit{subinterval}'' means an unbroken
chronological passage $U_{b:a}$.  A selection that skips an intermediate
prescribed step, such as $U_3\circ U_1$,
represents a different circuit,
not an additional subinterval condition.
\end{definition}
Prefix consistency is the $a=1$ case of boundary stability: it
compares each coherent prefix with its fully stepwise-checked realization from
the original boundary.  Boundary stability makes the
fixed-context sequential statistics Kolmogorov consistent under every
insertion or deletion of an intermediate context measurement, for every
context-diagonal entrance state; see Proposition~\ref{prop:schedule-consistency}.

\begin{theorem}[Qubit boundary separation]
\label{thm:qubit-paradox-main}
For a fixed qubit context:
\begin{enumerate}
 \item prefix-consistent sequences of arbitrary finite length exist with
 every step genuinely mixing;
 \item consistency may fail on certain contiguous
 subintervals.
\end{enumerate}
\end{theorem}

Here the question is compatibility across preparation
boundaries:
even when every step is genuinely mixing, prefix consistency
does not imply boundary stability.

To isolate the distinct issue of temporal-boundary
instability, now consider the deliberately restricted
\emph{constant-amplitude infinite one-qubit sequence}
\begin{equation}
 U_1=R_{\sigma_y}\!\left(\frac{\pi}{3}\right),
 \qquad
 U_k=R_{\sigma_z}(\alpha_k)\circ
 R_{\sigma_x}\!\left(\frac{\pi}{3}\right),
 \qquad k\ge2,
 \label{eq:constant-infinity-notes}
\end{equation}
where $0<\alpha_k<\pi/2$ is fixed by
\begin{equation}
 \tan\alpha_k=\frac{\sqrt3}{2\sqrt{4^{k-1}-1}}.
 \label{eq:constant-infinity-phase}
\end{equation}
For its first three steps,
$\alpha_2=\arctan(1/2)$ and
$\alpha_3=\arctan(1/(2\sqrt5))$.  Its third prefix composes, but the
unchanged last two steps fail after an internal restart:
\begin{equation}
 \begin{aligned}
 \Born(U_3\circ U_2\circ U_1)
 &=K(1/8)=\Born(U_3)\circ\Born(U_2)\circ\Born(U_1),\\
 \Born(U_3\circ U_2)
 &=K\!\left(\frac{1-6/\sqrt5}{4}\right)
 \ne K(1/4)=\Born(U_3)\circ\Born(U_2).
 \end{aligned}
 \label{eq:qubit-third-step-main}
\end{equation}
Thus $\nu_{\C}(U_3,U_2)\ne0$.
Every step has the same genuinely mixing Born kernel,
\begin{equation}
 \Born(U_k)=K(1/2)
 =\begin{pmatrix}\frac34&\frac14\\[2pt]\frac14&\frac34\end{pmatrix},
 \label{eq:constant-infinity-local}
\end{equation}
The compensating phase changes with its position in the sequence.  For the first $\ell$-step
passage $U_{\ell:1}$, one obtains
\begin{equation}
 \Born(U_{\ell:1})=K(2^{-\ell})
 =\Born(U_\ell)\circ\cdots\circ\Born(U_1)
 \qquad\text{for every }\ell\ge1.
 \label{eq:constant-infinity-prefix}
\end{equation}
Thus the sequence is prefix-consistent at
every finite length $\ell\geq1$.

Yet restarting just before any $U_r$, $r\ge2$, fails already
over the next two steps:
\begin{equation}
 \Born(U_{r+1}\circ U_r)
 =K\!\left(\frac{1-3\cos\alpha_r}{4}\right)
 \ne K(1/4)
 =\Born(U_{r+1})\circ\Born(U_r).
 \label{eq:constant-infinity-cue-failure}
\end{equation}
These steps retain the same individual Born kernels; only the temporal
boundary has moved.  Appendix~\ref{app:qubit} proves the recurrence and its
sign-reversed
companion, whose local kernel is $K(-1/2)$.  It is not an inverse sequence.

For the first three steps, the
one-step kernels and composition identities
can be listed exhaustively.  Since $\Born(U_1)=K(1/2)$ is invertible, the

noncomposition of the suffix $U_3\circ U_2$ persists under right
composition with $\Born(U_1)$.

\begin{center}
\begingroup
\small
\renewcommand{\arraystretch}{1.55}
\begin{tabular}{p{0.46\linewidth}|p{0.46\linewidth}}
\hline
\multicolumn{2}{l}{Individual steps:
$\Born(U_1)=\Born(U_2)=\Born(U_3)=K(1/2)$.}\\
\hline
\textbf{Composition holds} & \textbf{Composition fails}\\
\hline
$\Born(U_2\circ U_1)=\Born(U_2)\circ\Born(U_1)$
& $\Born(U_3\circ U_2)\ne\Born(U_3)\circ\Born(U_2)$\\
$\Born(U_3\circ U_2\circ U_1)
 =\Born(U_3)\circ\Born(U_2\circ U_1)$
& $\Born(U_3\circ U_2\circ U_1)
 \ne\Born(U_3\circ U_2)\circ\Born(U_1)$\\
$\Born(U_3\circ U_2\circ U_1)
 =\Born(U_3)\circ\Born(U_2)\circ\Born(U_1)$
&
\end{tabular}
\endgroup
\end{center}

A check after $U_2$ and checks after both $U_1$ and $U_2$ preserve the
endpoint distribution, whereas a check only after $U_1$ changes it.  Adding
the second check restores an agreement that the first check alone destroys.
With a recorded uniform entrance and a fair bit selecting whether
to check the displayed cut, the two orphaned compositions give
$(\mathcal I(S;Y\mid X),\nu_{\C})\simeq(0.0834,0.6708)$ and
$(0.0204,0.3354)$, respectively, with information in bits per run.

The general bounds $\mathcal I(S;Y\mid X)\leq1$ bit and
$\nu_{\C}\leq1$ are proved in Appendices~\ref{app:schedule-information} and
\ref{app:bound}, respectively; the latter is sharp.

\noindent

Nothing here requires a retrocausal interpretation.  The endpoint law is
determined by the forward-ordered unitary evolution of the entire
prescribed passage,
and no later operation alters an earlier recorded event.  The issue is a
failure of faithful temporal composition: locally induced Born kernels need
not reproduce the coherent endpoint law after a fresh preparation at an
internal boundary or when prescribed readouts are inserted or removed.

The qubit construction provides the minimal-dimensional boundary-sensitive
witness.  The qutrit
analysis gives an exact condition for every
prescribed pair of qutrit unitary steps.

\begin{theorem}[Sharp full-support capacity]
\label{thm:full-support-capacity-main}
Let $d\geq2$ and let $U_1,\ldots,U_L$ be unitary operations on
$\mathbb C^d$.  If the sequence is boundary-stable, then at most $d$
one-step Born kernels $\Born(U_k)$ have full support.  Moreover, for every
prime $d$, there exists a boundary-stable $d$-step sequence in which every
one-step Born kernel has full support.
\end{theorem}
The
proofs are given in Appendices~\ref{app:qutrit} and
\ref{app:prime-count}.

\section[Quantum rhythm and transition prediction complexity]{%
\texorpdfstring{%
Quantum rhythm and transition prediction complexity}%
{Quantum rhythm and transition prediction complexity}}
\label{sec:quantum-music}

A musical composition combines a sequence of notes with its organization in
time. The state-machine connection provides a foundation for these two choices:
in the quantum model, unitary operations are notes and temporal boundaries are
cues. An output map assigns sounds to designated readout-boundary/outcome
pairs~\cite{Miranda,MirandaThomasItaborai}. Sounds are classical, but musical
rules need not be.

A prescribed operation sequence defines a score, while a readout schedule
determines when outcomes are recorded and sounded. With the current sound held
between readouts, the schedule specifies a rhythm. Born-kernel composition
determines which rhythms preserve the score's endpoint distribution.

\paragraph{Qubit rhythm laws.}
For an $L$-step score, choose a beat duration $\Delta>0$. A rhythm is specified
by an ordered tuple $\tau=(\ell_1,\ldots,\ell_m)$ of positive integers summing
to $L$, called an ordered composition of $L$. Its entries specify successive
hold durations $\ell_1\Delta,\ldots,\ell_m\Delta$, in chronological order from
the entrance cue to the endpoint. The initial sound, specified separately from
the quantum input state, is held from the entrance cue until the first readout.
It may continue from the preceding passage even when the quantum state is
freshly prepared. Put $t_0=0$ and
$t_j=\sum_{r=1}^{j}\ell_r$; the schedule reads out at boundaries
$t_1,\ldots,t_m=L$.
Its induced endpoint kernel is
\begin{equation}
 \mathsf K_\tau(U_1,\ldots,U_L)
 :=\Born_{\C}(U_{t_m:t_{m-1}+1})\circ\cdots\circ
   \Born_{\C}(U_{t_1:1}).
 \label{eq:born-rhythm-schedule-kernel-main}
\end{equation}

A rhythm $\tau$ is admissible when $\mathsf K_\tau=\Born_{\C}(U_{L:1})$.
We denote the set of admissible rhythms by $\mathscr R(U_1,\ldots,U_L)$.

\Needspace{0.62\textheight}
For the single-qubit passage in
Eqs.~\eqref{eq:constant-infinity-notes}--\eqref{eq:constant-infinity-cue-failure},
the admissible rhythms are:
\begin{center}
\resizebox{\textwidth}{!}{\begin{tikzpicture}[x=1cm,y=1cm,font=\small]
\node[anchor=west,font=\bfseries\large] at (-1.0,3.45)
  {Admissible rhythms of the qubit passage};
\node[anchor=west] at (-1.0,2.86)
  {$U_1=R_{\sigma_y}(\pi/3),\qquad
    U_k=R_{\sigma_z}(\alpha_k)\circ R_{\sigma_x}(\pi/3),\quad k=2,3,$};
\node[anchor=west] at (-1.0,2.34)
  {$\alpha_2=\arctan(1/2),\qquad
    \alpha_3=\arctan\!\bigl(1/(2\sqrt5)\bigr).$};

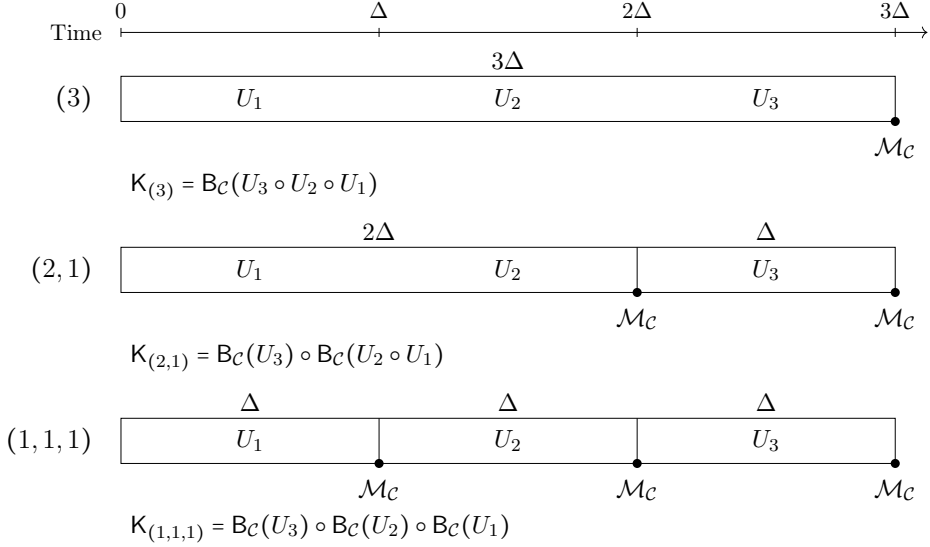
\draw[thin,->] (0,1.38)--(12.5,1.38);
\node[anchor=east,font=\footnotesize] at (-0.15,1.38) {Time};
\foreach \x/\label in {0/0,4/\Delta,8/2\Delta,12/3\Delta}{
  \draw[thin] (\x,1.31)--(\x,1.45);
  \node[above,font=\footnotesize] at (\x,1.46) {$\label$};
}

\foreach \y/\rhythm in {0/{(3)},-2.65/{(2,1)},-5.3/{(1,1,1)}}{
  \draw[thin] (0,\y) rectangle (12,{\y+0.7});
  \node[anchor=east,font=\bfseries] at (-0.3,{\y+0.35}) {$\rhythm$};
  \foreach \x/\k in {2/1,6/2,10/3}{
    \node at (\x,{\y+0.35}) {$U_{\k}$};
  }
}

\node at (6,0.95) {$3\Delta$};
\fill (12,0) circle (2pt);
\node[below=3pt] at (12,0) {$\mathcal M_{\mathcal C}$};
\node[anchor=west] at (0,-1.0)
  {$\mathsf K_{(3)}=\mathsf B_{\mathcal C}(U_3\circ U_2\circ U_1)$};

\draw[thin] (8,-2.65)--(8,-1.95);
\node at (4,-1.70) {$2\Delta$};
\node at (10,-1.70) {$\Delta$};
\foreach \x in {8,12}{
  \fill (\x,-2.65) circle (2pt);
  \node[below=3pt] at (\x,-2.65) {$\mathcal M_{\mathcal C}$};
}
\node[anchor=west] at (0,-3.65)
  {$\mathsf K_{(2,1)}=\mathsf B_{\mathcal C}(U_3)\circ
    \mathsf B_{\mathcal C}(U_2\circ U_1)$};

\foreach \x in {4,8}{\draw[thin] (\x,-5.3)--(\x,-4.6);}
\foreach \x in {2,6,10}{\node at (\x,-4.35) {$\Delta$};}
\foreach \x in {4,8,12}{
  \fill (\x,-5.3) circle (2pt);
  \node[below=3pt] at (\x,-5.3) {$\mathcal M_{\mathcal C}$};
}
\node[anchor=west] at (0,-6.3)
  {$\mathsf K_{(1,1,1)}=\mathsf B_{\mathcal C}(U_3)\circ
    \mathsf B_{\mathcal C}(U_2)\circ\mathsf B_{\mathcal C}(U_1)$};

\node[anchor=west,font=\footnotesize] at (-1.0,-7.15)
  {A solid dot marks a readout $\mathcal M_{\mathcal C}$; block widths give the hold durations.};
\node[anchor=west,font=\footnotesize] at (-1.0,-7.64)
  {All three kernels equal $\mathsf B_{\mathcal C}(U_3\circ U_2\circ U_1)$,
  meaning that all three schedules give the same final sound distribution.};
\end{tikzpicture}
}
\end{center}
Evaluating Eq.~\eqref{eq:born-rhythm-schedule-kernel-main} gives the following
kernel comparison:
\begin{equation}
 \begin{aligned}
  \text{\bfseries Endpoint-faithful schedules}\quad&\\[-2pt]
  \underbrace{\Born_{\C}(U_3\circ U_2\circ U_1)}_{(3)}
  &=
  \underbrace{\Born_{\C}(U_3)\circ
  \Born_{\C}(U_2\circ U_1)}_{(2,1)}
  \\
  &=
  \underbrace{\Born_{\C}(U_3)\circ
  \Born_{\C}(U_2)\circ
  \Born_{\C}(U_1)}_{(1,1,1)},
  \\[2pt]
  \text{\bfseries Excluded schedule}\quad&\\[-2pt]
  \Born_{\C}(U_3\circ U_2\circ U_1)
  &\ne
  \underbrace{\Born_{\C}(U_3\circ U_2)\circ
  \Born_{\C}(U_1)}_{(1,2)}.
 \end{aligned}
\label{eq:music-composition-main}
\end{equation}
Every tuple in Eq.~\eqref{eq:music-composition-main} is already in
chronological musical order; the Born factors appear in reverse textual order
because operator composition acts from right to left.
Therefore, for the first three steps of the single-qubit passage in
Eqs.~\eqref{eq:constant-infinity-notes} through
\eqref{eq:constant-infinity-cue-failure}, the set of admissible rhythms for
this sequence is
\begin{equation}
 \mathscr R(U_1,U_2,U_3)=
 \{(3),(2,1),(1,1,1)\}.
 \label{eq:qubit-born-rhythm}
\end{equation}
  This is not a restriction on
all qubit scores.  Membership in
this set is controlled by relative phases across composite blocks.
The admissible rhythm $(2,1)$ reads out after $U_2$ and then $U_3$, whereas
the excluded schedule $(1,2)$ would read out after $U_1$ and then $U_3$.
The latter leaves $U_3\circ U_2$ as a coherent block beginning at an internal
cue, where Born-kernel composition fails
[Eq.~\eqref{eq:constant-infinity-cue-failure}]. Removing the readout after
$U_2$ from $(1,1,1)$ therefore gives the inadmissible schedule $(1,2)$;
removing the remaining intermediate readout restores admissibility by giving
$(3)$. For this score, reading out after every operation preserves the endpoint
distribution of every prefix from the entrance cue. Restarting before $U_2$
means freshly preparing a context-basis input there and applying $U_2$
followed by $U_3$. A readout after $U_2$ then changes the final readout
distribution after $U_3$. This is the quantum music paradox.

To sound the single-qubit passage, assign readout outcome $0$ to the pitch
$\mathrm E_2$ and outcome $1$ to the neighboring pitch $\mathrm F_2$.
This two-pitch realization, called the \emph{Qubit Ostinato}, is described in
Appendix~\ref{app:qubit-ostinato}.
The three faithful schedules
articulate the same total span as one three-beat hold, a two-beat hold followed
by a one-beat hold, or three one-beat holds;
each enabled intermediate readout starts the assigned sound anew, even when
its pitch is unchanged.  Schedules outside $\mathscr R$ are
physically possible but fail the Born-kernel composition condition.

Born-kernel composition therefore
determines the admissible rhythms of the quantum passage.
Intermediate outcomes may be retained and sounded:
consistency concerns the terminal marginal, and marginalizing a record does
not undo the measurement that created it.  This supplies a concrete
realization of Xenakis's ``\textit{minimum of logical constraints necessary
for the construction of a musical process}''~\cite[p.~16]{Xenakis1992}.
As demonstrated below, one measurement of a selector qutrit chooses among
exactly these three faithful schedules, so the rhythmic subdivision becomes a
quantum-circuit output. Allowing arbitrary polynomial-size selector circuits
with two musical output labels gives the prediction problem in
Definition~\ref{def:qmp-main}, which is \PromiseBQP-complete by
Theorem~\ref{thm:qmp-promisebqp-main}.

\paragraph{A compositional entangling score.}
Boundary stability means that every readout schedule is admissible on every
contiguous passage of the score, for every context-basis input freshly
prepared at its entrance (Proposition~\ref{prop:schedule-consistency}).
Appendix~\ref{app:walsh} gives an analytic three-step two-qubit family
$U_1,U_2,U_3$ for which, at every temporal cut,
\begin{equation}
 \begin{aligned}
  \Born_{\C}(U_{b:a})
  &=\Born_{\C}(U_b)\circ\cdots\circ\Born_{\C}(U_a),\\[-1pt]
  &\hspace{1em}(1\leq a\leq b\leq3).
 \end{aligned}
 \label{eq:walsh-identities-main}
\end{equation}
At $\xi=\eta=1/6$ and $\zeta=0$, every entry of every step is nonzero,
all three steps are nonproduct perfect entanglers, and several nonzero paths
reach every endpoint.  Their interference terms do not vanish separately:
Eq.~\eqref{eq:walsh-identities-main} holds by aggregate cancellation.

\paragraph{A realization of the qubit-defined rhythms.}
Fix $\lambda=2^{-1/3}$, set $\theta=\arccos\lambda$, and define
\[
 \begin{aligned}
  u_1&=R_{\sigma_y}(\theta),\\
  u_k&=R_{\sigma_z}(\alpha_k)\circ R_{\sigma_x}(\theta),
       &&k=2,3,\\
  \tan\alpha_k
  &=\frac{\lambda^{k-1}\sqrt{1-\lambda^2}}
          {\sqrt{1-\lambda^{2(k-1)}}},
       &&k=2,3.
 \end{aligned}
\]
Here $0<\alpha_k<\pi/2$ fixes the phase branch. These are the first three
single-qubit operations in Corollary~\ref{cor:constant-kernel-infinity}.
Put
$A_k=u_k\otimes u_k$, let
$C_Z:=\operatorname{diag}(1,1,1,-1)$, and define the displayed two-qubit data
steps by
\[
 U_1=C_Z\circ A_1,
 \qquad U_2=A_2\circ C_Z,
 \qquad U_3=C_Z\circ A_3.
\]
The $C_Z$ factors either cancel between adjacent steps or remain as
computational-basis phases at a segment boundary, so they do not change the
relevant Born kernels.  Each $U_k$ has full amplitude support and is locally
equivalent to $C_Z$, hence is a perfect entangler.  Prepare $|00\rangle$ on
each run and use the computational-basis instrument whose non-selective action
is
$\Mc(\rho)=\sum_{x=0}^{3}|x\rangle\!\langle x|\rho|x\rangle\!\langle x|$.
The quantum state after outcome $x$ is $|x\rangle$.
The classical output map assigns a pitch using both the outcome and the
readout boundary; rows $Y_1,Y_2,Y_{\rm end}$ refer to readouts after
$U_1,U_2,U_3$, respectively:
\[
\begin{array}{c|cccc}
 &00&01&10&11\\ \hline
Y_1&\mathrm E_4&\mathrm G_4&\mathrm C_4&\mathrm C_5\\
Y_2&\mathrm G_4&\mathrm C_4&\mathrm E_4&\mathrm C_5\\
Y_{\rm end}&\mathrm C_4&\mathrm E_4&\mathrm G_4&\mathrm C_5
\end{array}
\]
For the fully checked rhythm $(1,1,1)$, the most probable pitches at the
successive readouts are $\mathrm E_4,\mathrm G_4,\mathrm C_4$.
All possible output pitches belong to the C-major triad, consisting of
C, E, and G.
\begin{equation}
 \begin{aligned}
  \Born_{\C}(U_3\circ U_2\circ U_1)
  &=\Born_{\C}(U_3)\circ
    \Born_{\C}(U_2\circ U_1)\\
  &=\Born_{\C}(U_3)\circ\Born_{\C}(U_2)
    \circ\Born_{\C}(U_1)
    =[K(1/2)]^{\otimes2},\\
  \Pr(Y_{\rm end}\mid X=00)
  &=\frac1{16}(9,3,3,1).
 \end{aligned}
 \label{eq:audible-internal-law-main}
\end{equation}

Thus the three selected readout rhythms have the same visibly biased
endpoint law.
Figure~\ref{fig:emulated-quantum-score}(c) also plots the excluded
$(1,2)$ endpoint law from $|00\rangle$, approximately
$(0.4032,0.2318,0.2318,0.1332)$.  For a recorded uniform entrance and a fair
bit selecting the coherent or $(1,2)$ protocol, the corresponding
schedule-information/current-size pair is
$(\mathcal I(S;Y_{\rm end}\mid X),
\nu_{\C}(U_3\circ U_2,U_1))\simeq(0.0223,0.2139)$, with information in bits
per run.

\Needspace{14\baselineskip}

Born-kernel composition determines the three admissible rhythms of this
passage. A separate qutrit controller selects among them, with selection
probabilities given by its own Born kernel.
Before each data
pass, apply
$V^{\rm sel}$ to the qutrit prepared in its preceding observed state and
measure it once, beginning with input $S=0$.  More generally,
$V^{\rm sel}(m)$ may be compiled from the prior audible history $m$, allowing
previously observed sounds to control the selector between passages. In the
selector basis $|0\rangle,|1\rangle,|2\rangle$, let $J_3$ be the all-ones
matrix and define the displayed selector by
\[
 \begin{gathered}
  P_3|s\rangle=|(s+1)\bmod3\rangle,\qquad
  a=\frac{-1+\ii\sqrt{55}}8,\\
  V^{\rm sel}=aP_3+\frac14(J_3-P_3).
 \end{gathered}
\]
Its Born kernel is
\begin{equation}
 \begin{gathered}
  \Born_{\C_s}(V^{\rm sel})
  =\frac1{16}\begin{pmatrix}1&1&14\\14&1&1\\1&14&1\end{pmatrix},
  \qquad \Pr\!\bigl(S_{r+1}=(S_r+1)\bmod 3\bigr)=\frac78.
 \end{gathered}
 \label{eq:score-selector-law-main}
\end{equation}
All three selected rhythms have the same endpoint kernel. Changing their
selection probabilities therefore changes the distribution of rhythms while
preserving the endpoint distribution for every context-basis input. This
remains true when the selector depends on earlier sound history, provided it
selects a rhythm before the freshly prepared data pass.

For a proposition $P$, let $\Mc^{[P]}$ denote $\Mc$ when $P$ is true
and $\openone$ otherwise.  The measured value $S\in\{0,1,2\}$ controls the data
pass
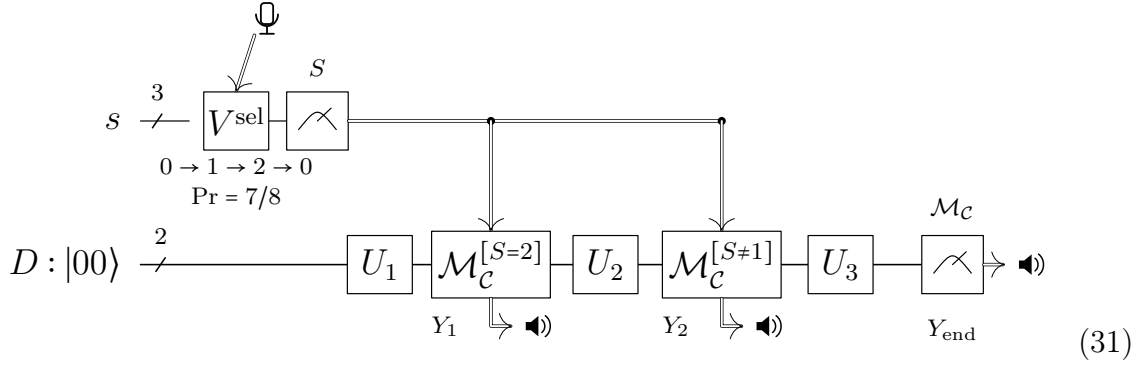
\begin{equation}
 \resizebox{0.92\linewidth}{!}{\begingroup
\color{black}
\newsavebox{\selectorMicBox}
\sbox{\selectorMicBox}{%
 \begin{tikzpicture}[x=1.65ex,y=1.65ex,baseline=-0.40ex,
   line width=0.6pt,line cap=round,line join=round]
  \draw[rounded corners=0.22ex] (-0.24,0.10) rectangle (0.24,0.92);
  \draw (-0.43,0.48) arc[start angle=180,end angle=360,
    x radius=0.43,y radius=0.46];
  \draw (0,0.02)--(0,-0.28);
  \draw (-0.27,-0.28)--(0.27,-0.28);
 \end{tikzpicture}%
}
\newsavebox{\selectorSpeakerBox}
\sbox{\selectorSpeakerBox}{%
 \begin{tikzpicture}[x=1.2ex,y=1.2ex,baseline=-0.45ex,
   line width=0.5pt,line cap=round,line join=round]
  \fill (0,0.22)--(0.28,0.22)--(0.68,0.55)--(0.68,-0.55)
    --(0.28,-0.22)--(0,-0.22)--cycle;
  \draw (0.86,0.34) to[out=-35,in=35] (0.86,-0.34);
  \draw (1.05,0.56) to[out=-35,in=35] (1.05,-0.56);
 \end{tikzpicture}%
}

\begin{tikzpicture}[
 line cap=round,line join=round,
 gate/.style={draw,fill=white,minimum width=0.76cm,minimum height=0.68cm,
              inner sep=1pt},
 readout/.style={draw,fill=white,minimum width=0.72cm,minimum height=0.68cm,
                 inner sep=0pt},
 branch/.style={draw,fill=white,minimum width=1.34cm,minimum height=0.78cm,
                align=center,inner sep=2pt}
]
  \def\selY{1.68}
  \def\dataY{0}

  \node[anchor=east] at (-0.10,\selY) {$s$};
  \draw (-0.02,\selY) -- (0.55,\selY);
  \draw[line width=0.52pt] (0.12,{\selY-0.08}) -- (0.23,{\selY+0.08});
  \node[font=\scriptsize,anchor=south] at (0.18,{\selY+0.10}) {$3$};
  \node[gate] (selector) at (1.10,\selY) {$V^{\rm sel}$};
  \node[inner sep=0pt] (historymic) at (1.46,2.88)
    {\usebox{\selectorMicBox}};
  \draw[->,double distance=0.65pt,line width=0.25pt]
    (historymic.south west) -- (selector.north);
  \draw (selector.east) -- (1.69,\selY);
  \node[readout] (sread) at (2.05,\selY) {};
  \draw ([xshift=-0.22cm,yshift=-0.10cm]sread.center)
    .. controls ([yshift=0.17cm]sread.center) ..
    ([xshift=0.22cm,yshift=-0.10cm]sread.center);
  \draw (sread.center) -- ([xshift=0.16cm,yshift=0.15cm]sread.center);
  \node[font=\scriptsize,anchor=south] at ([yshift=0.04cm]sread.north) {$S$};
  \node[font=\scriptsize,align=center] at (1.10,0.94)
    {$0\to1\to2\to0$\\$\Pr=7/8$};

  \node[anchor=east] at (-0.10,\dataY) {$D:|00\rangle$};
  \draw[line width=0.52pt] (-0.02,\dataY) -- (10.08,\dataY);
  \draw[line width=0.52pt] (0.16,-0.08) -- (0.27,0.08);
  \node[font=\scriptsize,anchor=south] at (0.22,0.10) {$2$};
  \node[gate] (u1) at (2.78,\dataY) {$U_1$};
  \node[branch] (choiceone) at (4.08,\dataY) {$\Mc^{[S=2]}$};
  \node[gate] (u2) at (5.42,\dataY) {$U_2$};
  \node[branch] (choicetwo) at (6.78,\dataY) {$\Mc^{[S\ne1]}$};
  \node[gate] (u3) at (8.17,\dataY) {$U_3$};
  \node[readout] (yread) at (9.48,\dataY) {};
  \draw ([xshift=-0.22cm,yshift=-0.10cm]yread.center)
    .. controls ([yshift=0.17cm]yread.center) ..
    ([xshift=0.22cm,yshift=-0.10cm]yread.center);
  \draw (yread.center) -- ([xshift=0.16cm,yshift=0.15cm]yread.center);
  \node[font=\scriptsize,anchor=south] at ([yshift=0.04cm]yread.north) {$\Mc$};

  \draw[double distance=0.65pt,line width=0.25pt]
    (sread.east) -- (6.78,\selY);
  \fill (4.08,\selY) circle (1.35pt);
  \fill (6.78,\selY) circle (1.35pt);
  \draw[->,double distance=0.65pt,line width=0.25pt]
    (4.08,\selY) -- (choiceone.north);
  \draw[->,double distance=0.65pt,line width=0.25pt]
    (6.78,\selY) -- (choicetwo.north);
  \draw[->,double distance=0.65pt,line width=0.25pt]
    (choiceone.south) -- (4.08,-0.72) -- (4.32,-0.72);
  \node[font=\scriptsize,anchor=east] at (3.85,-0.72) {$Y_1$};
  \node[anchor=west,inner sep=0pt] at (4.48,-0.72)
    {\usebox{\selectorSpeakerBox}};
  \draw[->,double distance=0.65pt,line width=0.25pt]
    (choicetwo.south) -- (6.78,-0.72) -- (7.02,-0.72);
  \node[font=\scriptsize,anchor=east] at (6.55,-0.72) {$Y_2$};
  \node[anchor=west,inner sep=0pt] at (7.18,-0.72)
    {\usebox{\selectorSpeakerBox}};
  \draw[->,double distance=0.65pt,line width=0.25pt]
    (yread.east) -- (10.08,\dataY);
  \node[anchor=west,inner sep=0pt] at (10.25,\dataY)
    {\usebox{\selectorSpeakerBox}};
  \node[font=\scriptsize,anchor=north,align=center] at (9.48,-0.50)
    {$Y_{\rm end}$};
\end{tikzpicture}
\endgroup
}
 \label{eq:qutrit-schedule-control-main}
\end{equation}
\begingroup
\color{black}
\definecolor{scorelatest}{HTML}{176345}
\definecolor{scoretrail}{HTML}{4C806D}
\gdef\scoreEnsembleMinEndBeats{8}
\gdef\scoreEnsembleMaxEndBeats{16}
\gdef\scoreFrameBeats{17}
\gdef\scoreFrameHalfBeats{34}
\gdef\scoreBarlineBeats{17}
\gdef\scoreBarlineHalfBeats{34}
\gdef\scoreMainEndpointBeats{15}
\gdef\scoreMainEndpointHalfBeats{30}
\gdef\scoreDisplayedEndHalfBeats{32}
\gdef\scoreDisplayedEndBeats{16}
\def\drawdisplayedscore{%
  \drawscorenote{3}{5.10}{3}{$\mathrm C_4$}%
  \drawscorenote{7}{5.10}{1}{$\mathrm C_4$}%
  \drawscorenote{9}{5.30}{1}{$\mathrm E_4$}%
  \drawscorenote{11}{5.50}{1}{$\mathrm G_4$}%
  \drawscorenote{14}{5.10}{2}{$\mathrm C_4$}%
  \drawscorenote{17}{5.50}{1}{$\mathrm G_4$}%
  \drawscorenote{21}{5.10}{3}{$\mathrm C_4$}%
  \drawscorenote{25}{5.10}{1}{$\mathrm C_4$}%
  \drawscorenote{27}{5.30}{1}{$\mathrm E_4$}%
  \drawscorenote{29}{5.50}{1}{$\mathrm G_4$}%
  \drawscoreendpoint{30}{5.10}{$\mathrm C_4$}{mandatory $Y_{\rm end}$}%
  \drawscorenote{31}{5.10}{1}{$\mathrm C_4$}%
  \drawscorerest{33}%
}
\def\drawdisplayedpath{%
  \draw[draw=scorelatest,line width=1.12pt,
      preaction={draw=white,opacity=0.90,line width=2.30pt}]
    (0,0) -- ({6-\scoreRamp},0) -- ({6+\scoreRamp},0) -- ({8-\scoreRamp},0) -- ({8+\scoreRamp},1) -- ({10-\scoreRamp},1) -- ({10+\scoreRamp},2) -- ({12-\scoreRamp},2) -- ({12+\scoreRamp},0) -- ({16-\scoreRamp},0) -- ({16+\scoreRamp},2) -- ({18-\scoreRamp},2) -- ({18+\scoreRamp},0) -- ({24-\scoreRamp},0) -- ({24+\scoreRamp},0) -- ({26-\scoreRamp},0) -- ({26+\scoreRamp},1) -- ({28-\scoreRamp},1) -- ({28+\scoreRamp},2) -- ({30-\scoreRamp},2) -- (30,0) -- (32,0);%
  \filldraw[draw=white,fill=scorelatest,line width=0.40pt] (30,0) circle (0.075);%
}
\def\drawscoreensemble{%
  \pgfmathsetmacro{\trailopacity}{0.014+0.00013*1}\drawensemble{9}{3}{6}{1}{6}{8}{2}{0}%
  \pgfmathsetmacro{\trailopacity}{0.014+0.00013*2}\drawensemble{0}{3}{3}{7}{6}{0}{3}{6}%
  \pgfmathsetmacro{\trailopacity}{0.014+0.00013*3}\drawensemble{3}{0}{3}{7}{6}{2}{9}{3}%
  \pgfmathsetmacro{\trailopacity}{0.014+0.00013*4}\drawensemble{7}{3}{11}{3}{3}{0}{4}{6}%
  \pgfmathsetmacro{\trailopacity}{0.014+0.00013*5}\drawensemble{3}{3}{6}{1}{6}{5}{0}{3}%
  \pgfmathsetmacro{\trailopacity}{0.014+0.00013*6}\drawensemble{0}{0}{3}{1}{0}{5}{0}{3}%
  \pgfmathsetmacro{\trailopacity}{0.014+0.00013*7}\drawensemble{0}{3}{6}{1}{0}{5}{6}{9}%
  \pgfmathsetmacro{\trailopacity}{0.014+0.00013*8}\drawensemble{0}{3}{6}{1}{3}{8}{0}{3}%
  \pgfmathsetmacro{\trailopacity}{0.014+0.00013*9}\drawensemble{6}{3}{6}{7}{3}{11}{0}{3}%
  \pgfmathsetmacro{\trailopacity}{0.014+0.00013*10}\drawensemble{0}{6}{0}{1}{0}{5}{9}{3}%
  \pgfmathsetmacro{\trailopacity}{0.014+0.00013*11}\drawensemble{2}{0}{3}{6}{1}{6}{2}{11}%
  \pgfmathsetmacro{\trailopacity}{0.014+0.00013*12}\drawensemble{8}{0}{0}{3}{6}{3}{6}{1}%
  \pgfmathsetmacro{\trailopacity}{0.014+0.00013*13}\drawensemble{0}{3}{6}{1}{6}{5}{0}{3}%
  \pgfmathsetmacro{\trailopacity}{0.014+0.00013*14}\drawensemble{6}{0}{9}{10}{9}{11}{0}{3}%
  \pgfmathsetmacro{\trailopacity}{0.014+0.00013*15}\drawensemble{3}{3}{6}{1}{3}{5}{0}{3}%
  \pgfmathsetmacro{\trailopacity}{0.014+0.00013*16}\drawensemble{6}{0}{3}{10}{3}{8}{6}{3}%
  \pgfmathsetmacro{\trailopacity}{0.014+0.00013*17}\drawensemble{0}{0}{9}{10}{6}{1}{6}{1}%
  \pgfmathsetmacro{\trailopacity}{0.014+0.00013*18}\drawensemble{0}{3}{6}{1}{6}{2}{3}{3}%
  \pgfmathsetmacro{\trailopacity}{0.014+0.00013*19}\drawensemble{11}{0}{0}{3}{7}{6}{2}{3}%
  \pgfmathsetmacro{\trailopacity}{0.014+0.00013*20}\drawensemble{0}{3}{3}{6}{3}{3}{8}{0}%
  \pgfmathsetmacro{\trailopacity}{0.014+0.00013*21}\drawensemble{3}{3}{6}{9}{3}{6}{1}{6}%
  \pgfmathsetmacro{\trailopacity}{0.014+0.00013*22}\drawensemble{1}{0}{5}{0}{3}{3}{7}{6}%
  \pgfmathsetmacro{\trailopacity}{0.014+0.00013*23}\drawensemble{0}{0}{6}{1}{6}{2}{0}{3}%
  \pgfmathsetmacro{\trailopacity}{0.014+0.00013*24}\drawensemble{9}{3}{6}{7}{3}{8}{0}{3}%
  \pgfmathsetmacro{\trailopacity}{0.014+0.00013*25}\drawensemble{0}{3}{3}{7}{3}{8}{0}{3}%
  \pgfmathsetmacro{\trailopacity}{0.014+0.00013*26}\drawensemble{6}{3}{6}{0}{3}{6}{0}{3}%
  \pgfmathsetmacro{\trailopacity}{0.014+0.00013*27}\drawensemble{0}{3}{6}{1}{6}{2}{0}{6}%
  \pgfmathsetmacro{\trailopacity}{0.014+0.00013*28}\drawensemble{6}{3}{6}{1}{6}{2}{0}{0}%
  \pgfmathsetmacro{\trailopacity}{0.014+0.00013*29}\drawensemble{0}{3}{6}{1}{6}{1}{6}{2}%
  \pgfmathsetmacro{\trailopacity}{0.014+0.00013*30}\drawensemble{0}{3}{6}{1}{6}{2}{0}{3}%
  \pgfmathsetmacro{\trailopacity}{0.014+0.00013*31}\drawensemble{9}{3}{3}{7}{6}{2}{3}{3}%
  \pgfmathsetmacro{\trailopacity}{0.014+0.00013*32}\drawensemble{0}{3}{6}{1}{6}{5}{3}{3}%
  \pgfmathsetmacro{\trailopacity}{0.014+0.00013*33}\drawensemble{0}{3}{6}{1}{6}{2}{6}{3}%
  \pgfmathsetmacro{\trailopacity}{0.014+0.00013*34}\drawensemble{0}{6}{6}{1}{6}{2}{0}{3}%
  \pgfmathsetmacro{\trailopacity}{0.014+0.00013*35}\drawensemble{0}{3}{0}{1}{3}{8}{0}{3}%
  \pgfmathsetmacro{\trailopacity}{0.014+0.00013*36}\drawensemble{0}{3}{3}{7}{3}{8}{3}{3}%
  \pgfmathsetmacro{\trailopacity}{0.014+0.00013*37}\drawensemble{0}{3}{0}{5}{3}{3}{6}{1}%
  \pgfmathsetmacro{\trailopacity}{0.014+0.00013*38}\drawensemble{0}{3}{0}{4}{6}{2}{0}{3}%
  \pgfmathsetmacro{\trailopacity}{0.014+0.00013*39}\drawensemble{0}{3}{6}{1}{3}{8}{6}{3}%
  \pgfmathsetmacro{\trailopacity}{0.014+0.00013*40}\drawensemble{9}{3}{6}{1}{0}{5}{10}{6}%
  \pgfmathsetmacro{\trailopacity}{0.014+0.00013*41}\drawensemble{0}{3}{6}{1}{6}{0}{6}{0}%
  \pgfmathsetmacro{\trailopacity}{0.014+0.00013*42}\drawensemble{0}{3}{3}{1}{6}{2}{3}{3}%
  \pgfmathsetmacro{\trailopacity}{0.014+0.00013*43}\drawensemble{0}{3}{6}{0}{3}{6}{1}{6}%
  \pgfmathsetmacro{\trailopacity}{0.014+0.00013*44}\drawensemble{0}{3}{6}{1}{6}{2}{3}{3}%
  \pgfmathsetmacro{\trailopacity}{0.014+0.00013*45}\drawensemble{3}{3}{6}{7}{6}{2}{0}{3}%
  \pgfmathsetmacro{\trailopacity}{0.014+0.00013*46}\drawensemble{7}{9}{11}{0}{3}{6}{1}{0}%
  \pgfmathsetmacro{\trailopacity}{0.014+0.00013*47}\drawensemble{0}{3}{6}{1}{0}{0}{3}{6}%
  \pgfmathsetmacro{\trailopacity}{0.014+0.00013*48}\drawensemble{0}{3}{6}{0}{3}{6}{1}{6}%
  \pgfmathsetmacro{\trailopacity}{0.014+0.00013*49}\drawensemble{0}{3}{6}{1}{6}{8}{3}{0}%
  \pgfmathsetmacro{\trailopacity}{0.014+0.00013*50}\drawensemble{9}{3}{6}{7}{6}{2}{0}{3}%
  \pgfmathsetmacro{\trailopacity}{0.014+0.00013*51}\drawensemble{2}{6}{3}{6}{1}{0}{2}{0}%
  \pgfmathsetmacro{\trailopacity}{0.014+0.00013*52}\drawensemble{6}{3}{6}{1}{6}{2}{0}{3}%
  \pgfmathsetmacro{\trailopacity}{0.014+0.00013*53}\drawensemble{0}{6}{0}{5}{9}{3}{6}{1}%
  \pgfmathsetmacro{\trailopacity}{0.014+0.00013*54}\drawensemble{6}{9}{9}{10}{6}{2}{0}{3}%
  \pgfmathsetmacro{\trailopacity}{0.014+0.00013*55}\drawensemble{3}{0}{3}{10}{6}{2}{8}{3}%
  \pgfmathsetmacro{\trailopacity}{0.014+0.00013*56}\drawensemble{6}{3}{6}{7}{3}{8}{3}{3}%
  \pgfmathsetmacro{\trailopacity}{0.014+0.00013*57}\drawensemble{0}{3}{6}{1}{3}{6}{3}{0}%
  \pgfmathsetmacro{\trailopacity}{0.014+0.00013*58}\drawensemble{0}{3}{0}{4}{0}{5}{0}{3}%
  \pgfmathsetmacro{\trailopacity}{0.014+0.00013*59}\drawensemble{1}{0}{5}{6}{3}{6}{1}{6}%
  \pgfmathsetmacro{\trailopacity}{0.014+0.00013*60}\drawensemble{6}{3}{6}{1}{0}{5}{0}{3}%
  \pgfmathsetmacro{\trailopacity}{0.014+0.00013*61}\drawensemble{2}{0}{3}{6}{1}{6}{2}{0}%
  \pgfmathsetmacro{\trailopacity}{0.014+0.00013*62}\drawensemble{0}{3}{6}{7}{6}{2}{0}{3}%
  \pgfmathsetmacro{\trailopacity}{0.014+0.00013*63}\drawensemble{9}{3}{6}{1}{6}{8}{6}{3}%
  \pgfmathsetmacro{\trailopacity}{0.014+0.00013*64}\drawensemble{8}{6}{3}{6}{0}{3}{6}{1}%
  \pgfmathsetmacro{\trailopacity}{0.014+0.00013*65}\drawensemble{3}{3}{3}{7}{9}{11}{3}{3}%
  \pgfmathsetmacro{\trailopacity}{0.014+0.00013*66}\drawensemble{1}{6}{2}{1}{0}{5}{0}{3}%
  \pgfmathsetmacro{\trailopacity}{0.014+0.00013*67}\drawensemble{0}{3}{6}{7}{3}{8}{6}{3}%
  \pgfmathsetmacro{\trailopacity}{0.014+0.00013*68}\drawensemble{6}{3}{6}{1}{6}{8}{0}{0}%
  \pgfmathsetmacro{\trailopacity}{0.014+0.00013*69}\drawensemble{6}{3}{0}{4}{6}{8}{3}{3}%
  \pgfmathsetmacro{\trailopacity}{0.014+0.00013*70}\drawensemble{3}{3}{6}{7}{6}{0}{6}{0}%
  \pgfmathsetmacro{\trailopacity}{0.014+0.00013*71}\drawensemble{9}{6}{0}{10}{0}{5}{3}{3}%
  \pgfmathsetmacro{\trailopacity}{0.014+0.00013*72}\drawensemble{9}{3}{6}{0}{3}{6}{4}{6}%
  \pgfmathsetmacro{\trailopacity}{0.014+0.00013*73}\drawensemble{0}{3}{6}{4}{0}{5}{8}{3}%
  \pgfmathsetmacro{\trailopacity}{0.014+0.00013*74}\drawensemble{6}{3}{6}{1}{6}{2}{6}{3}%
  \pgfmathsetmacro{\trailopacity}{0.014+0.00013*75}\drawensemble{3}{3}{6}{1}{0}{2}{0}{3}%
  \pgfmathsetmacro{\trailopacity}{0.014+0.00013*76}\drawensemble{0}{3}{6}{1}{3}{8}{0}{3}%
  \pgfmathsetmacro{\trailopacity}{0.014+0.00013*77}\drawensemble{9}{3}{0}{4}{6}{2}{0}{3}%
  \pgfmathsetmacro{\trailopacity}{0.014+0.00013*78}\drawensemble{0}{3}{6}{1}{6}{2}{0}{3}%
  \pgfmathsetmacro{\trailopacity}{0.014+0.00013*79}\drawensemble{1}{3}{8}{0}{3}{0}{4}{3}%
  \pgfmathsetmacro{\trailopacity}{0.014+0.00013*80}\drawensemble{3}{3}{0}{4}{6}{2}{6}{6}%
  \pgfmathsetmacro{\trailopacity}{0.014+0.00013*81}\drawensemble{2}{3}{3}{0}{4}{0}{5}{0}%
  \pgfmathsetmacro{\trailopacity}{0.014+0.00013*82}\drawensemble{9}{3}{3}{7}{9}{11}{0}{0}%
  \pgfmathsetmacro{\trailopacity}{0.014+0.00013*83}\drawensemble{2}{0}{3}{6}{1}{6}{2}{3}%
  \pgfmathsetmacro{\trailopacity}{0.014+0.00013*84}\drawensemble{2}{0}{3}{6}{7}{6}{2}{9}%
  \pgfmathsetmacro{\trailopacity}{0.014+0.00013*85}\drawensemble{6}{0}{3}{6}{0}{3}{7}{6}%
  \pgfmathsetmacro{\trailopacity}{0.014+0.00013*86}\drawensemble{2}{0}{3}{0}{1}{0}{5}{6}%
  \pgfmathsetmacro{\trailopacity}{0.014+0.00013*87}\drawensemble{0}{3}{6}{4}{6}{8}{3}{3}%
  \pgfmathsetmacro{\trailopacity}{0.014+0.00013*88}\drawensemble{9}{9}{9}{5}{3}{3}{6}{1}%
  \pgfmathsetmacro{\trailopacity}{0.014+0.00013*89}\drawensemble{5}{0}{6}{9}{10}{0}{11}{4}%
  \pgfmathsetmacro{\trailopacity}{0.014+0.00013*90}\drawensemble{0}{3}{3}{7}{9}{11}{3}{3}%
  \pgfmathsetmacro{\trailopacity}{0.014+0.00013*91}\drawensemble{3}{3}{6}{1}{6}{2}{0}{0}%
  \pgfmathsetmacro{\trailopacity}{0.014+0.00013*92}\drawensemble{6}{3}{6}{4}{0}{5}{6}{0}%
  \pgfmathsetmacro{\trailopacity}{0.014+0.00013*93}\drawensemble{4}{6}{2}{0}{3}{6}{7}{0}%
  \pgfmathsetmacro{\trailopacity}{0.014+0.00013*94}\drawensemble{2}{0}{3}{6}{0}{6}{6}{2}%
  \pgfmathsetmacro{\trailopacity}{0.014+0.00013*95}\drawensemble{6}{0}{3}{7}{6}{2}{1}{6}%
  \pgfmathsetmacro{\trailopacity}{0.014+0.00013*96}\drawensemble{0}{3}{6}{1}{6}{2}{0}{3}%
  \pgfmathsetmacro{\trailopacity}{0.014+0.00013*97}\drawensemble{0}{3}{6}{0}{3}{6}{1}{6}%
  \pgfmathsetmacro{\trailopacity}{0.014+0.00013*98}\drawensemble{0}{3}{3}{10}{6}{2}{0}{6}%
  \pgfmathsetmacro{\trailopacity}{0.014+0.00013*99}\drawensemble{3}{3}{6}{1}{3}{8}{0}{6}%
  \pgfmathsetmacro{\trailopacity}{0.014+0.00013*100}\drawensemble{6}{3}{6}{1}{6}{7}{6}{2}%
  \pgfmathsetmacro{\trailopacity}{0.014+0.00013*101}\drawensemble{2}{6}{3}{6}{1}{3}{8}{0}%
  \pgfmathsetmacro{\trailopacity}{0.014+0.00013*102}\drawensemble{3}{3}{3}{7}{6}{2}{6}{0}%
  \pgfmathsetmacro{\trailopacity}{0.014+0.00013*103}\drawensemble{0}{3}{6}{1}{3}{8}{1}{0}%
  \pgfmathsetmacro{\trailopacity}{0.014+0.00013*104}\drawensemble{0}{3}{6}{4}{6}{2}{0}{3}%
  \pgfmathsetmacro{\trailopacity}{0.014+0.00013*105}\drawensemble{9}{3}{6}{1}{9}{5}{2}{3}%
  \pgfmathsetmacro{\trailopacity}{0.014+0.00013*106}\drawensemble{0}{3}{6}{2}{9}{3}{6}{2}%
  \pgfmathsetmacro{\trailopacity}{0.014+0.00013*107}\drawensemble{0}{3}{6}{1}{3}{8}{3}{3}%
  \pgfmathsetmacro{\trailopacity}{0.014+0.00013*108}\drawensemble{6}{0}{3}{7}{3}{8}{0}{3}%
  \pgfmathsetmacro{\trailopacity}{0.014+0.00013*109}\drawensemble{9}{3}{3}{10}{3}{8}{2}{0}%
  \pgfmathsetmacro{\trailopacity}{0.014+0.00013*110}\drawensemble{1}{3}{5}{0}{3}{6}{0}{3}%
  \pgfmathsetmacro{\trailopacity}{0.014+0.00013*111}\drawensemble{11}{0}{3}{6}{1}{6}{2}{0}%
  \pgfmathsetmacro{\trailopacity}{0.014+0.00013*112}\drawensemble{6}{9}{9}{9}{3}{3}{7}{6}%
  \pgfmathsetmacro{\trailopacity}{0.014+0.00013*113}\drawensemble{0}{3}{6}{1}{6}{2}{0}{3}%
  \pgfmathsetmacro{\trailopacity}{0.014+0.00013*114}\drawensemble{3}{6}{0}{1}{3}{8}{6}{3}%
  \pgfmathsetmacro{\trailopacity}{0.014+0.00013*115}\drawensemble{0}{3}{6}{0}{3}{6}{1}{6}%
  \pgfmathsetmacro{\trailopacity}{0.014+0.00013*116}\drawensemble{3}{3}{6}{1}{6}{1}{0}{5}%
  \pgfmathsetmacro{\trailopacity}{0.014+0.00013*117}\drawensemble{6}{6}{6}{7}{6}{2}{0}{3}%
  \pgfmathsetmacro{\trailopacity}{0.014+0.00013*118}\drawensemble{0}{0}{6}{1}{0}{5}{6}{3}%
  \pgfmathsetmacro{\trailopacity}{0.014+0.00013*119}\drawensemble{0}{3}{0}{4}{6}{2}{3}{6}%
  \pgfmathsetmacro{\trailopacity}{0.014+0.00013*120}\drawensemble{1}{0}{11}{0}{3}{6}{1}{9}%
  \pgfmathsetmacro{\trailopacity}{0.014+0.00013*121}\drawensemble{2}{0}{3}{6}{2}{3}{3}{6}%
  \pgfmathsetmacro{\trailopacity}{0.014+0.00013*122}\drawensemble{0}{3}{6}{1}{6}{2}{0}{3}%
  \pgfmathsetmacro{\trailopacity}{0.014+0.00013*123}\drawensemble{0}{3}{6}{1}{6}{2}{3}{3}%
  \pgfmathsetmacro{\trailopacity}{0.014+0.00013*124}\drawensemble{6}{3}{6}{1}{3}{8}{11}{0}%
  \pgfmathsetmacro{\trailopacity}{0.014+0.00013*125}\drawensemble{0}{9}{9}{7}{0}{4}{6}{2}%
  \pgfmathsetmacro{\trailopacity}{0.014+0.00013*126}\drawensemble{6}{6}{0}{10}{6}{2}{3}{0}%
  \pgfmathsetmacro{\trailopacity}{0.014+0.00013*127}\drawensemble{2}{8}{3}{0}{3}{7}{0}{5}%
  \pgfmathsetmacro{\trailopacity}{0.014+0.00013*128}\drawensemble{0}{6}{0}{10}{3}{2}{6}{3}%
  \pgfmathsetmacro{\trailopacity}{0.014+0.00013*129}\drawensemble{3}{3}{6}{1}{6}{2}{0}{6}%
  \pgfmathsetmacro{\trailopacity}{0.014+0.00013*130}\drawensemble{3}{3}{6}{7}{3}{8}{0}{3}%
  \pgfmathsetmacro{\trailopacity}{0.014+0.00013*131}\drawensemble{1}{0}{5}{9}{6}{0}{4}{6}%
  \pgfmathsetmacro{\trailopacity}{0.014+0.00013*132}\drawensemble{3}{3}{6}{1}{6}{2}{7}{6}%
  \pgfmathsetmacro{\trailopacity}{0.014+0.00013*133}\drawensemble{6}{3}{6}{1}{6}{4}{6}{2}%
  \pgfmathsetmacro{\trailopacity}{0.014+0.00013*134}\drawensemble{1}{0}{5}{6}{3}{3}{7}{9}%
  \pgfmathsetmacro{\trailopacity}{0.014+0.00013*135}\drawensemble{9}{3}{6}{2}{0}{3}{6}{10}%
  \pgfmathsetmacro{\trailopacity}{0.014+0.00013*136}\drawensemble{9}{6}{0}{4}{6}{2}{2}{7}%
  \pgfmathsetmacro{\trailopacity}{0.014+0.00013*137}\drawensemble{6}{3}{3}{7}{6}{2}{0}{6}%
  \pgfmathsetmacro{\trailopacity}{0.014+0.00013*138}\drawensemble{1}{6}{2}{0}{3}{6}{4}{3}%
  \pgfmathsetmacro{\trailopacity}{0.014+0.00013*139}\drawensemble{0}{6}{0}{4}{6}{2}{6}{3}%
  \pgfmathsetmacro{\trailopacity}{0.014+0.00013*140}\drawensemble{2}{0}{3}{6}{1}{6}{2}{0}%
  \pgfmathsetmacro{\trailopacity}{0.014+0.00013*141}\drawensemble{0}{3}{6}{1}{3}{8}{0}{3}%
  \pgfmathsetmacro{\trailopacity}{0.014+0.00013*142}\drawensemble{6}{3}{0}{4}{6}{2}{4}{6}%
  \pgfmathsetmacro{\trailopacity}{0.014+0.00013*143}\drawensemble{0}{3}{6}{4}{9}{11}{0}{3}%
  \pgfmathsetmacro{\trailopacity}{0.014+0.00013*144}\drawensemble{3}{3}{6}{1}{6}{2}{3}{3}%
  \pgfmathsetmacro{\trailopacity}{0.014+0.00013*145}\drawensemble{6}{6}{0}{1}{9}{11}{3}{3}%
  \pgfmathsetmacro{\trailopacity}{0.014+0.00013*146}\drawensemble{9}{3}{3}{7}{0}{5}{6}{3}%
  \pgfmathsetmacro{\trailopacity}{0.014+0.00013*147}\drawensemble{5}{0}{3}{6}{1}{3}{8}{3}%
  \pgfmathsetmacro{\trailopacity}{0.014+0.00013*148}\drawensemble{0}{3}{6}{1}{3}{8}{0}{3}%
  \pgfmathsetmacro{\trailopacity}{0.014+0.00013*149}\drawensemble{8}{3}{0}{3}{7}{6}{2}{1}%
  \pgfmathsetmacro{\trailopacity}{0.014+0.00013*150}\drawensemble{0}{0}{3}{7}{3}{2}{1}{3}%
  \pgfmathsetmacro{\trailopacity}{0.014+0.00013*151}\drawensemble{0}{3}{6}{1}{9}{9}{3}{6}%
  \pgfmathsetmacro{\trailopacity}{0.014+0.00013*152}\drawensemble{0}{3}{0}{4}{0}{5}{3}{3}%
  \pgfmathsetmacro{\trailopacity}{0.014+0.00013*153}\drawensemble{0}{6}{0}{3}{3}{6}{4}{3}%
  \pgfmathsetmacro{\trailopacity}{0.014+0.00013*154}\drawensemble{1}{6}{0}{3}{6}{1}{3}{11}%
  \pgfmathsetmacro{\trailopacity}{0.014+0.00013*155}\drawensemble{0}{3}{6}{4}{6}{8}{0}{3}%
  \pgfmathsetmacro{\trailopacity}{0.014+0.00013*156}\drawensemble{1}{3}{8}{3}{3}{3}{7}{6}%
  \pgfmathsetmacro{\trailopacity}{0.014+0.00013*157}\drawensemble{9}{3}{3}{8}{3}{3}{6}{7}%
  \pgfmathsetmacro{\trailopacity}{0.014+0.00013*158}\drawensemble{0}{3}{6}{1}{6}{2}{0}{3}%
  \pgfmathsetmacro{\trailopacity}{0.014+0.00013*159}\drawensemble{4}{6}{2}{1}{6}{2}{9}{3}%
  \pgfmathsetmacro{\trailopacity}{0.014+0.00013*160}\drawensemble{6}{3}{6}{1}{6}{5}{5}{3}%
  \pgfmathsetmacro{\trailopacity}{0.014+0.00013*161}\drawensemble{3}{3}{3}{1}{6}{8}{0}{3}%
  \pgfmathsetmacro{\trailopacity}{0.014+0.00013*162}\drawensemble{7}{6}{2}{0}{3}{6}{3}{6}%
  \pgfmathsetmacro{\trailopacity}{0.014+0.00013*163}\drawensemble{0}{3}{6}{7}{6}{0}{3}{0}%
  \pgfmathsetmacro{\trailopacity}{0.014+0.00013*164}\drawensemble{0}{6}{0}{4}{0}{11}{0}{3}%
  \pgfmathsetmacro{\trailopacity}{0.014+0.00013*165}\drawensemble{6}{3}{6}{1}{3}{6}{3}{6}%
  \pgfmathsetmacro{\trailopacity}{0.014+0.00013*166}\drawensemble{0}{3}{6}{1}{6}{5}{0}{3}%
  \pgfmathsetmacro{\trailopacity}{0.014+0.00013*167}\drawensemble{7}{6}{2}{0}{3}{6}{1}{6}%
  \pgfmathsetmacro{\trailopacity}{0.014+0.00013*168}\drawensemble{2}{0}{3}{6}{1}{0}{5}{6}%
  \pgfmathsetmacro{\trailopacity}{0.014+0.00013*169}\drawensemble{4}{6}{1}{6}{2}{3}{6}{9}%
  \pgfmathsetmacro{\trailopacity}{0.014+0.00013*170}\drawensemble{3}{3}{6}{7}{0}{5}{2}{0}%
  \pgfmathsetmacro{\trailopacity}{0.014+0.00013*171}\drawensemble{3}{0}{3}{7}{0}{11}{6}{3}%
  \pgfmathsetmacro{\trailopacity}{0.014+0.00013*172}\drawensemble{0}{3}{6}{1}{6}{3}{3}{6}%
  \pgfmathsetmacro{\trailopacity}{0.014+0.00013*173}\drawensemble{7}{6}{5}{0}{3}{3}{7}{3}%
  \pgfmathsetmacro{\trailopacity}{0.014+0.00013*174}\drawensemble{0}{3}{3}{7}{6}{2}{0}{3}%
  \pgfmathsetmacro{\trailopacity}{0.014+0.00013*175}\drawensemble{3}{6}{0}{4}{9}{11}{0}{3}%
  \pgfmathsetmacro{\trailopacity}{0.014+0.00013*176}\drawensemble{6}{3}{6}{1}{6}{2}{0}{3}%
  \pgfmathsetmacro{\trailopacity}{0.014+0.00013*177}\drawensemble{0}{3}{6}{1}{6}{2}{6}{3}%
  \pgfmathsetmacro{\trailopacity}{0.014+0.00013*178}\drawensemble{0}{0}{3}{10}{0}{5}{1}{6}%
  \pgfmathsetmacro{\trailopacity}{0.014+0.00013*179}\drawensemble{8}{3}{3}{0}{4}{6}{2}{0}%
  \pgfmathsetmacro{\trailopacity}{0.014+0.00013*180}\drawensemble{3}{6}{0}{4}{0}{11}{0}{0}%
  \pgfmathsetmacro{\trailopacity}{0.014+0.00013*181}\drawensemble{0}{0}{3}{1}{6}{2}{3}{6}%
  \pgfmathsetmacro{\trailopacity}{0.014+0.00013*182}\drawensemble{1}{3}{8}{0}{0}{3}{6}{3}%
  \pgfmathsetmacro{\trailopacity}{0.014+0.00013*183}\drawensemble{0}{3}{6}{1}{6}{2}{0}{3}%
  \pgfmathsetmacro{\trailopacity}{0.014+0.00013*184}\drawensemble{5}{6}{3}{6}{1}{3}{1}{6}%
  \pgfmathsetmacro{\trailopacity}{0.014+0.00013*185}\drawensemble{2}{4}{6}{2}{6}{3}{6}{1}%
  \pgfmathsetmacro{\trailopacity}{0.014+0.00013*186}\drawensemble{1}{0}{5}{3}{3}{6}{7}{6}%
  \pgfmathsetmacro{\trailopacity}{0.014+0.00013*187}\drawensemble{11}{4}{0}{5}{0}{3}{6}{7}%
  \pgfmathsetmacro{\trailopacity}{0.014+0.00013*188}\drawensemble{0}{3}{6}{1}{3}{8}{0}{3}%
  \pgfmathsetmacro{\trailopacity}{0.014+0.00013*189}\drawensemble{0}{3}{6}{1}{6}{2}{2}{6}%
  \pgfmathsetmacro{\trailopacity}{0.014+0.00013*190}\drawensemble{0}{3}{6}{4}{6}{2}{0}{3}%
  \pgfmathsetmacro{\trailopacity}{0.014+0.00013*191}\drawensemble{0}{3}{6}{1}{9}{11}{6}{3}%
  \pgfmathsetmacro{\trailopacity}{0.014+0.00013*192}\drawensemble{7}{6}{2}{0}{3}{6}{4}{0}%
  \pgfmathsetmacro{\trailopacity}{0.014+0.00013*193}\drawensemble{0}{3}{6}{4}{6}{1}{3}{8}%
  \pgfmathsetmacro{\trailopacity}{0.014+0.00013*194}\drawensemble{0}{0}{3}{7}{6}{0}{3}{6}%
  \pgfmathsetmacro{\trailopacity}{0.014+0.00013*195}\drawensemble{0}{0}{3}{6}{3}{6}{4}{6}%
  \pgfmathsetmacro{\trailopacity}{0.014+0.00013*196}\drawensemble{3}{3}{6}{1}{6}{2}{2}{7}%
  \pgfmathsetmacro{\trailopacity}{0.014+0.00013*197}\drawensemble{3}{3}{6}{1}{0}{5}{0}{3}%
  \pgfmathsetmacro{\trailopacity}{0.014+0.00013*198}\drawensemble{0}{6}{0}{4}{6}{7}{0}{4}%
  \pgfmathsetmacro{\trailopacity}{0.014+0.00013*199}\drawensemble{0}{3}{6}{1}{6}{8}{0}{3}%
  \pgfmathsetmacro{\trailopacity}{0.014+0.00013*200}\drawensemble{3}{3}{6}{1}{6}{2}{0}{0}%
  \pgfmathsetmacro{\trailopacity}{0.014+0.00013*201}\drawensemble{0}{3}{6}{1}{6}{0}{0}{6}%
  \pgfmathsetmacro{\trailopacity}{0.014+0.00013*202}\drawensemble{0}{3}{6}{1}{6}{2}{0}{6}%
  \pgfmathsetmacro{\trailopacity}{0.014+0.00013*203}\drawensemble{0}{3}{6}{1}{0}{8}{0}{3}%
  \pgfmathsetmacro{\trailopacity}{0.014+0.00013*204}\drawensemble{9}{3}{6}{1}{3}{8}{6}{3}%
  \pgfmathsetmacro{\trailopacity}{0.014+0.00013*205}\drawensemble{2}{6}{0}{9}{10}{6}{2}{9}%
  \pgfmathsetmacro{\trailopacity}{0.014+0.00013*206}\drawensemble{0}{3}{6}{1}{0}{5}{9}{3}%
  \pgfmathsetmacro{\trailopacity}{0.014+0.00013*207}\drawensemble{9}{3}{6}{7}{6}{2}{0}{3}%
  \pgfmathsetmacro{\trailopacity}{0.014+0.00013*208}\drawensemble{0}{3}{6}{1}{6}{2}{3}{6}%
  \pgfmathsetmacro{\trailopacity}{0.014+0.00013*209}\drawensemble{6}{3}{3}{7}{6}{2}{9}{3}%
  \pgfmathsetmacro{\trailopacity}{0.014+0.00013*210}\drawensemble{0}{3}{6}{1}{6}{1}{6}{2}%
  \pgfmathsetmacro{\trailopacity}{0.014+0.00013*211}\drawensemble{9}{3}{6}{1}{6}{2}{4}{6}%
  \pgfmathsetmacro{\trailopacity}{0.014+0.00013*212}\drawensemble{0}{3}{6}{1}{0}{5}{0}{3}%
  \pgfmathsetmacro{\trailopacity}{0.014+0.00013*213}\drawensemble{0}{3}{6}{1}{6}{2}{6}{3}%
  \pgfmathsetmacro{\trailopacity}{0.014+0.00013*214}\drawensemble{6}{3}{6}{1}{6}{5}{6}{3}%
  \pgfmathsetmacro{\trailopacity}{0.014+0.00013*215}\drawensemble{6}{3}{6}{1}{3}{8}{0}{3}%
  \pgfmathsetmacro{\trailopacity}{0.014+0.00013*216}\drawensemble{2}{0}{3}{6}{1}{6}{2}{0}%
  \pgfmathsetmacro{\trailopacity}{0.014+0.00013*217}\drawensemble{0}{3}{6}{7}{6}{0}{3}{3}%
  \pgfmathsetmacro{\trailopacity}{0.014+0.00013*218}\drawensemble{6}{3}{6}{1}{3}{8}{9}{6}%
  \pgfmathsetmacro{\trailopacity}{0.014+0.00013*219}\drawensemble{7}{6}{2}{1}{6}{2}{5}{6}%
  \pgfmathsetmacro{\trailopacity}{0.014+0.00013*220}\drawensemble{7}{6}{2}{0}{0}{3}{7}{6}%
  \pgfmathsetmacro{\trailopacity}{0.014+0.00013*221}\drawensemble{6}{3}{6}{1}{3}{8}{0}{9}%
  \pgfmathsetmacro{\trailopacity}{0.014+0.00013*222}\drawensemble{1}{6}{8}{0}{3}{6}{1}{6}%
  \pgfmathsetmacro{\trailopacity}{0.014+0.00013*223}\drawensemble{2}{0}{9}{9}{11}{0}{3}{6}%
  \pgfmathsetmacro{\trailopacity}{0.014+0.00013*224}\drawensemble{9}{3}{6}{2}{0}{3}{6}{1}%
  \pgfmathsetmacro{\trailopacity}{0.014+0.00013*225}\drawensemble{6}{3}{6}{1}{6}{2}{2}{1}%
  \pgfmathsetmacro{\trailopacity}{0.014+0.00013*226}\drawensemble{3}{3}{6}{7}{6}{2}{1}{0}%
  \pgfmathsetmacro{\trailopacity}{0.014+0.00013*227}\drawensemble{0}{3}{0}{4}{6}{2}{0}{3}%
  \pgfmathsetmacro{\trailopacity}{0.014+0.00013*228}\drawensemble{5}{0}{3}{6}{4}{6}{2}{3}%
  \pgfmathsetmacro{\trailopacity}{0.014+0.00013*229}\drawensemble{0}{3}{6}{1}{6}{2}{2}{6}%
  \pgfmathsetmacro{\trailopacity}{0.014+0.00013*230}\drawensemble{9}{3}{6}{7}{9}{11}{6}{3}%
  \pgfmathsetmacro{\trailopacity}{0.014+0.00013*231}\drawensemble{0}{3}{6}{1}{6}{2}{9}{3}%
  \pgfmathsetmacro{\trailopacity}{0.014+0.00013*232}\drawensemble{0}{9}{9}{10}{6}{2}{3}{6}%
  \pgfmathsetmacro{\trailopacity}{0.014+0.00013*233}\drawensemble{2}{0}{3}{6}{1}{3}{8}{3}%
  \pgfmathsetmacro{\trailopacity}{0.014+0.00013*234}\drawensemble{0}{0}{3}{6}{3}{6}{1}{0}%
  \pgfmathsetmacro{\trailopacity}{0.014+0.00013*235}\drawensemble{6}{3}{3}{7}{9}{11}{6}{3}%
  \pgfmathsetmacro{\trailopacity}{0.014+0.00013*236}\drawensemble{7}{6}{2}{3}{3}{6}{1}{6}%
  \pgfmathsetmacro{\trailopacity}{0.014+0.00013*237}\drawensemble{1}{3}{8}{0}{3}{6}{1}{6}%
  \pgfmathsetmacro{\trailopacity}{0.014+0.00013*238}\drawensemble{11}{0}{3}{0}{4}{6}{2}{6}%
  \pgfmathsetmacro{\trailopacity}{0.014+0.00013*239}\drawensemble{2}{6}{6}{0}{4}{3}{8}{6}%
  \pgfmathsetmacro{\trailopacity}{0.014+0.00013*240}\drawensemble{0}{6}{0}{4}{6}{0}{3}{6}%
  \pgfmathsetmacro{\trailopacity}{0.014+0.00013*241}\drawensemble{2}{3}{3}{6}{1}{6}{2}{6}%
  \pgfmathsetmacro{\trailopacity}{0.014+0.00013*242}\drawensemble{0}{3}{0}{4}{6}{2}{0}{3}%
  \pgfmathsetmacro{\trailopacity}{0.014+0.00013*243}\drawensemble{6}{6}{9}{10}{6}{2}{9}{3}%
  \pgfmathsetmacro{\trailopacity}{0.014+0.00013*244}\drawensemble{5}{3}{6}{0}{4}{6}{8}{0}%
  \pgfmathsetmacro{\trailopacity}{0.014+0.00013*245}\drawensemble{6}{3}{6}{1}{3}{6}{6}{0}%
  \pgfmathsetmacro{\trailopacity}{0.014+0.00013*246}\drawensemble{0}{3}{6}{1}{6}{2}{5}{0}%
  \pgfmathsetmacro{\trailopacity}{0.014+0.00013*247}\drawensemble{0}{3}{6}{1}{6}{2}{0}{3}%
  \pgfmathsetmacro{\trailopacity}{0.014+0.00013*248}\drawensemble{1}{6}{2}{3}{3}{6}{1}{0}%
  \pgfmathsetmacro{\trailopacity}{0.014+0.00013*249}\drawensemble{11}{1}{6}{1}{6}{2}{0}{3}%
  \pgfmathsetmacro{\trailopacity}{0.014+0.00013*250}\drawensemble{0}{3}{6}{10}{6}{0}{3}{3}%
  \pgfmathsetmacro{\trailopacity}{0.014+0.00013*251}\drawensemble{1}{9}{10}{3}{8}{3}{0}{3}%
  \pgfmathsetmacro{\trailopacity}{0.014+0.00013*252}\drawensemble{2}{0}{3}{6}{1}{6}{2}{3}%
  \pgfmathsetmacro{\trailopacity}{0.014+0.00013*253}\drawensemble{0}{3}{6}{1}{3}{8}{6}{3}%
  \pgfmathsetmacro{\trailopacity}{0.014+0.00013*254}\drawensemble{2}{0}{3}{6}{1}{6}{8}{3}%
  \pgfmathsetmacro{\trailopacity}{0.014+0.00013*255}\drawensemble{0}{0}{3}{7}{6}{8}{6}{3}%
  \pgfmathsetmacro{\trailopacity}{0.014+0.00013*256}\drawensemble{2}{0}{3}{6}{1}{6}{2}{6}%
  \ifcsname drawscorebundle\endcsname
  \drawscorebundle{0}{0}{143}%
  \drawscorebundle{0}{1}{36}%
  \drawscorebundle{0}{2}{50}%
  \drawscorebundle{1}{0}{51}%
  \drawscorebundle{1}{1}{30}%
  \drawscorebundle{2}{0}{59}%
  \drawscorebundle{2}{1}{156}%
  \drawscorebundle{2}{2}{31}%
  \drawscorebundle{4}{0}{56}%
  \drawscorebundle{4}{1}{45}%
  \drawscorebundle{4}{2}{140}%
  \drawscorebundle{6}{0}{139}%
  \drawscorebundle{6}{1}{43}%
  \drawscorebundle{6}{2}{54}%
  \drawscorebundle{7}{0}{112}%
  \drawscorebundle{7}{1}{41}%
  \drawscorebundle{7}{2}{45}%
  \drawscorebundle{8}{0}{115}%
  \drawscorebundle{8}{1}{75}%
  \drawscorebundle{8}{2}{47}%
  \drawscorebundle{10}{0}{45}%
  \drawscorebundle{10}{1}{50}%
  \drawscorebundle{10}{2}{143}%
  \drawscorebundle{12}{0}{136}%
  \drawscorebundle{12}{1}{47}%
  \drawscorebundle{12}{2}{52}%
  \drawscorebundle{13}{0}{109}%
  \drawscorebundle{13}{1}{46}%
  \drawscorebundle{13}{2}{47}%
  \drawscorebundle{14}{0}{110}%
  \drawscorebundle{14}{1}{78}%
  \drawscorebundle{14}{2}{51}%
  \drawscorebundle{15}{0}{79}%
  \drawscorebundle{15}{1}{35}%
  \drawscorebundle{15}{2}{46}%
  \drawscorebundle{16}{0}{90}%
  \drawscorebundle{16}{1}{43}%
  \drawscorebundle{16}{2}{107}%
  \drawscorebundle{17}{0}{49}%
  \drawscorebundle{17}{2}{30}%
  \drawscorebundle{18}{0}{129}%
  \drawscorebundle{18}{1}{42}%
  \drawscorebundle{18}{2}{46}%
  \drawscorebundle{19}{0}{64}%
  \drawscorebundle{19}{1}{33}%
  \drawscorebundle{20}{0}{68}%
  \drawscorebundle{20}{1}{120}%
  \drawscorebundle{20}{2}{33}%
  \drawscorebundle{21}{0}{57}%
  \drawscorebundle{21}{1}{111}%
  \drawscorebundle{22}{0}{43}%
  \drawscorebundle{22}{2}{31}%
  \fi
}

\begin{figure}[p]
\centering
\begin{minipage}[t]{0.61\linewidth}
\centering
\textbf{(a)}\quad Rhythm selection

\medskip
{\small
\renewcommand{\arraystretch}{1.22}
\begin{tabular}{ccl}
\toprule
$S$ & Rhythm & Readouts after\\
\midrule
0 & $(2,1)$ & $U_2$, then $U_3$\\
1 & $(3)$ & $U_3$\\
2 & $(1,1,1)$ & $U_1$, then $U_2$, then $U_3$\\
\bottomrule
\end{tabular}
}

\medskip
{\small $0\longrightarrow1\longrightarrow2\longrightarrow0$\\
Each forward selector transition has probability $7/8$.}
\end{minipage}
\hfill
\begin{minipage}[t]{0.35\linewidth}
\centering
\textbf{(c)}\\[-0.10em]
Endpoint-law comparison

\vspace{0.15em}
\resizebox{0.98\linewidth}{!}{%
\begin{tikzpicture}[x=0.86cm,y=0.52cm,line cap=round,line join=round]
  \node[font=\tiny,anchor=east] at (0.36,3.55) {$\Pr$};
  \draw[black!45,line width=0.38pt] (0.45,3.55) -- (4.00,3.55);
  \foreach \pp/\lab in {0/0,0.2/20,0.4/40,0.6/60,0.8/80,1/100}
    {\pgfmathsetmacro{\xx}{0.45+3.55*\pp}
     \draw[black!45,line width=0.32pt] (\xx,3.42) -- (\xx,3.68);
     \node[font=\tiny,anchor=south] at (\xx,3.68) {\lab};}
  \node[font=\tiny,anchor=west] at (4.08,3.68) {\%};
  \foreach \yy/\lab in {2.75/$\mathrm C_4$,1.90/$\mathrm E_4$,
                           1.05/$\mathrm G_4$,0.20/$\mathrm C_5$}
    {\draw[black!14,line width=0.28pt] (0.45,\yy) -- (4.00,\yy);
     \node[font=\scriptsize,anchor=east] at (0.34,\yy) {\lab};}

  \def\drawendpointcomparison#1#2#3{%
    \pgfmathsetmacro{\xfaithful}{0.45+3.55*(#1)}%
    \pgfmathsetmacro{\xexcluded}{0.45+3.55*(#2)}%
    \draw[black!28,line width=0.52pt]
      (\xfaithful,#3) -- (\xexcluded,#3);
    \draw[black!68,line width=0.62pt]
      (\xexcluded,#3) ++(-2.35pt,-2.35pt) -- ++(4.70pt,4.70pt)
      (\xexcluded,#3) ++(-2.35pt,2.35pt) -- ++(4.70pt,-4.70pt);
    \draw[black!62,densely dashed,line width=0.60pt]
      (\xfaithful,#3) circle (3.20pt);
    \fill[scorelatest] (\xfaithful,#3) circle (1.70pt);%
  }
  \drawendpointcomparison{0.5625}{0.403195287768852}{2.75}
  \drawendpointcomparison{0.1875}{0.231781316342553}{1.90}
  \drawendpointcomparison{0.1875}{0.231781316342553}{1.05}
  \drawendpointcomparison{0.0625}{0.133242079546042}{0.20}

  \draw[black!62,densely dashed,line width=0.60pt]
    (0.52,-0.42) circle (3.20pt);
  \fill[scorelatest] (0.52,-0.42) circle (1.70pt);
  \node[font=\tiny,anchor=west] at (0.70,-0.42)
    {common: $(3),(2,1),(1,1,1)$};
  \draw[black!68,line width=0.62pt]
    (0.52,-0.92) ++(-2.35pt,-2.35pt) -- ++(4.70pt,4.70pt)
    (0.52,-0.92) ++(-2.35pt,2.35pt) -- ++(4.70pt,-4.70pt);
  \node[font=\tiny,anchor=west] at (0.70,-0.92)
    {excluded: $(1,2)$};
\end{tikzpicture}%
}
\end{minipage}

\vspace{0.05em}
\begin{minipage}[t]{0.985\linewidth}
\centering
\textbf{(b)}\quad
Five complete passes, final sound, and one silent beat

{\scriptsize filled: one beat\qquad open: two beats\qquad
dotted open: three beats}

\vspace{0.10em}
\resizebox{0.995\linewidth}{!}{%
\begin{tikzpicture}[x=0.34cm,y=0.73cm,line cap=round,line join=round]
  \font\scoreclef=musix20 at 25pt
  \foreach \j in {0,...,4}
    {\draw[black!60,line width=0.30pt]
       (0,{5.30+0.20*\j}) -- (\scoreFrameHalfBeats,{5.30+0.20*\j});}
  \fill[black,opacity=0.035]
    (\scoreDisplayedEndHalfBeats,5.20) rectangle (\scoreBarlineHalfBeats,6.18);
  \node at (-1.55,5.70) {{\scoreclef\char71}};
  \draw[black!60,line width=0.35pt] (0,5.30) -- (0,6.10);
  \draw[black!68,line width=0.55pt]
    (\scoreBarlineHalfBeats,5.30) -- (\scoreBarlineHalfBeats,6.10);
  \draw[black!60,double distance=1.1pt,line width=0.30pt]
    (\scoreFrameHalfBeats,5.30) -- (\scoreFrameHalfBeats,6.10);
  \def\drawscorenote#1#2#3#4{%
    \ifnum#3=1
      \draw[fill=black,rotate around={-20:(#1,#2)}]
        (#1,#2) ellipse (0.34 and 0.10);
    \else
      \draw[fill=white,rotate around={-20:(#1,#2)}]
        (#1,#2) ellipse (0.34 and 0.10);
    \fi
    \ifdim #2pt>5.65pt
      \draw[line width=0.48pt] ({#1-0.30},#2) -- ({#1-0.30},{#2-0.75});
    \else
      \draw[line width=0.48pt] ({#1+0.30},#2) -- ({#1+0.30},{#2+0.75});
    \fi
    \ifnum#3=3
      \ifdim #2pt<5.70pt
        \fill ({#1+0.58},{#2+0.10}) circle (0.075);
      \else
        \fill ({#1+0.58},#2) circle (0.075);
      \fi
    \fi
    \ifdim #2pt<5.20pt
      \draw[line width=0.35pt] ({#1-0.60},#2) -- ({#1+0.60},#2);
    \fi
    \node[font=\tiny] at (#1,4.30) {#4};%
  }
  \def\drawscoreendpoint#1#2#3#4{%
    \draw[scorelatest,line width=0.62pt] (#1,#2) circle (0.17);
    \fill[scorelatest] (#1,#2) circle (0.075);
    \node[font=\tiny,text=black!72,anchor=north] at (#1,4.90) {#4};
  }
  \def\drawscorerest#1{%
    \draw[black!62,line width=0.55pt]
      ({#1-0.20},5.88) -- ({#1+0.12},5.65)
      -- ({#1-0.10},5.42) -- ({#1+0.20},5.28);
  }

  \foreach \yy/\lab in {0/$\mathrm C_4$,1/$\mathrm E_4$,
                           2/$\mathrm G_4$,3/$\mathrm C_5$}
    {\draw[black!18,line width=0.28pt] (0,\yy) -- (\scoreFrameHalfBeats,\yy);
     \node[font=\scriptsize,anchor=east] at (-0.24,\yy) {\lab};}
  \foreach \xx in {0,...,\scoreFrameHalfBeats}
    {\draw[black!7,line width=0.20pt] (\xx,-0.12) -- (\xx,3.12);}
  \foreach \xx in {0,2,...,\scoreFrameHalfBeats}
    {\draw[black!13,line width=0.25pt] (\xx,-0.12) -- (\xx,3.12);}
  \foreach \xx/\beat in {0/0,8/4,16/8,24/12}
    {\draw[black!22,line width=0.32pt] (\xx,-0.12) -- (\xx,3.12);
     \node[font=\tiny,anchor=north] at (\xx,-0.17) {\beat};}
  \draw[black!22,line width=0.32pt] (30,-0.12) -- (30,3.12);
  \node[font=\tiny,anchor=north east] at (30,-0.17) {15};
  \draw[black!22,line width=0.32pt] (32,-0.12) -- (32,3.12);
  \node[font=\tiny,anchor=north east] at (32,-0.17) {16};
  \node[font=\tiny,anchor=north east] at (34,-0.17) {17};
  \draw[black!38,line width=0.55pt]
    (\scoreBarlineHalfBeats,-0.12) -- (\scoreBarlineHalfBeats,3.12);
  \node[font=\scriptsize,anchor=north]
    at ({0.5*\scoreFrameHalfBeats},-0.62)
    {final readout: 15; sound ends: 16; silence: 16--17};

  \def\scoreRamp{0.40}
  \def\trailopacity{0.07}
  \tikzset{ensemble/.style={draw=black!60,opacity=\trailopacity,line width=0.32pt}}
  \def\drawensemble#1#2#3#4#5#6#7#8{%
    \pgfmathtruncatemacro{\pA}{int(#1/3)}
    \pgfmathtruncatemacro{\pB}{int(#2/3)}
    \pgfmathtruncatemacro{\pC}{int(#3/3)}
    \pgfmathtruncatemacro{\pD}{int(#4/3)}
    \pgfmathtruncatemacro{\pE}{int(#5/3)}
    \pgfmathtruncatemacro{\pF}{int(#6/3)}
    \pgfmathtruncatemacro{\pG}{int(#7/3)}
    \pgfmathtruncatemacro{\pH}{int(#8/3)}
    \pgfmathtruncatemacro{\nA}{2*(mod(#1,3)+1)}
    \pgfmathtruncatemacro{\nB}{2*(mod(#2,3)+1)}
    \pgfmathtruncatemacro{\nC}{2*(mod(#3,3)+1)}
    \pgfmathtruncatemacro{\nD}{2*(mod(#4,3)+1)}
    \pgfmathtruncatemacro{\nE}{2*(mod(#5,3)+1)}
    \pgfmathtruncatemacro{\nF}{2*(mod(#6,3)+1)}
    \pgfmathtruncatemacro{\nG}{2*(mod(#7,3)+1)}
    \pgfmathtruncatemacro{\nH}{2*(mod(#8,3)+1)}
    \pgfmathtruncatemacro{\hA}{0}
    \pgfmathtruncatemacro{\hB}{\hA+\nA}
    \pgfmathtruncatemacro{\hC}{\hB+\nB}
    \pgfmathtruncatemacro{\hD}{\hC+\nC}
    \pgfmathtruncatemacro{\hE}{\hD+\nD}
    \pgfmathtruncatemacro{\hF}{\hE+\nE}
    \pgfmathtruncatemacro{\hG}{\hF+\nF}
    \pgfmathtruncatemacro{\hH}{\hG+\nG}
    \pgfmathtruncatemacro{\hI}{\hH+\nH}
    \draw[ensemble] (\hA,\pA)
      -- ({\hB-\scoreRamp},\pA) -- ({\hB+\scoreRamp},\pB)
      -- ({\hC-\scoreRamp},\pB) -- ({\hC+\scoreRamp},\pC)
      -- ({\hD-\scoreRamp},\pC) -- ({\hD+\scoreRamp},\pD)
      -- ({\hE-\scoreRamp},\pD) -- ({\hE+\scoreRamp},\pE)
      -- ({\hF-\scoreRamp},\pE) -- ({\hF+\scoreRamp},\pF)
      -- ({\hG-\scoreRamp},\pF) -- ({\hG+\scoreRamp},\pG)
      -- ({\hH-\scoreRamp},\pG) -- ({\hH+\scoreRamp},\pH)
      -- (\hI,\pH);%
  }
  \def\drawscorebundle#1#2#3{%
    \pgfmathsetmacro{\bundleweight}{max(0,min(1,(#3-30)/80))}%
    \pgfmathsetmacro{\bundleopacity}{0.035+0.105*\bundleweight}%
    \pgfmathsetmacro{\bundlewidth}{0.30+0.36*\bundleweight}%
    \draw[draw=black!55,opacity=\bundleopacity,line width=\bundlewidth pt]
      (#1,#2) -- ({#1+1},#2);%
  }
  \drawscoreensemble
  \drawdisplayedscore
  \drawdisplayedpath
  \foreach \xx in {0,6,12,18,24,30,32,34}
    {\draw[black,opacity=0.16,line width=0.28pt]
      (\xx,5.20) -- (\xx,6.28);}
  \foreach \xa/\xb in {0/6,6/12,12/18,18/24,24/30,30/32,32/34}
    {\draw[black,opacity=0.38,line width=0.35pt]
      (\xa,6.28) -- (\xb,6.28);}
  \node[font=\tiny,text=black!72] at (3,6.48) {$S=1:\ (3)$};
  \node[font=\tiny,text=black!72] at (9,6.48) {$S=2:\ (1,1,1)$};
  \node[font=\tiny,text=black!72] at (15,6.48) {$S=0:\ (2,1)$};
  \node[font=\tiny,text=black!72] at (21,6.48) {$S=1:\ (3)$};
  \node[font=\tiny,text=black!72] at (27,6.48) {$S=2:\ (1,1,1)$};
  \node[font=\tiny,text=black!58] at (31,6.76) {hold};
  \node[font=\tiny,text=black!58] at (33,6.76) {rest};
  \node[font=\scriptsize,anchor=west,text=black!55] at (0,3.40)
    {256 seeded reruns; span: $\scoreEnsembleMinEndBeats$--$\scoreEnsembleMaxEndBeats$ beats};
  \node[font=\scriptsize,anchor=east,text=black!55]
    at (\scoreFrameHalfBeats,3.40)
    {final hold $\mid$ silence};
\end{tikzpicture}%
}
\vspace{0em}
\textbf{(d)}\quad Circuit-sampled score realizations
\end{minipage}

\caption{\small\captionlead{Quantum rhythm selection and its sounded realizations.}
\textbf{(a)} The qutrit outcome in the main-text circuit,
Eq.~\eqref{eq:qutrit-schedule-control-main}, selects one of the three
admissible rhythms of the same $U_1,U_2,U_3$ passage. Every rhythm includes
the final readout; the excluded schedule $(1,2)$ is never selected.
\textbf{(b)} The seed-1766 performance begins at an observed endpoint.
Five complete passes occupy beats 0--15, with selector outcomes $1,2,0,1,2$.
Each pass begins in $|00\rangle$ while the preceding sound is held until the
next readout. Filled, open, and dotted-open noteheads denote one-, two-, and
three-beat holds, centered over their intervals. The circle at beat 15 marks
the mandatory final readout; its $\mathrm C_4$ is held until beat 16,
followed by one silent beat to beat 17. This terminal hold adds no quantum
operation or readout.
\textbf{(c)} Green points inside dashed circles show the common endpoint law
$(9,3,3,1)/16$ from $|00\rangle$ for $(3)$, $(2,1)$, and $(1,1,1)$.
Gray crosses show the excluded $(1,2)$ law,
approximately $(0.4032,0.2318,0.2318,0.1332)$; horizontal segments display
the probability displacement.
\textbf{(d)} Pale gray paths show 256 seeded simulator reruns, each retaining
eight readouts and ending at its actual terminal beat. Their spans range
from $\scoreEnsembleMinEndBeats$ to $\scoreEnsembleMaxEndBeats$ beats.
The highlighted performance has its final readout at 15, holds its final
sound through 16, and is silent until 17. Ramps soften only the visual joins
at readout boundaries; hold durations remain unchanged. Shared same-pitch
background edges darken with their rerun count. The quantum measurement
outcomes determine subsequent state preparation; the boundary-dependent
pitch map assigns the classical sounds.}
\label{fig:emulated-quantum-score}
\end{figure}
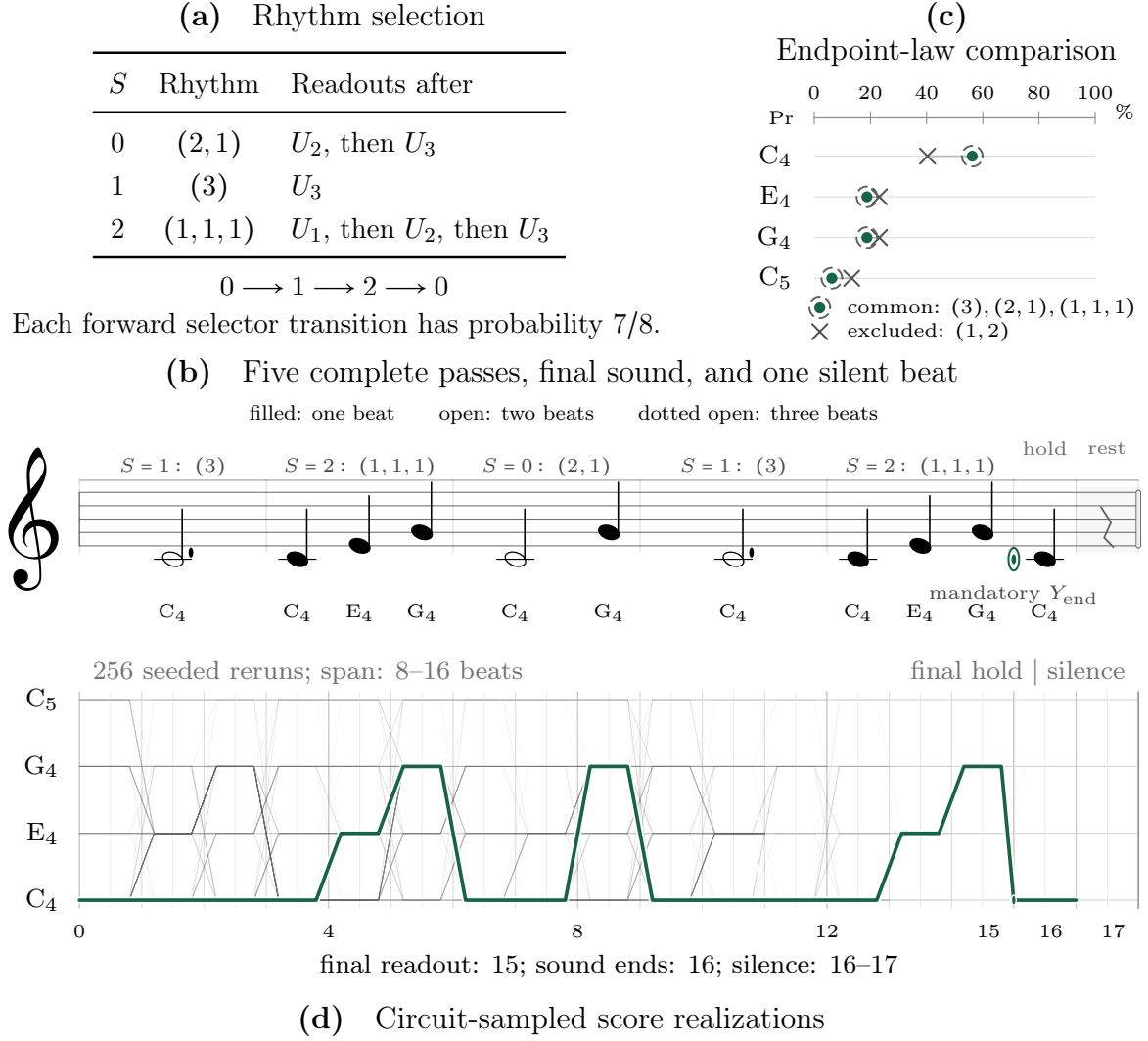
\endgroup

The upper wire carries the selector qutrit; after $V^{\rm sel}$, its
measurement supplies the classical value $S$. The double lines carry this
value to the two conditional readouts on the lower data wire. The first is
enabled only for $S=2$, the second for $S=0$ or $2$, and the endpoint is
always read out. Thus $S=0,1,2$ select $(2,1)$, $(3)$, and $(1,1,1)$,
respectively; the failing $(1,2)$ mask is never selected. The microphone
marks the optional prior sound history controlling $V^{\rm sel}$ between
passes, and the speakers mark sounds assigned to readout outcomes.
Every enabled readout starts its assigned sound anew, including when the
pitch repeats. Every enabled intermediate readout retains its classical
outcome, and the endpoint outcome
$Y_{\rm end}$ is always sounded.

\begin{samepage}
With adjacent circuit boundaries separated
by $\Delta$, the current sound is held until the next enabled readout,
producing one-, two-, or three-$\Delta$ holds.  The four pitches and three
durations therefore define twelve pitch--duration symbols.
\par
\end{samepage}

The displayed seed-1766 window in Figure~\ref{fig:emulated-quantum-score}
begins at an observed endpoint. The entrance sound carries over from that
endpoint while the data state is freshly prepared in $|00\rangle$ for each
pass. Five successive selector outcomes are \(1,2,0,1,2\), and the five
complete circuit passes occupy 15 beats. The mandatory final readout
\(Y_{\rm end}\) at \(t=15\Delta\) closes the fifth pass. Its assigned pitch
$\mathrm C_4$ is held until \(t=16\Delta\), followed by one silent beat ending
at \(t=17\Delta\). This terminal hold is a playback convention after the
measured passage and adds no quantum operation or readout.
The figure also shows 256 separate eight-readout comparison excerpts; their
sampled outcomes and durations are independent of this displayed terminal hold.

In the Walsh family, every readout schedule is admissible on every contiguous
passage. For the displayed two-qubit score, precisely $(3)$, $(2,1)$, and
$(1,1,1)$ are admissible on the full three-step passage, whereas $(1,2)$ fails
the Born-kernel composition condition.
A score is cue-stable when the Born kernel of every independently initialized
contiguous passage equals the ordered product of its step kernels, for every
context-basis entrance.
By Theorem~\ref{thm:full-support-capacity-main}, a cue-stable
$d$-outcome score contains at most $d$ full-support note
operations,
with equality attainable in every prime dimension.

\Needspace{14\baselineskip}
\paragraph{Transition prediction.}
Given a circuit that generates the next musical output, the prediction
question asks whether the probability of a designated output lies above an
upper threshold or below a lower threshold, promised that one case holds.
It concerns the output distribution across repeated preparations.

The computational model now ranges over arbitrary quantum circuits whose
numbers of gates and qubits grow polynomially with the input length. The
fixed qubit examples and the $3\times3$ qutrit selector illustrate the
musical mechanisms. The completeness theorem concerns this growing circuit
family, without imposing universal Born-kernel composition. Universally
composing reversible state machines and tests of boundary stability retain
the separate scopes established earlier.

\begin{definition}[Quantum musical-transition prediction
  (\textup{\textsc{QMTP}})]
\label{def:qmp-main}
Fix an integer $c\geq1$. An instance consists of an explicitly encoded
polynomial-size quantum circuit $Q$,
initialized in $|0^n\rangle$, with a distinguished output qubit whose outcomes
carry two fixed musical labels, and rational thresholds $0\leq a<b\leq1$.
Let $p_Q$ be the probability of the designated label and let
$\ell:=|\langle Q,a,b\rangle|$ be the total input bit length.
The thresholds have inverse-polynomial separation $b-a\geq\ell^{-c}$.
Decide whether
\begin{equation}
 p_Q\geq b
 \qquad\text{or}\qquad
 p_Q\leq a,
 \label{eq:qmp-promise-main}
\end{equation}
promised that one alternative holds.
\end{definition}
Replacing the two musical labels with ordinary machine output symbols leaves
the probabilities, promise gap, and prediction problem unchanged.

\Needspace{5\baselineskip}
\begin{theorem}[\PromiseBQP-completeness]
\label{thm:qmp-promisebqp-main}
\begin{sloppypar}

For every fixed $c\geq1$, quantum musical-transition prediction
(\textup{\textsc{QMTP}}) is \PromiseBQP-complete under deterministic
classical polynomial-time many-one reductions.
\end{sloppypar}
\end{theorem}

This completeness is inherited from standard circuit-output
prediction: sampling relative to the thresholds' midpoint gives membership,
while hardness maps a verifier's acceptance bit to the designated musical label
~\cite{JanzingWocjan2007,WatrousComplexity}. The complete gate-set conventions,
encoding, and reduction are given in
Appendix~\ref{app:quantum-musical-prediction}.

The double-Hadamard passage is a workhorse of consistent-histories
analysis~\cite{Griffiths,DowkerHalliwell,PazZurek,HalliwellReview}.  In the
present fixed-context formulation it maximally violates Born-kernel
composition, attaining the maximal normalized Born--Chapman--Kolmogorov
current $\nu_{\C}(H,H)=1$, and thereby connects this prediction problem
directly to readout timing.
Prepare $|0\rangle$,
apply $H$ followed by $H$, and always read out the endpoint. A separate
circuit enables the intermediate readout when its output is $S=1$, with
probability $p$. With $Y$ the final outcome,
\[
 \Pr(Y=1\mid S=0)=0,\qquad
 \Pr(Y=1\mid S=1)=\frac12,\qquad
 \Pr(Y=1)=\frac p2.
\]
Without the intermediate readout, $H^2=\openone$ returns the prepared state;
with it, the final outcome is uniform. The readout choice selects the rhythms
$(2)$ and $(1,1)$, whose endpoint kernels differ for this passage. Assigning
distinct sounds to $Y=0,1$ makes the endpoint probability an audible-output
probability. The input thresholds $a,b$ become $a/2,b/2$, so their
separation remains inverse polynomial. The resulting promise problem is
\PromiseBQP-complete
by Corollary~\ref{cor:qmp-endpoint-check-app}. Here the endpoint distribution
depends on rhythm selection; the three admissible rhythms in
Figure~\ref{fig:emulated-quantum-score} share one endpoint distribution.

If quantum musical-transition prediction were efficiently
classically predictable, then
$\mathrm{BQP}=\mathrm{BPP}$.

\Needspace{12\baselineskip}
\paragraph{Experimental outlook.}
The normalized magnitude $\nu_{\C}(V,U)$ of the composition defect is
operationally estimable.
For each context-basis input, one compares endpoint frequencies from two
executions of the same two-step passage, differing only by the insertion of an
unread context measurement at the intermediate boundary; the resulting
kernels determine $\nu_{\C}$ through their normalized Frobenius separation.
For the saturating Hadamard pair of Theorem~\ref{thm:bound-main}, a sequence of
two balanced beam splitters gives $K_{\rm end}=\openone$ and
$K_{\rm step}=J_2/2$, hence $\nu_{\C}(H,H)=1$: the endpoint-only protocol
returns either prepared alternative with certainty, whereas the stepwise
protocol is uniform~\cite{FeynmanLectures}.  More generally, for a fixed
context, a sequence is boundary-stable precisely when these endpoint kernels
agree for every context-basis input and every internal cut of every contiguous
subpassage.  Mid-circuit measurement has been demonstrated in superconducting
and trapped-ion processors~\cite{CorcolesDynamic,ZhuMidCircuit}, making direct
estimation of $\nu_{\C}$ experimentally feasible.  Theorem~\ref{thm:qmp-promisebqp-main}
addresses the distinct prediction problem: an efficient classical algorithm
for predicting the designated musical output of arbitrary quantum circuits
would imply $\mathrm{BQP}=\mathrm{BPP}$.

\Needspace{8\baselineskip}
\section{Conclusion}

The quantum--classical interface is rich: unread intermediate measurements can change the quantum state without changing the
endpoint statistics, while genuinely mixing sequences can reproduce
their stochastic transition law at every prefix for arbitrarily many steps.
That agreement can nevertheless fail when the state is freshly prepared at
an internal boundary and the unchanged remaining operations are applied.
Endpoint fidelity from the
preparation boundary need not supply a reusable stochastic law on
proper subintervals.

Boundary stability means that the
operation-supplied transition laws must
remain valid on every independently initialized contiguous passage.  This
demand has a sharp finite-dimensional cost: at most $d$ full-support steps in
dimension $d$, with the bound attained in every prime dimension.  The dense
two-qubit realization shows that exact composition can survive entangled
dynamics with active interference.

The Born--Chapman--Kolmogorov current records the composition defect, and its
normalized Frobenius norm makes the failure quantitative.
The associated conditional
schedule information measures, in bits per run and conditioned on the chosen
input, what the endpoint reveals about a fair choice
between coherent and checked protocols; it vanishes exactly when their
endpoint laws agree.

The state-machine connection provides a foundation for a quantum theory of
music: unitary operations are notes, temporal boundaries are cues, and a
classical output map assigns sounds to readout outcomes. Born-kernel
composition determines which readout rhythms preserve a score's endpoint
distribution. A score can preserve that distribution under stepwise readout
from its entrance cue yet fail the corresponding test after fresh preparation
at an internal cue. This is the quantum music paradox.
For the circuit-universal score model, promised designated-transition
prediction is \PromiseBQP-complete.

Our question remains
narrower than the general measurement problem, but identifies a common
boundary between measurement-induced consistency, stochastic processes,
and computation.  This framework introduces broad classes of quantum state machines.

\enlargethispage{\baselineskip}
\medskip
\noindent\textbf{Acknowledgments.}
I thank Ian Durham, Dimitry Guskov, George Musser Jr., Cai Waegell,
Guy-Philippe Nadon, and Michael McGuffin for helpful discussions, and all
participants in the \'ETS-organized \emph{Des quanta aux qualia} workshop, held
in Montr{\'e}al and Eastman, Qu{\'e}bec, 11--15 August 2026, where this work
was presented.  Audio rendering uses
\href{https://github.com/Tonejs/audio/tree/master/salamander}{Salamander Grand Piano recordings by Alexander Holm},
licensed \href{https://creativecommons.org/licenses/by/3.0/}{CC BY 3.0}.
The recordings render circuit-generated outcomes and do not determine the
probability law.  Generative AI tools assisted with literature searches and
the refinement of Python, Mathematica, and JavaScript code for algebraic
and numerical examples.  The author checked the cited sources and resulting
calculations and takes responsibility for the manuscript's claims and
conclusions.  This work was
supported by the Minist\`ere de l'\'Economie, de
l'Innovation et de l'\'Energie du Qu\'ebec (MEIE) through the MEIE Principal
Research Chair in Quantum Information Processing.

\clearpage
\begin{center}
\fbox{\begin{minipage}{0.94\linewidth}
\small
\raggedright
\textbf{Popular summary.}

\medskip
\textbf{The problem.}

Quantum probabilities for a multistep process generally cannot be
reconstructed by multiplying step-by-step transition probabilities: checking
the intermediate steps can change the final outcome statistics.

\textbf{The question.}

When can a sequence of unitary operations be measured at every intermediate
step without changing the final outcome distribution?  How badly can this agreement fail
otherwise?

\textbf{The answer.}

We classify when stepwise and endpoint statistics agree when measured on every contiguous
interval.  We then introduce a normalized measure of their disagreement, prove
its bound is sharp, and prove that deciding whether a designated musical
output probability lies above or below promised thresholds is
\PromiseBQP-complete. The prediction
result also applies to ordinary machine output symbols.

\textbf{The consequence.}

Applied to quantum music\footnote{Interactive project page with musical demonstrations: \url{https://institut-kets.github.io/quantum-composition/\#composition-heading}.}, unitary operations are notes and temporal
boundaries are cues. Recorded outcomes produce sounds, and the times between
readouts determine how long those sounds are held. Composition identifies
the rhythms that preserve the final outcome distribution. Reading out after
every operation can preserve that distribution from the opening cue, yet
change it when the unchanged remaining operations begin from a fresh
preparation at a later cue. This is the quantum music paradox.
\end{minipage}}
\end{center}

\clearpage
\pagestyle{appendixpages}
\section*{Appendix to ``The Quantum Composition Paradox''}
\par
{\noindent\textbf{Jacob Biamonte}\quad
\small {\'E}TS Montr{\'e}al, Universit{\'e} du Qu{\'e}bec, Montreal, QC, Canada\par}
\medskip

This appendix is organized around verification of the main paper's claims.
Compact conventions and a proof map are followed by the universal
classification and its state-machine consequence, prescribed-pair and
boundary results, quantitative bounds, and the analytic circuit. Musical
realizations and transition-prediction complexity come last, so their
mathematical foundations are available before the application.

\appendix
\startcontents[appendix]
\section*{\contentsname}
\begingroup
\small
\printcontents[appendix]{}{1}{}
\endgroup
\medskip
\clearpage

\section{Conventions and proof map}
\label{app:formal-opening}

Section~\ref{app:classification} proves the universal classification and
its reversible-state-machine consequence. Section~\ref{app:qubit} develops
the prescribed-pair criteria, prefix and boundary laws, and their relation
to dephasing and histories. Section~\ref{app:quantitative-bounds} gives the
schedule-information, current, and capacity bounds.
Section~\ref{app:walsh} verifies the analytic three-step circuit.
Section~\ref{app:quantum-music-complexity} supplies the musical realizations
and the prediction and sampling proofs.

\subsection{Fixed context and chronological conventions}
\label{app:standing-notation}

\begin{remark}[Standard notation]
\label{rem:standard-notation-app}
We write \(\ii:=\sqrt{-1}\), \(\openone\) for the identity, \(\Tr\) for the
trace, and \(\U(d)\) for the unitary group. The symbols \(A^\dagger\),
\(\overline A\), and \(A^{\mathsf T}\) denote conjugate transpose,
entrywise conjugation, and transpose. We use the norms
\[
 \|A\|_{\mathrm F}:=[\Tr(A^\dagger A)]^{1/2},\qquad
 \|A\|_{\mathrm{op}}:=\sup_{\|v\|_2=1}\|Av\|_2.
\]
\end{remark}

\begin{definition}[Finite fixed-context model]
\label{def:fixed-context-model-app}
Fix \(\mathsf H\cong\mathbb C^d\), \(d\geq2\), the outcome set
\(X=\{0,\ldots,d-1\}\), and the atomic context
\(\C=\{P_x=|x\rangle\langle x|:x\in X\}\) of an orthonormal basis%
, with \(\sum_xP_x=\openone\).
A \(\C\)-diagonal preparation is
\(\rho=\sum_xq(x)P_x\), where \(q\) is a probability distribution.
The Born kernel of \(U\in\U(d)\) and the unread context measurement are
\[
 [\Born_{\C}(U)]_{yx}:=\Tr(P_yUP_xU^\dagger)
 =|\langle y|U|x\rangle|^2,\qquad
 \Mc(\rho):=\sum_xP_x\rho P_x.
\]
We abbreviate \(\Born_{\C}\) to \(\Born\) when the context is unambiguous.
The Born kernel acts on outcome probability vectors.
\end{definition}

\begin{definition}[Stochastic kernels and chronological composition]
\label{def:kernel-conventions-app}
A stochastic kernel on \(X\) is a real matrix \(K=(K_{yx})\) with
\(K_{yx}\geq0\) for all \(x,y\) and \(\sum_yK_{yx}=1\) for every input \(x\).  Columns label
inputs, rows label outputs, and column probability vectors evolve as
\(q'=Kq\).  A kernel is \emph{doubly
stochastic} (also called \emph{bistochastic}) when its rows also sum to one and
\emph{unistochastic} when \(K=\Born_{\C}(U)\) for some unitary \(U\).
Throughout, \(V\circ U\) means that \(U\) acts
before \(V\); the same
chronological convention is used for unitary operations and stochastic
kernels.
\end{definition}

\begin{definition}[Passages, boundaries, and restarts]
\label{def:passage-boundary-app}
A prescribed passage is an ordered finite sequence
\((U_1,\ldots,U_L)\in\U(d)^{\times L}\), with \(U_1\) acting first.
For \(1\leq a\leq b\leq L\), write
\[
 U_{b:a}:=U_b\circ U_{b-1}\circ\cdots\circ U_a.
\]
for the contiguous passage containing the steps
\(U_a,\ldots,U_b\), with both endpoint steps included; in particular,
\(U_{k:k}=U_k\).  Set
\(K_{b:a}:=\Born_{\C}(U_{b:a})\).  Number the temporal boundaries
\(0,\ldots,L\), so that \(U_k\) carries boundary \(k-1\) to boundary \(k\)
and \(U_{b:a}\) begins at boundary \(a-1\) and ends at boundary \(b\).
A \emph{restart immediately before \(U_a\)}, at boundary \(a-1\),
means preparing a fresh \(\C\)-diagonal state there and applying
\(U_{b:a}\); it does not mean carrying forward the quantum state from an
earlier passage. In the musical semantics, every temporal boundary is a
\emph{cue}; a restart specifies a fresh preparation at the chosen cue.
\end{definition}

\begin{remark}[Step-index check]
\label{rem:boundary-index-check-app}
A three-step passage has boundaries \(0,1,2,3\):
\begin{equation*}
 U_{3:1}=U_3\circ U_2\circ U_1,\qquad
 U_{3:2}=U_3\circ U_2.
\end{equation*}
This inclusive step-index convention is used throughout the paper
and appendix.
\end{remark}

\subsection{Composition scopes}
\label{app:axioms}

\begin{definition}[Composition scopes and operation-supplied genealogy]
\label{def:composition-scopes-app}
\label{def:operation-supplied-genealogy-app}
\label{def:prefix-boundary-stability-app}
An interval family \(\mathscr I\) consists of inclusive step intervals
\([a,b]\). A cut after \(c\), with \(a\leq c<b\), is allowed when
\([a,b],[a,c],[c+1,b]\in\mathscr I\). The actual interval Born kernels
carry a \emph{Born-induced stochastic genealogy}, or
\emph{operation-supplied genealogy}, when
\begin{equation}
 K_{b:a}=K_{b:c+1}\circ K_{c:a}
 \qquad\text{at every allowed cut}.
 \label{eq:genealogy-criterion}
\end{equation}
A sequence is \emph{prefix-consistent} when
\begin{equation}
 \Born_{\C}(U_{k:1})
 =\Born_{\C}(U_k)\circ\cdots\circ\Born_{\C}(U_1)
 \qquad(2\le k\le L),
 \label{eq:appendix-prefix-law}
\end{equation}
and \emph{boundary-stable} when
\begin{equation}
 \begin{aligned}
  \Born_{\C}(U_{b:a})
  &=\Born_{\C}(U_b)\circ\cdots\circ\Born_{\C}(U_a),\\[-1pt]
  &\hspace{1em}(1\le a\le b\le L).
 \end{aligned}
 \label{eq:appendix-cue-law}
\end{equation}
Boundary stability is equivalent to a genealogy on all contiguous step
intervals; prefix consistency is weaker. In the musical semantics,
boundary stability is \emph{cue stability} under fresh preparation at any cue.
\end{definition}

\begin{remark}[Four composition scopes]
\label{rem:four-composition-scopes-app}
A pair identity concerns two adjacent blocks; prefix consistency concerns
every passage from boundary \(0\); boundary stability concerns every
contiguous passage. Universal composition quantifies over every ordered
pair in a composition-closed family (Section~\ref{app:classification}).
An identity only for the complete passage's endpoint kernel does not
impose prefix consistency.
\end{remark}

\begin{definition}[Support and qubit mixing]
\label{def:support-mixing-app}
A stochastic kernel has \emph{full support} when every entry is strictly
positive.  For qubit Born kernels only, \emph{genuinely mixing} means full
support.  We reserve \emph{strictly positive} for an entrance distribution
\(\pi\) satisfying \(\pi(x)>0\) for every \(x\).  A \emph{unitary monomial}
is a product \(\Delta\circ\Pi\) of a diagonal unitary \(\Delta\) and a
permutation matrix \(\Pi\);
phase-dressed permutation
is its descriptive
synonym.
\end{definition}

\begin{remark}[Register and qubit notation]
\label{rem:global-notation-app}
We reserve \(L\) for the number of passage steps and \(n\) for the number of
qubits. An \(n\)-qubit register has \(d=N:=2^n\) computational-basis outcomes
and wires \(q_0,\ldots,q_{n-1}\). For the Pauli matrices
\(\sigma_x,\sigma_y,\sigma_z\), we use
\(R_\sigma(\theta):=e^{-\ii\theta\sigma/2}\).
\end{remark}

\paragraph{Remark (context-diagonal preparations).}
\phantomsection\label{app:all-context-inputs}
For an entrance probability vector $p$ and the preparation
$\rho=\sum_xp_xP_x$, the context-outcome distribution after $U$ is
\[
 q_y=\Tr(P_yU\rho U^\dagger)
     =\sum_x[\Born_{\C}(U)]_{yx}p_x,
 \qquad q=\Born_{\C}(U)p.
\]
Thus two endpoint kernels agree if and only if they give the same output
distribution for every context-basis preparation, equivalently for every
$\C$-diagonal preparation.
The displayed equality follows from linearity in
$\rho$. Kernel equality
implies equality on every probability vector. Conversely, the preparations
$p_x=\delta_{xz}$ test each column $z$ separately. The same argument applies
to products of Born kernels induced by readout schedules.

\subsection{\texorpdfstring{%
Born kernels and the composition current}%
{Born kernels and the composition current}}

Classical path probabilities multiply, and marginalizing an intermediate
label composes the adjacent kernels by the Chapman--Kolmogorov
rule~\cite{app:Chapman,app:Kolmogorov}. The current measures the failure of
this rule for the Born kernels of prescribed unitary passages.

\begin{definition}[Born--Chapman--Kolmogorov current]
\label{def:bck-current-app}
For \(U\) followed by \(V\), the endpoint-checked and stepwise-checked
kernels are
\begin{equation}
 \begin{aligned}
 K_{\rm end}&=\Born(V\circ U),\\
 K_{\rm step}&=\Born(V)\circ\Born(U).
 \end{aligned}
 \label{eq:endpoint-stepwise}
\end{equation}
Here \emph{endpoint-checked} means that the passage remains coherent
until one final context readout; \emph{stepwise-checked} inserts the unread
measurement channel \(\Mc\) at each prescribed intermediate boundary.
Their signed difference
\begin{equation}
 \mathcal J_{\C}(V,U):=K_{\rm end}-K_{\rm step}
 \label{eq:defect}
\end{equation}
becomes a dimensionless signed Born-composition defect,
which we also call the Born--Chapman--Kolmogorov current.  Its columns sum to
zero.  It vanishes exactly when an unread context check
after \(U\) preserves every final context probability for every basis
preparation.
\end{definition}

\begin{definition}[Normalized Frobenius norm of the current]
\label{def:normalized-current-app}
Assume $d\geq2$.  For a $d\times d$ complex matrix
$M=(M_{yx})_{x,y\in X}$, the Frobenius norm is the Euclidean length of the
vector formed by its $d^2$ entries:
\begin{equation}
 \|M\|_{\mathrm F}
 :=\left(\sum_{x,y\in X}|M_{yx}|^2\right)^{1/2}
 =\left[\Tr(M^\dagger M)\right]^{1/2}.
 \label{eq:frobenius-norm-app}
\end{equation}
Applying
Eq.~\eqref{eq:frobenius-norm-app} to the current, we define its normalized
Frobenius norm by
\begin{equation}
 \nu_{\C}(V,U)
 :=\frac{1}{\sqrt{d-1}}\left\|\mathcal J_{\C}(V,U)\right\|_{\mathrm F}.
 \label{eq:metric}
\end{equation}

The factor $\sqrt{d-1}$ is the sharp maximum of the current's
Frobenius norm, proved in Theorem~\ref{thm:bound-main}.  Consequently
$0\leq\nu_{\C}(V,U)\leq1$, and $\nu_{\C}(V,U)=0$ exactly when the current
vanishes.
\end{definition}

\section{Universal classification and reversible state machines}
\label{app:classification}

\subsection{Universal composition}

\begin{definition}[Unitary grammar and universal composition]
\label{def:unitary-grammar-app}
\begin{sloppypar}
A \emph{unitary grammar}
$\Gamma$ is a
composition-closed subsemigroup of $\U(d)$ generated by the allowed
operations.  Universal composition means that the canonical Born map
is a homomorphism on every ordered pair:
\begin{equation}
 \Born_{\C}(V\circ U)=\Born_{\C}(V)\circ\Born_{\C}(U)
 \label{eq:strict-extra}
\end{equation}
for every $U,V\in\Gamma$.  This quantifier is stronger than the genealogy law
for one prescribed interval family in Eq.~\eqref{eq:genealogy-criterion}.
\end{sloppypar}
\end{definition}

\subsection{Classification theorem}
\label{app:strict-classification}

\begin{theorem}[%
Classification of Born-consistent unitary semigroups]
\label{thm:semigroup-app}
Let \(\Gamma\subseteq\U(d)\) be closed under composition.  Then
\[
 \Born_{\C}(V\circ U)=\Born_{\C}(V)\circ\Born_{\C}(U)
 \qquad\text{for all }U,V\in\Gamma
\]
if and only if every element of \(\Gamma\) is a unitary monomial in the
context basis.
\end{theorem}

\subsection{Bistochastic contractions}

\begin{lemma}[Bistochastic contraction]
\label{lem:bistochastic-contraction}
Every real $d\times d$ bistochastic matrix $M$ is a contraction in Euclidean
norm:
\begin{equation}
 \|Mv\|_2\le\|v\|_2
 \qquad(v\in\R^{d}).
 \label{eq:bistochastic-contraction}
\end{equation}
If $M$ is an isometry, then $M$ is a permutation matrix.
\end{lemma}

\begin{proof}
Jensen's inequality applied to each row gives
\begin{align}
 \|Mv\|_2^2
 &=\sum_y\left(\sum_x M_{yx}v_x\right)^2
 \le\sum_y\sum_x M_{yx}v_x^2\\
 &=\sum_x v_x^2\sum_y M_{yx}
 =\|v\|_2^2.
\end{align}
If equality holds for every $v$, then $M^{\mathsf T}M=\openone$.  Each column of
$M$ is nonnegative, has sum one, and has squared Euclidean norm one.  A
probability vector has squared norm one only when one entry is one and all
others vanish.  The row sums place these unit entries in distinct rows, so
$M$ is a permutation matrix.
\end{proof}

\subsection{Lemmas for the semigroup theorem}

\begin{lemma}[Unitary recurrence]
\label{lem:unitary-recurrence}
For every finite-dimensional unitary $U$, there is a sequence
$m_j\to\infty$ for which $U^{m_j}\to\openone$.
\end{lemma}

\begin{proof}

Write the eigenvalues of \(U\) as
\(e^{2\pi\ii\theta_1},\ldots,e^{2\pi\ii\theta_d}\).  Simultaneous
Diophantine approximation supplies integers \(m_j\to\infty\) for which every
distance \(\|m_j\theta_r\|_{\mathbb R/\mathbb Z}\) tends to zero.  Hence each
eigenvalue of \(U^{m_j}\) tends to one, and therefore
\(U^{m_j}\to\openone\) in operator norm.
\end{proof}

\begin{lemma}[Recurrent bistochastic powers]
\label{lem:recurrent-bistochastic}
Let $A$ be a real bistochastic matrix.  If
$A^{m_j}\to\openone$ along a sequence $m_j\to\infty$, then $A$ is a
permutation matrix.
\end{lemma}

\begin{proof}
For every $v\in\R^d$, Lemma~\ref{lem:bistochastic-contraction} gives
\[
 \|v\|_2
 =\lim_j\|A^{m_j}v\|_2
 \leq\|Av\|_2
 \leq\|v\|_2.
\]
Thus $A$ is an isometry, and the final assertion of
Lemma~\ref{lem:bistochastic-contraction} makes it a permutation matrix.
\end{proof}

\subsection{\texorpdfstring{%
Proof of the %
classification of Born-consistent unitary semigroups}%
{Proof of the classification of Born-consistent unitary semigroups}}
\label{app:semigroup-proof}

\begin{proof}[Proof of
Theorem~\ref{thm:semigroup-app}]
Assume the universal
composition identity in Theorem~\ref{thm:semigroup-app} and choose $U\in\Gamma$.  Put
$A=\Born_{\C}(U)$.  Closure and induction give
\begin{equation}
 \Born_{\C}(U^m)=A^m
 \qquad(m\ge1).
 \label{eq:powers}
\end{equation}
By Lemma~\ref{lem:unitary-recurrence}, choose $m_j\to\infty$ with
$U^{m_j}\to\openone$.  Continuity of the Born map and
Eq.~\eqref{eq:powers} give $A^{m_j}\to\openone$, so
Lemma~\ref{lem:recurrent-bistochastic} makes $A$ a permutation matrix.  Since
$|U_{yx}|^2=A_{yx}$, the unitary $U$ has one nonzero entry of unit modulus in
each row and column.  Thus $U$ is a
 unitary monomial, that is, a
phase-dressed permutation.

Conversely, write $U=\Delta_U\circ\Pi_U$ and
$V=\Delta_V\circ\Pi_V$, where the $\Delta$ matrices are diagonal unitary and
the $\Pi$ matrices are permutations.  Then
\begin{equation}
 \Born_{\C}(U)=\Pi_U,
 \qquad
 \Born_{\C}(V)=\Pi_V,
 \qquad
 \Born_{\C}(V\circ U)=\Pi_V\circ\Pi_U,
\end{equation}
which proves the converse
direction of Theorem~\ref{thm:semigroup-app}.
\end{proof}

Equivalently,
these are precisely the unitaries that preserve incoherence relative to
\(\C\): the normalizer
of the diagonal algebra~\cite{app:BaumgratzCramerPlenio}.

\subsection{Universal composition forces reversible state machines}
\label{app:state-machine-corollary}

\begin{definition}[Moore machine and fixed-context unitary completion]
\label{def:moore-completion-app}
A Moore machine is a tuple
\(\mathfrak M=(X,\Sigma,\delta,\lambda_{\rm out},x_0)\), where \(X\) is a finite state
set, \(\Sigma\) is a finite input alphabet,
\(\delta:X\times\Sigma\to X\) is the transition function,
\(\lambda_{\rm out}:X\to\mathsf A\) is an output map to a finite alphabet, and
\(x_0\in X\) is the initial state.  A fixed-context unitary completion assigns
\(U_\sigma\in\U(d)\) to each \(\sigma\in\Sigma\) so that
\([\Born_{\C}(U_\sigma)]_{yx}=1\) exactly when
\(y=\delta(x,\sigma)\).  The completion is universal when the generated
composition-closed family obeys universal Born-kernel composition.
\end{definition}

\begin{corollary}[Universal completion is exactly reversible]
\label{cor:moore-reversible-app}
A finite Moore machine admits a universal fixed-context unitary
completion if and only if every map
\(\delta_\sigma:x\mapsto\delta(x,\sigma)\) is bijective.  Equivalently,
universal completion uniquely determines reversible deterministic
transitions, although their diagonal unitary phases remain arbitrary.
\end{corollary}

\begin{proof}
Universal Born-kernel composition uniquely determines a reversible classical
transition for every input symbol.  Let $X$ label the outcomes of the atomic
context $\C$, let $\Sigma$ be a finite input alphabet, assign a unitary
$U_\sigma$ to each $\sigma\in\Sigma$, and let
\begin{equation}
 \Gamma=\langle U_\sigma:\sigma\in\Sigma\rangle
 \label{eq:moore-generated-semigroup-app}
\end{equation}
be the composition-closed family generated by those operations.  If
\begin{equation}
 \Born_{\C}(V\circ U)=\Born_{\C}(V)\circ\Born_{\C}(U)
 \qquad\text{for all }U,V\in\Gamma
 \label{eq:moore-universal-composition-app}
\end{equation}
holds, Theorem~\ref{thm:semigroup-app} gives, for every $\sigma\in\Sigma$, a
unique permutation $\pi_\sigma$ of $X$ with
\begin{equation}
 \Born_{\C}(U_\sigma)=\Pi_\sigma,
 \qquad
 (\Pi_\sigma)_{yx}=\begin{cases}
  1,&y=\pi_\sigma(x),\\
  0,&y\ne\pi_\sigma(x),
 \end{cases}
 \label{eq:moore-permutation-app}
\end{equation}
Define $\delta(x,\sigma):=\pi_\sigma(x)$ and write
$\delta_\sigma(x):=\delta(x,\sigma)$.  For any finite output alphabet
$\mathsf A$, output map $\lambda_{\rm out}:X\to\mathsf A$, and initial state $x_0\in X$,
the induced Moore machine
$\mathfrak M=(X,\Sigma,\delta,\lambda_{\rm out},x_0)$ is therefore deterministic and
reversible: every $\delta_\sigma:X\to X$ is a bijection.

\begingroup
\widowpenalty=10000
Conversely, a preassigned Moore machine admits a universal fixed-context
unitary completion, meaning
$[\Born_{\C}(U_\sigma)]_{yx}=1$ exactly when
$y=\delta(x,\sigma)$, if and only if every $\delta_\sigma$ is bijective.
For a reversible machine, take the permutation matrix $\Pi_\sigma$ of
$\delta_\sigma$ and choose
$U_\sigma=\Delta_\sigma\circ\Pi_\sigma$ with arbitrary diagonal unitaries
$\Delta_\sigma$.  The output map $\lambda_{\rm out}$ and initial state $x_0$ play no role
in this classification.  The phases make the unitary completion nonunique but
remain invisible to its uniquely induced fixed-context transition.
\par\endgroup
\end{proof}

\section{Prescribed pairs and boundary results}
\label{app:qubit}

We first express pair composition as a compression identity and then give
the qubit and qutrit criteria. Prefix recurrence and the path formulation
lead to the boundary and schedule conditions, followed by the distinction
from consistent fine-grained histories.

\subsection{The double-operator Born-kernel linearization}
\label{app:compression}

\begin{definition}[Double-operator representation and sector projections]
\label{def:double-operator-app}
Let \(\mathsf H_d\cong\mathbb C^{d}\) be a separate Hilbert space
over the classical labels, and
let \(\HS(\mathsf H):=\operatorname{End}(\mathsf H)\) denote the
complex vector space of linear operators on \(\mathsf H\), equipped with the
Hilbert--Schmidt inner product
\(\langle A,B\rangle_{\HS}=\Tr(A^\dagger B)\).  Define
\begin{equation}
 \iota_{\C}
 :\mathsf H_d\longrightarrow\HS(\mathsf H),
 \qquad
 \iota_{\C}|x\rangle_{\mathrm{cl}}:=P_x.
 \label{eq:delta-clear}
\end{equation}
The projectors \(P_x\) are orthonormal, so
\(\iota_{\C}^\dagger\iota_{\C}=\openone\).  The adjoint action of a
unitary is the double-operator action
\begin{equation}
 L_U:\HS(\mathsf H)\longrightarrow\HS(\mathsf H),
 \qquad
 L_U(A)=UAU^\dagger.
 \label{eq:liouville-action-app}
\end{equation}
With column-stacking vectorization in the context basis, \(L_U\) is represented
by \(\overline U\otimes U\).  The Born kernel is its
contraction to the diagonal sector:
\begin{equation}
 \Born_{\C}(U)=\iota_{\C}^\dagger\circ L_U\circ\iota_{\C}.
 \label{eq:born-compression-app}
\end{equation}

Let
\begin{equation}
 P_{\C}:=\iota_{\C}\circ\iota_{\C}^\dagger,
 \qquad
 Q_{\C}:=\openone-P_{\C}.
\end{equation}
The two projections select the diagonal and off-diagonal
Hilbert--Schmidt sectors.
\end{definition}

\begin{proposition}[Double-operator composition identity]
For all \(U,V\in\U(d)\),
\begin{equation}
 \Born_{\C}(V\circ U)
 =\Born_{\C}(V)\circ\Born_{\C}(U)
 +\iota_{\C}^\dagger\circ L_V\circ Q_{\C}\circ L_U\circ\iota_{\C}.
 \label{eq:compression-identity}
\end{equation}
\end{proposition}

\begin{proof}
Use \(L_{V\circ U}=L_V\circ L_U\) and insert
\(\openone=P_{\C}+Q_{\C}\) between the lifted operations.  Since
\(P_{\C}=\iota_{\C}\circ\iota_{\C}^\dagger\), the diagonal term is the product
of the two contracted Born kernels, and the off-diagonal term is the
displayed correction.
\end{proof}

The defect is therefore the coherence-mediated correction that leaves the
diagonal population sector under \(U\) and returns to that sector under \(V\).
Quantum composition itself is unchanged;
the mismatch is created by contracting to that sector between operations.
Equivalently, on Hilbert--Schmidt space,
\begin{equation}
 P_{\C}\circ L_V\circ L_U\circ P_{\C}
 =P_{\C}\circ L_V\circ P_{\C}\circ L_U\circ P_{\C}
 +P_{\C}\circ L_V\circ Q_{\C}\circ L_U\circ P_{\C}.
 \label{eq:projection-memory-identity}
\end{equation}
The correction
is coherence created by \(U\) and read back into populations by \(V\).  It is a
fixed-context projection effect, not an open-system memory kernel; compare
to the projection-operator literature~\cite{app:Nakajima,app:Zwanzig}.

\subsection{The scalar form of every qubit Born kernel}

\begin{definition}[Qubit Born-kernel coordinate]
\label{def:qubit-kernel-coordinate-app}
For \(-1\leq\lambda\leq1\), define
\[
 K(\lambda):=\frac{1+\lambda}{2}\openone
 +\frac{1-\lambda}{2}\sigma_x.
\]
Every \(2\times2\) bistochastic matrix has this form, and
\(K(\mu)\circ K(\lambda)=K(\mu\lambda)\).  For a qubit unitary \(U\), its
population polarization is
\[
 \beta(U):=\langle0|U^\dagger\sigma_zU|0\rangle,
 \qquad \Born_{\C}(U)=K(\beta(U)).
\]
The kernel has full support, and is therefore genuinely mixing in the
terminology of Definition~\ref{def:support-mixing-app}, exactly when
\(|\lambda|<1\).
\end{definition}

Write
\begin{equation}
 \begin{aligned}
 \sigma_x&=
 |0\rangle\langle1|+|1\rangle\langle0|,\\
 \sigma_y&=
 -\ii|0\rangle\langle1|+\ii|1\rangle\langle0|,\\
 \sigma_z&=
 |0\rangle\langle0|-|1\rangle\langle1|.
 \end{aligned}
 \label{eq:pauli-matrices-app}
\end{equation}
In the computational context, Dirac notation gives
\begin{equation}
 \Born(U)=
 \begin{pmatrix}
  |\langle0|U|0\rangle|^2&|\langle0|U|1\rangle|^2\\
  |\langle1|U|0\rangle|^2&|\langle1|U|1\rangle|^2
 \end{pmatrix}
 =K(\beta(U)),
 \label{eq:qubit-dirac-kernel}
\end{equation}
where
\begin{equation}
 \beta(U)
 :=\langle0|U^\dagger \sigma_zU|0\rangle
 =|\langle0|U|0\rangle|^2-|\langle1|U|0\rangle|^2.
 \label{eq:qubit-dirac-beta}
\end{equation}
Thus
\begin{equation}
 K(\lambda):=\frac{1+\lambda}{2}\openone
 +\frac{1-\lambda}{2}\sigma_x.
 \label{eq:K-lambda}
\end{equation}
Every $2\times2$ bistochastic matrix has this form, and
\begin{equation}
 K(\mu)\circ K(\lambda)=K(\mu\cdot\lambda).
 \label{eq:K-product}
\end{equation}
Here $\lambda$ is a generic kernel parameter, whereas $\beta(U)$ is the
unitary-specific polarization.  The former is the nontrivial eigenvalue of
$K(\lambda)$ on the signed population mode $(1,-1)^{\mathsf T}$.  The kernel
has full support exactly when $|\lambda|<1$;

Definition~\ref{def:support-mixing-app} calls this genuine mixing
and also fixes the equivalent unitary-monomial terminology.
Therefore
\begin{equation}
 \Born(V\circ U)=\Born(V)\circ\Born(U)
 \quad\Longleftrightarrow\quad
 \beta(V\circ U)=\beta(V)\beta(U).
 \label{eq:scalar-iff}
\end{equation}

\begin{definition}[Qubit interference scalar]
\label{def:qubit-interference-scalar-app}
For two paths through the intermediate context, define the real interference
scalar
\begin{equation}
 \chi(V,U):=4\operatorname{Re}\!\left[
 \langle0|V|0\rangle\langle0|U|0\rangle
 \overline{\langle0|V|1\rangle\langle1|U|0\rangle}
 \right].
 \label{eq:qubit-dirac-current}
\end{equation}
\end{definition}

\begin{theorem}[%
Classification of pairs of one-qubit unitary steps]
For arbitrary qubit unitaries $U,V$,
\begin{equation}
 \Born(V\circ U)=\Born(V)\circ\Born(U)
 \quad\Longleftrightarrow\quad
 \chi(V,U)=0,
 \label{eq:qubit-pair-iff-dirac}
\end{equation}
and
\begin{equation}
 \Born(V\circ U)-\Born(V)\circ\Born(U)
 =\frac{\chi(V,U)}2(\openone-\sigma_x).
 \label{eq:qubit-full-defect-dirac}
\end{equation}
\end{theorem}

Because \(\|\openone-\sigma_x\|_{\mathrm F}=2\), the normalized current
introduced in Eq.~\eqref{eq:metric} reduces for a qubit to
\(\nu_{\C}(V,U)=|\chi(V,U)|\).  Thus the scalar classification is exactly the
zero set of the current used in the main text.

\begin{proof}
For entrance and endpoint $0$, the coherent amplitude is the sum of the two
context-path amplitudes
\begin{equation}
 \langle0|V\circ U|0\rangle
 =\langle0|V|0\rangle\langle0|U|0\rangle
 +\langle0|V|1\rangle\langle1|U|0\rangle.
\end{equation}
The difference between the squared modulus of this sum and the sum of the two
squared moduli is $\chi(V,U)/2$.  Both kernels are $2\times2$ bistochastic, so
their signed difference has zero row and column sums and is consequently this
single scalar times $\openone-\sigma_x$.  This proves both claims.
\end{proof}

The same scalar identity has a direct path-amplitude form.  Resolve the
intermediate context identity and define
\begin{equation}
 a_j(V,U):=\langle0|V|j\rangle\langle j|U|0\rangle,
 \qquad j\in\{0,1\}.
 \label{eq:qubit-path-amplitudes}
\end{equation}
Then
\begin{equation}
 \langle0|V\circ U|0\rangle
 =\sum_{j=0}^{1}\langle0|V|j\rangle\langle j|U|0\rangle
 =a_0(V,U)+a_1(V,U),
 \label{eq:qubit-path-amplitude-sum}
\end{equation}
and
\begin{equation}
 \chi(V,U)=4\operatorname{Re}\!\left[a_0(V,U)\overline{a_1(V,U)}\right].
 \label{eq:qubit-path-current}
\end{equation}

Equation~\eqref{eq:qubit-path-current} already shows that the
composition law holds precisely when the two path amplitudes have vanishing
real cross term.

\begingroup
\setlength{\emergencystretch}{1em}
The body witness follows immediately.
 For $U=R_{\sigma_y}(\pi/3)$ and
$V=R_{\sigma_x}(\phi)$ with $\cos\phi=1/3$, the two path amplitudes are
\par\endgroup
\[
 a_0(V,U)=\frac{\sqrt2}{2},
 \qquad
 a_1(V,U)=-\frac{\ii}{2\sqrt3}.
\]
Both are nonzero and have relative phase $-\pi/2$, so their interference
term has zero real part.  Consequently,
\begin{equation}
 \Born(U)=K(1/2),
 \qquad
 \Born(V)=\;K(1/3),
 \qquad
 \Born(V\circ U)=\;K(1/6).
 \label{eq:body-witness-from-classification}
\end{equation}

More generally, for
$U=R_{\sigma_y}(\theta)$ and $V=R_{\sigma_x}(\phi)$, the two path
amplitudes are respectively real and purely imaginary.  Therefore
\begin{equation}
 \Born(V\circ U)
 =K(\cos\phi\cos\theta)
 =\Born(V)\circ\Born(U).
 \label{eq:cross-axis-qubit-family}
\end{equation}
This continuous family includes pairs with unequal, genuinely mixing local
kernels, so equality of the two individual Born kernels is neither required
nor generic.

The defect therefore vanishes in two ways.  One path amplitude may be zero,
in which case at least one of the qubit unitaries is
 a unitary monomial.  Otherwise the
two path amplitudes are both nonzero but differ in phase by
$\pm\pi/2$, so their conjugate cross terms cancel exactly.

\subsection[Qutrit two-step analysis]{%
\texorpdfstring{%
Qutrit two-step analysis}{Qutrit two-step analysis}}
\label{app:qutrit}

The scalar kernel $K(\lambda)$ is special to a qubit.  A generic qutrit
Born kernel lies in the four-dimensional set of $3\times3$ unistochastic
matrices, and phases invisible to an individual Born kernel can reappear
when two unitaries are composed.  Thus there is no single-parameter qutrit
law analogous to Eq.~\eqref{eq:K-product}.  There is, however, a
necessary-and-sufficient geometric criterion for every prescribed
ordered pair of qutrit unitary
steps $(U,V)$.

In this subsection, $(U,V)$ denotes an ordered pair of qutrit
unitary steps, with $U$ acting before $V$.

\begin{sloppypar}
\begin{definition}[Qutrit population--coherence blocks and frames]
\label{def:qutrit-blocks-app}
Let \(\mathsf{Herm}_0(3)\) be the real Hilbert space of traceless Hermitian
$3\times3$ matrices with the Hilbert--Schmidt inner product.  The fixed
context gives the orthogonal decomposition
\begin{equation}
 \mathsf{Herm}_0(3)
 =\mathsf D\oplus\mathsf C_{\rm coh},
 \qquad
 \dim_{\R}\mathsf D=2,
 \qquad
 \dim_{\R}\mathsf C_{\rm coh}=6,
 \label{eq:qutrit-pop-coh-split}
\end{equation}
where $\mathsf D$ is the diagonal traceless population plane and
$\mathsf C_{\rm coh}$ is the off-diagonal coherence space.  Let
$P_{\rm pop}$ and $P_{\rm coh}=\openone-P_{\rm pop}$ be the corresponding
orthogonal projections.  The adjoint action
$\operatorname{Ad}_U(H)=UHU^\dagger$ has the block form
\begin{equation}
 \operatorname{Ad}_U=
 \begin{pmatrix}
  M_U&B_U\\
  G_U&C_U
 \end{pmatrix}
 :=
 \begin{pmatrix}
  P_{\rm pop}\operatorname{Ad}_UP_{\rm pop}
   &P_{\rm pop}\operatorname{Ad}_UP_{\rm coh}\\
  P_{\rm coh}\operatorname{Ad}_UP_{\rm pop}
   &P_{\rm coh}\operatorname{Ad}_UP_{\rm coh}
 \end{pmatrix}.
 \label{eq:qutrit-blocks}
\end{equation}
Thus
$M_U:\mathsf D\to\mathsf D$,
$G_U:\mathsf D\to\mathsf C_{\rm coh}$,
$B_U:\mathsf C_{\rm coh}\to\mathsf D$, and
$C_U:\mathsf C_{\rm coh}\to\mathsf C_{\rm coh}$.
The first is the centered-population action of $\Born(U)$; $G_U$ generates
coherence from a population contrast, $B_U$ reads coherence back into a
population contrast, and $C_U$ transports coherence without reading it.
The symbol \({}^\star\) below denotes the adjoint of a block map with
respect to the real Hilbert--Schmidt inner products.
Orthogonality of the adjoint representation gives
\begin{equation}
 B_U=G_{U^\dagger}^{\star}.
 \label{eq:qutrit-read-create-adjoint}
\end{equation}
The subspaces \(\operatorname{ran}G_U\) and
\(\operatorname{ran}G_{V^\dagger}\) are respectively the coherence frame
created by \(U\) and the coherence frame readable by \(V\).
\end{definition}
\end{sloppypar}

Block multiplication isolates the entire qutrit composition defect:
\begin{equation}
 M_{V\circ U}=M_VM_U+B_VG_U.
 \label{eq:qutrit-block-product}
\end{equation}

\noindent\par
\begin{theorem}[%
Qutrit two-step criterion]
\label{thm:qutrit-pair-app}
For arbitrary $U,V\in\U(3)$,
\begin{equation}
 \Born(V\circ U)=\Born(V)\circ\Born(U)
 \quad\Longleftrightarrow\quad
 B_VG_U=0
 \quad\Longleftrightarrow\quad
 G_{V^\dagger}^{\star}G_U=0.
 \label{eq:qutrit-pair-iff}
\end{equation}
Equivalently, the coherence frame created by $U$ is orthogonal inside
$\mathsf C_{\rm coh}$ to the coherence frame readable by $V$.
\end{theorem}

\begin{proof}
Both sides are bistochastic and therefore agree on the uniform
population mode.  Hence the Born kernels agree if and only if their actions
agree on the centered population plane $\mathsf D$.
Equation~\eqref{eq:qutrit-block-product}
shows that their difference on this plane is $B_VG_U$.  Equation
\eqref{eq:qutrit-read-create-adjoint} gives the final equivalence.
\end{proof}

The criterion can be written as four scalar equations.  Put
\begin{equation}
 h_x:=P_x-\frac{\openone_3}{3},
 \qquad
 c_x(U):=G_Uh_x,
 \qquad
 d_y(V):=G_{V^\dagger}h_y.
 \label{eq:qutrit-created-readable-vectors}
\end{equation}
Since $\sum_xc_x(U)=\sum_yd_y(V)=0$, two vectors from each triple span its
frame.  The entrywise defect is the cross-Gram matrix
\begin{equation}
 \Born(V\circ U)_{yx}
 -[\Born(V)\circ\Born(U)]_{yx}
 =\langle d_y(V),c_x(U)\rangle_{\HS}.
 \label{eq:qutrit-cross-gram}
\end{equation}
Indeed, expansion in the context projectors gives
\begin{equation}
 \langle d_y(V),c_x(U)\rangle_{\HS}
 =|\langle y|V\circ U|x\rangle|^2
  -\sum_z|\langle y|V|z\rangle|^2|\langle z|U|x\rangle|^2,
 \label{eq:qutrit-cross-gram-expanded}
\end{equation}
which is exactly the displayed Born-kernel defect.
Consequently Eq.~\eqref{eq:qutrit-pair-iff} is equivalent to the vanishing
of the $2\times2$ block with $x,y\in\{0,1\}$.  This is the qutrit
replacement for the single qubit dot product in
Eq.~\eqref{eq:qubit-pair-iff-dirac}.

\begin{definition}[Qutrit coherence rank]
\label{def:qutrit-coherence-rank-app}
The coherence rank of a qutrit
unitary \(U\) is
\begin{equation}
 r(U):=\rank G_U\in\{0,1,2\}.
 \label{eq:qutrit-coherence-rank}
\end{equation}
\end{definition}

The necessary-and-sufficient cross-Gram criterion above proves the
qutrit claim used in the main text.

\subsection{Prefixes and their faithful recurrence}

We return to qubit steps and their scalar Born-kernel coordinate.

\begin{definition}[Prefix interference and prefix defect]
\label{def:prefix-defect-app}
For $k\geq1$, put
\begin{equation}
 \beta_k:=\beta(U_k),
\end{equation}
and
\begin{equation}
 D_k^{\rm pref}:=\beta(U_{k:1})
 -\prod_{r=1}^k\beta_r.
\end{equation}
For $k\geq2$, define the one-step interference increment by
\begin{equation}
 \varepsilon_k
 :=\beta(U_{k:1})-\beta(U_k)\beta(U_{k-1:1}).
\end{equation}
The scalar \(\varepsilon_k\) is the interference created at the new prefix
cut, while \(D_k^{\rm pref}\) is the cumulative prefix-composition defect.
\end{definition}
For $k\ge2$,
\begin{equation}
 D_k^{\rm pref}=\beta_kD_{k-1}^{\rm pref}+\varepsilon_k,
 \label{eq:defect-recurrence}
\end{equation}
Indeed, substitute
$\beta(U_{k:1})
=\beta_k\beta(U_{k-1:1})+\varepsilon_k$
and add and subtract
$\beta_k\prod_{r=1}^{k-1}\beta_r$.  Iteration gives
\begin{equation}
 D_L^{\rm pref}=\sum_{k=2}^L
 \varepsilon_k\prod_{m=k+1}^L\beta_m.
 \label{eq:defect-expanded}
\end{equation}
Prefix consistency at every step is equivalent to
$\varepsilon_k=0$ for every $k\ge2$.  By the pair classification,
\begin{equation}
 \varepsilon_k
 =\chi(U_k,U_{k-1:1})
 =4\operatorname{Re}\!\left[
 a_0(U_k,U_{k-1:1})
 \overline{a_1(U_k,U_{k-1:1})}
 \right].
 \label{eq:prefix-orthogonality}
\end{equation}

\begin{proposition}[Infinite qubit sequences]
There exists an infinite sequence of genuinely mixing qubit steps for which
every finite prefix satisfies the identity in
Eq.~\eqref{eq:appendix-prefix-law}.
\end{proposition}

\begin{proof}
Choose a genuinely mixing first step.  Suppose $U_1,\ldots,U_{k-1}$ have
been chosen prefix-consistently, with every
$|\beta_r|<1$.  Then
\begin{equation}
 |\beta(U_{k-1:1})|
 =\left|\prod_{r=1}^{k-1}\beta_r\right|<1,
\end{equation}
so both amplitudes
\begin{equation}
 \gamma_0:=\langle0|U_{k-1:1}|0\rangle,
 \qquad
 \gamma_1:=\langle1|U_{k-1:1}|0\rangle
\end{equation}
are nonzero.  Write $\gamma_j=|\gamma_j|e^{\ii\theta_j}$, choose any
$\beta_k\in(-1,1)$, and put
\begin{equation}
 c_k:=\sqrt{\frac{1+\beta_k}{2}},
 \qquad
 s_k:=\sqrt{\frac{1-\beta_k}{2}}.
\end{equation}
Choose the first row of $U_k$ to satisfy
\begin{equation}
 \langle0|U_k|0\rangle=c_k,
 \qquad
 \langle0|U_k|1\rangle
 =\ii s_k e^{\ii(\theta_0-\theta_1)},
 \label{eq:dirac-prefix-row-choice}
\end{equation}
and choose the second row as
\begin{equation}
 \langle1|U_k|0\rangle=\ii s_ke^{-\ii(\theta_0-\theta_1)},
 \qquad
 \langle1|U_k|1\rangle=c_k.
\end{equation}
The resulting two rows are orthonormal.  The two path amplitudes then
have the common phase $e^{\ii\theta_0}$ and a relative factor of $\ii$.
Hence $\chi(U_k,U_{k-1:1})=0$, while
$\beta(U_k)=c_k^2-s_k^2=\beta_k$.  The new prefix is consistent and
$|\beta_k|<1$ makes $U_k$ genuinely mixing.  Induction continues for all
$k$; truncation gives a sequence of every prescribed finite length.
\end{proof}

\begin{corollary}[Constant-kernel infinite sequences]
\label{cor:constant-kernel-infinity}
Fix $-1<\lambda<1$, set $\theta=\arccos\lambda$, and define
\begin{equation}
 U_1=R_{\sigma_y}(\theta),
 \qquad
 U_k=R_{\sigma_z}(\alpha_k)\circ R_{\sigma_x}(\theta),
 \qquad k\ge2,
 \label{eq:signed-infinity-notes}
\end{equation}
where $\alpha_k\in(-\pi/2,\pi/2)$ is determined by
\begin{equation}
 \tan\alpha_k
 =\frac{\lambda^{k-1}\sqrt{1-\lambda^2}}
        {\sqrt{1-\lambda^{2(k-1)}}}.
 \label{eq:signed-infinity-phase}
\end{equation}
Then every step has Born kernel $K(\lambda)$ and every finite prefix obeys
\begin{equation}
 \Born(U_\ell\circ\cdots\circ U_1)
 =K(\lambda^\ell)
 =\Born(U_\ell)\circ\cdots\circ\Born(U_1).
 \label{eq:signed-infinity-prefix}
\end{equation}

For every restart just before $U_r$, $r\ge2$, however,
\begin{equation}
 \Born(U_{r+1}\circ U_r)
 =K\!\left(\lambda^2-(1-\lambda^2)\cos\alpha_r\right)
 \ne K(\lambda^2)
 =\Born(U_{r+1})\circ\Born(U_r).
 \label{eq:signed-infinity-cue}
\end{equation}
\end{corollary}

\begin{proof}
Let $\boldsymbol r_k$ be the Bloch vector of
$U_k\circ\cdots\circ U_1|0\rangle$.  We claim
\begin{equation}
 \boldsymbol r_k
 =\left(\sqrt{1-\lambda^{2k}},0,\lambda^k\right).
 \label{eq:signed-infinity-bloch}
\end{equation}
It holds at $k=1$.  If it holds at $k-1$, then $R_{\sigma_x}(\theta)$ sends the last
two components to
\begin{equation}
 \left(-\lambda^{k-1}\sqrt{1-\lambda^2},\lambda^k\right).
\end{equation}
Equation~\eqref{eq:signed-infinity-phase} makes the following
$R_{\sigma_z}(\alpha_k)$ rotation cancel the transverse $y$ component and leaves
Eq.~\eqref{eq:signed-infinity-bloch}.  Therefore
$\Born(U_k)=K(\lambda)$ and the prefix kernel is $K(\lambda^k)$, proving
Eq.~\eqref{eq:signed-infinity-prefix} by
$K(\lambda)\circ K(\mu)=K(\lambda\cdot\mu)$.

For a two-step segment beginning at $U_r$, direct Bloch rotation gives the
polarization
\begin{equation}
 \lambda^2-(1-\lambda^2)\cos\alpha_r.
\end{equation}
Since $|\lambda|<1$ and $\alpha_r\in(-\pi/2,\pi/2)$, this differs from
$\lambda^2$, proving Eq.~\eqref{eq:signed-infinity-cue}.
At $\lambda=1/2$, \
the local formulas in Eqs.~\eqref{eq:signed-infinity-notes} and
\eqref{eq:signed-infinity-phase} give the constant-amplitude specialization.
At $\lambda=-1/2$, they give the signed companion, with
\begin{equation}
 \tan\alpha_k=(-1)^{k-1}
 \frac{\sqrt3}{2\sqrt{4^{k-1}-1}}.
\end{equation}
\end{proof}

\subsection{Discrete path-integral formulation}
\label{app:discrete-path-integral}

\begin{definition}[Context paths, amplitudes, and weights]
\label{def:context-paths-app}
The composition law admits an exact sum-over-paths form.  Fix

an entrance \(x_0=x\) and an endpoint \(x_L=y\), and let
\begin{equation}
 \Omega_{yx}^{(L)}
 :=\{(x_0,\ldots,x_L)\in X^{L+1}:x_0=x,\ x_L=y\}
 \label{eq:discrete-path-set}
\end{equation}
be the context paths joining them.  The amplitude and the corresponding
stepwise path weight are
\begin{align}
 \mathcal A(\gamma)
 &:=\prod_{t=1}^{L}\langle x_t|U_t|x_{t-1}\rangle,
 \label{eq:discrete-path-amplitude}\\
 w(\gamma)
 &:=|\mathcal A(\gamma)|^2
   =\prod_{t=1}^{L}\langle x_t|\Born_{\C}(U_t)|x_{t-1}\rangle.
 \label{eq:discrete-path-weight}
\end{align}
\end{definition}
These expressions follow by inserting
$\openone=\sum_{x_t\in X}|x_t\rangle\langle x_t|$ at every intermediate
boundary.  Explicitly,
\begin{equation}
\langle y|U_{L:1}|x\rangle
 =\sum_{x_1,\ldots,x_{L-1}}
   \prod_{t=1}^{L}\langle x_t|U_t|x_{t-1}\rangle.
 \label{eq:inserted-context-identities}
\end{equation}
The endpoint-checked route adds amplitudes before applying the Born rule,
whereas the stepwise-checked route adds their classical weights:
\begin{align}
\Born_{\C}(U_{L:1})_{yx}
 &=\left|\sum_{\gamma\in\Omega_{yx}^{(L)}}\mathcal A(\gamma)\right|^2,
 \label{eq:coherent-discrete-path-sum}\\
 \bigl(\Born_{\C}(U_L)\circ\cdots\circ\Born_{\C}(U_1)\bigr)_{yx}
 &=\sum_{\gamma\in\Omega_{yx}^{(L)}}w(\gamma).
 \label{eq:stochastic-discrete-path-sum}
\end{align}
Consequently,
\begin{align}
&\Born_{\C}(U_{L:1})_{yx}
 -\bigl(\Born_{\C}(U_L)\circ\cdots\circ\Born_{\C}(U_1)\bigr)_{yx}
 \nonumber\\
 &\hspace{3em}=
 \sum_{\substack{\gamma,\gamma'\in\Omega_{yx}^{(L)}\\
                   \gamma\ne\gamma'}}
 \mathcal A(\gamma)\overline{\mathcal A}(\gamma')
 =2\operatorname{Re}\!\sum_{\gamma<\gamma'}
 \mathcal A(\gamma)\overline{\mathcal A}(\gamma').
 \label{eq:path-interference-current}
\end{align}
\hspace{0.2em}
Here \(\gamma<\gamma'\) denotes any fixed ordering of the finite
path set; it merely counts each unordered pair once.
Equation~\eqref{eq:path-interference-current} is the path form of the
Born--Chapman--Kolmogorov current.  \ The
quantum-to-stochastic composition square
in Fig.~\ref{fig:quantum-stochastic-composition} commutes
exactly when this total cross-path contribution vanishes for every entrance
\(x\) and endpoint \(y\).  This is an aggregate cancellation condition.
Distinct paths may interfere, and the unitary steps may fail to commute.
Chronology is retained by the time-labeled factors and their endpoints in
Eq.~\eqref{eq:discrete-path-amplitude}, and by the ordered stochastic product
in Eq.~\eqref{eq:stochastic-discrete-path-sum}; scalar multiplication itself
is commutative.

When the stepwise endpoint probability is nonzero, the stepwise path law can
also be postselected on its endpoint:
\begin{equation}
 \mathbb P_{\rm step}(\gamma\mid X_0=x,X_L=y)
 =\frac{|\mathcal A(\gamma)|^2}
        {\displaystyle\sum_{\gamma'\in\Omega_{yx}^{(L)}}
         |\mathcal A(\gamma')|^2}.
 \label{eq:postselected-path-law}
\end{equation}
This is the endpoint-conditioned bridge law of the dephased Markov chain.

Finally, beginning at an internal boundary replaces
$\Omega_{yx}^{(L)}$ and its amplitudes by those of the prescribed suffix,
with a freshly prepared
\(\C\)-diagonal entrance state at that boundary.  Cancellation of the
cross terms for every prefix from the initial boundary therefore need not
imply cancellation for the suffix.  In path language, boundary instability
is the failure of aggregate interference cancellation when the path sum is
restarted at an internal boundary.  Section~\ref{app:not-consistent-histories} contrasts this
aggregate condition with pairwise consistency of fine-grained histories.

\subsection{Exact relation to dephasing and Kolmogorov consistency}
\label{app:ncgd-kolmogorov}

The vanishing-defect pair equation
$\mathcal J_{\C}(V,U)=0$ makes exact contact with work on dynamics that do not
generate and subsequently detect coherence.  Write
$\mathcal U(\rho)=U\rho U^\dagger$ and
$\mathcal V(\rho)=V\rho V^\dagger$.  Then
\begin{equation}
 \Mc\circ\mathcal V\circ\Mc\circ\mathcal U\circ\Mc
 =
 \Mc\circ\mathcal V\circ\mathcal U\circ\Mc
 \quad\Longleftrightarrow\quad
 \Born_{\C}(V)\circ\Born_{\C}(U)
 =\Born_{\C}(V\circ U).
 \label{eq:ncgd-born-equivalence}
\end{equation}
Indeed, after restricting their action to diagonal inputs and reading only
diagonal outputs, the two superoperators have the displayed Born kernels as
their matrices on the basis $\{P_x\}$.  Our Born-kernel equation therefore
recovers precisely the closed-unitary, fixed-context restriction of the
non-coherence-generating-and-detecting (NCGD) condition
\cite{app:SmirneEtAl2019,app:GarciaDiazAccessible}.
For one inserted measurement, this is no-signaling in time resolved
by the entrance outcome, required for every context-basis
preparation~\cite{app:KoflerBrukner,app:ClementeKofler,app:ClementeKofler2015}.
It does not assert that the
intermediate quantum state is undisturbed.

\begingroup
\setlength{\emergencystretch}{2em}

NCGD already imposes the dephasing identity for all triples of
times~\cite{app:GarciaDiazAccessible}; boundary stability specializes that
requirement to a prescribed finite unitary sequence.  Sakuldee, Taranto, and
Milz give a pair of genuinely
mixing qubit unitary steps satisfying the
identity~\cite[Example~6]{app:SakuldeeTarantoMilz}.  Here we classify unitary
solutions and distinguish an anchored family of composing prefixes from
compatibility under independent preparation at internal boundaries.
\par\endgroup

We now prove the schedule-consistency equivalence used by the main
text.

\begin{definition}[Checking schedules and schedule statistics]
\label{def:checking-schedule-app}
For every prescribed contiguous subpassage from boundary $r$ to
boundary $s$, where $0\leq r<s\leq L$, a schedule is
\(\tau=(r=t_0<t_1<\cdots<t_m=s)\).  Its sequential context statistics are
\begin{equation}
 p_\tau(x_{t_1},\ldots,x_{t_m}\mid x_r)
 :=\prod_{j=1}^{m}
 \Born_{\C}(U_{t_j:t_{j-1}+1})_{x_{t_j}x_{t_{j-1}}}.
 \label{eq:schedule-statistics-subpassage}
\end{equation}
The lower index $t_{j-1}+1$ names the first step after boundary
$t_{j-1}$; this shift is required because the $t_j$ label boundaries whereas
the subscripts of $U_{b:a}$ label the included steps.
For fixed entrance outcome $x_r$, this product is a normalized joint Markov
law obtained by reading $\C$ exactly at the scheduled boundaries.  If the
interior time $t_j$ is deleted, Kolmogorov marginal consistency means
\begin{equation}
 \sum_{x_{t_j}}p_\tau(x_{t_1},\ldots,x_{t_m}\mid x_r)
 =p_{\tau\setminus\{t_j\}}
   (x_{t_1},\ldots,\widehat{x}_{t_j},\ldots,x_{t_m}\mid x_r),
 \label{eq:schedule-marginalization}
\end{equation}
where the hat denotes omission.  Every internal-start test uses a fresh
context atom, or by linearity an arbitrary $\C$-diagonal mixture, rather than
the quantum state carried from an earlier passage.
In the rhythmic representation of Section~\ref{sec:quantum-music}, the same
schedule is written as the ordered composition
$(t_1-t_0,\ldots,t_m-t_{m-1})$ of $s-r$. With beat duration $\Delta>0$,
these differences specify successive sound holds. The schedule is
\emph{admissible} when its endpoint kernel equals
$\Born_{\C}(U_{s:r+1})$ for every entrance basis state, hence for every
$\C$-diagonal mixture. This is a condition on the endpoint marginal;
the stronger consistency under all readout insertions and deletions is stated
in Proposition~\ref{prop:schedule-consistency}.
\end{definition}

\Needspace{0.28\textheight}
\begin{proposition}[Boundary stability is full schedule consistency]
\label{prop:schedule-consistency}
For a prescribed unitary sequence in a fixed atomic context, the following
are equivalent:
\begin{enumerate}
 \item the sequence is boundary-stable in the sense of
 Eq.~\eqref{eq:appendix-cue-law};
 \item for every contiguous subpassage independently initialized
 in a \(\C\)-diagonal entrance state, its schedule statistics obey Kolmogorov
 marginal consistency under every insertion or deletion of a prescribed internal
 context readout.
\end{enumerate}
\end{proposition}

\begin{proof}
Fix a contiguous subpassage.
Deleting one readout merges two adjacent unitary blocks.  Boundary stability
identifies the Born kernel of that merged block with the stochastic product of
the two block kernels, so summing over the deleted outcome preserves the
coarser distribution as in Eq.~\eqref{eq:schedule-marginalization}.  Repetition
proves consistency for every schedule and every diagonal entrance mixture.
Conversely, compare on each subpassage the endpoint-checked schedule with the
schedule containing every internal readout.  Marginal consistency equates the
direct interval Born kernel with the product of its one-step Born kernels,
which is Eq.~\eqref{eq:appendix-cue-law}.
\end{proof}

Broader
schedule-consistency settings are studied in
\cite{app:SakuldeeTarantoMilz,app:StrasbergGarciaDiaz}.  Proposition
\ref{prop:schedule-consistency} records the finite-sequence dictionary:
boundary stability is fixed-context schedule consistency on every prescribed
subpassage independently initialized in a \(\C\)-diagonal state, whereas prefix
consistency retains only endpoint-versus-fully-refined comparisons from the
initial boundary.

Boundary stability is operation-supplied divisibility across every
prescribed cut.  In stochastic divisibility, an intermediary need
only exist~\cite{app:Pimenta}; here it must be the Born kernel of the actual
unitary subpassage.  The underlying closed-unitary processes remain
operationally Markovian and completely positive divisible, even when these
population kernels fail to compose~\cite{app:PollockEtAl,app:MilzModi}.
Quantum-channel divisibility is distinct: Wolf and Cirac ask
whether a quantum channel factors through intermediate quantum
channels~\cite{app:WolfCirac}.

\subsection{Born-kernel endpoint composition does not imply consistent
fine-grained histories}
\label{app:not-consistent-histories}

The consistent-histories framework assigns probabilities to a chosen family
of multitime alternatives when its decoherence functional satisfies
appropriate consistency conditions
\cite{app:Griffiths,app:DowkerHalliwell,app:GellMannHartle,app:PazZurek}.
Endpoint Born-kernel composition for a fixed passage is different from, and
strictly weaker than, the following fixed-context fine-grained-history
condition.

\begin{definition}[Fixed-context fine-grained-history consistency]
\label{def:fixed-context-history-consistency-app}
Fix an entrance $x$, an endpoint $y$, and a sequence
$U_1,\ldots,U_L$.  For an intermediate context path
$\alpha=(x_1,\ldots,x_{L-1})$, define
\begin{equation}
 \mathcal A_{\alpha}^{yx}
 :=\langle y|U_LP_{x_{L-1}}U_{L-1}\cdots P_{x_1}U_1|x\rangle,
 \qquad
 \mathfrak D_{\alpha\beta}^{yx}
 :=\mathcal A_{\alpha}^{yx}\overline{\mathcal A_{\beta}^{yx}}.
 \label{eq:history-decoherence-functional}
\end{equation}
Then
\begin{align}
 \Born(U_L\circ\cdots\circ U_1)_{yx}
 &=\sum_{\alpha,\beta}\mathfrak D_{\alpha\beta}^{yx},\nonumber\\
 \bigl(\Born(U_L)\circ\cdots\circ\Born(U_1)\bigr)_{yx}
 &=\sum_{\alpha}\mathfrak D_{\alpha\alpha}^{yx}.
 \label{eq:history-two-sums}
\end{align}
Consequently, endpoint Born kernels compose precisely when
\begin{equation}
 \sum_{\alpha\ne\beta}\mathfrak D_{\alpha\beta}^{yx}=0
 \qquad\text{for every }x,y.
 \label{eq:aggregate-history-cancellation}
\end{equation}
\begin{equation}
 \operatorname{Re}\mathfrak D_{\alpha\beta}^{yx}=0
 \qquad
 \text{for every }x,y\text{ and every }\alpha\ne\beta
 \label{eq:fine-grained-history-consistency}
\end{equation}

is the fixed-context fine-grained specialization of weak
decoherence.  It is the pointwise weak-decoherence condition for this particular
context-path family~\cite{app:Griffiths,app:DowkerHalliwell,app:GellMannHartle,app:PazZurek}.
\end{definition}

The scalar $\mathfrak D_{\alpha\beta}^{yx}$ is the entrance--endpoint-resolved
contribution to the standard decoherence functional.  For the natural
fine-grained context paths used below, medium decoherence requires
\(\mathfrak D_{\alpha\beta}^{yx}=0\) for every distinct pair,
whereas weak decoherence requires
\(\operatorname{Re}\mathfrak D_{\alpha\beta}^{yx}=0\) pairwise.  Either condition implies
the aggregate cancellation in Eq.~\eqref{eq:aggregate-history-cancellation};
the converse need not hold.

\begin{center}
\small
\begin{tabularx}{0.94\textwidth}{@{}l>{\raggedright\arraybackslash}Xl@{}}
\toprule
\textbf{Condition} & \textbf{What must vanish for each fixed $x,y$} & \textbf{Strength}\\
\midrule
medium decoherence & every off-diagonal $\mathfrak D_{\alpha\beta}^{yx}$ & strongest\\
weak decoherence & every $\operatorname{Re}\mathfrak D_{\alpha\beta}^{yx}$ & intermediate\\
endpoint Born-kernel composition & only their aggregate sum & weakest\\
\bottomrule
\end{tabularx}
\end{center}

\Needspace{12\baselineskip}
\begin{proposition}[Endpoint composition without weak decoherence]
For \(\sigma\in\{\sigma_x,\sigma_y,\sigma_z\}\), let
\(R_\sigma(\theta)=e^{-\ii\theta\sigma/2}\).  Choose
$\varphi\in(0,\pi/2)$ with $\tan\varphi=2$, and set
\begin{equation}
 U_1=R_{\sigma_y}(\pi/3),
 \qquad
 U_2=R_{\sigma_x}(2\pi/3),
 \qquad
 U_3=R_{\sigma_y}(-\pi/3)\circ R_{\sigma_z}(-\varphi).
 \label{eq:history-separation-witness}
\end{equation}
Both prefixes obey the fully refined Born-kernel identity:
\begin{align}
 \Born(U_2\circ U_1)
 &=\Born(U_2)\circ\Born(U_1)=K(-1/4),\nonumber\\
 \Born(U_3\circ U_2\circ U_1)
 &=\Born(U_3)\circ\Born(U_2)\circ\Born(U_1)=K(-1/8).
 \label{eq:history-separation-prefixes}
\end{align}
Nevertheless, the fine-grained histories of the complete three-step prefix
are not weakly decoherent.  Beginning instead at the internal boundary
before $U_2$ destroys the aggregate cancellation:
\begin{equation}
 \mathcal J_{\C}(U_3,U_2)
 =\frac{3\sqrt5}{20}
 (\sigma_x-\openone)
 \ne0.
 \label{eq:history-separation-current}
\end{equation}
Thus the witness establishes endpoint composition for the complete prefix.
\par\noindent

The sequence is not boundary-stable and therefore does not supply a
genealogy on all contiguous intervals.
\end{proposition}

\begin{proof}
The individual kernels are
$\Born(U_1)=K(1/2)$, $\Born(U_2)=K(-1/2)$, and
$\Born(U_3)=K(1/2)$.  Direct multiplication gives
Eq.~\eqref{eq:history-separation-prefixes}.  For the first two steps and
$x=y=0$, the two path amplitudes satisfy
\begin{equation}
 \mathfrak D_{01}^{00}=\frac{3\ii}{16}.
\end{equation}
Thus medium decoherence already fails, although its real part vanishes.

For the complete prefix, order the four intermediate paths as
$(0,0),(0,1),(1,0),(1,1)$.  At $x=y=0$, the real off-diagonal part of the
decoherence matrix is
\begin{equation}
 \operatorname{Re}\mathfrak D_{\mathrm{off}}^{00}
 =\frac{\sqrt5}{320}
 \begin{pmatrix}
  0&-18&0&3\\
  -18&0&9&0\\
  0&9&0&6\\
  3&0&6&0
 \end{pmatrix}.
 \label{eq:history-separation-matrix}
\end{equation}
The matrix has nonzero entries, so pairwise weak decoherence fails, but the
sum of all its entries is zero.  The other three entrance--endpoint pairs
have the same cancellation with the corresponding sign changes, which is
equivalent to the second identity in
Eq.~\eqref{eq:history-separation-prefixes}.  For the suffix, direct
evaluation gives Eq.~\eqref{eq:history-separation-current}.
\end{proof}

The prefix composition
identities in this witness therefore come from interference balance,
not from absent interference.
Cue instability is the loss of that balance when the remaining
passage is restarted at an internal boundary.

\section{Quantitative bounds}
\label{app:quantitative-bounds}

We quantify schedule information and the composition current, then bound
the number of full-support steps in a boundary-stable sequence.

\subsection{Conditional information carried by the checking schedule}
\label{app:schedule-information}

\begin{definition}[Checking-schedule information]
\label{def:schedule-information-app}
Fix a contiguous passage containing steps $a$ through $b$, where
$1\leq a\leq b\leq L$, and define its coherent and fully checked endpoint
kernels by
\begin{equation}
 \mathsf P_{b:a}:=\Born_{\C}(U_{b:a}),
 \qquad
 \mathsf Q_{b:a}:=
 \Born_{\C}(U_b)\circ\cdots\circ\Born_{\C}(U_a).
 \label{eq:schedule-information-kernels}
\end{equation}
Let $X_{a-1}\sim\pi$ be a recorded entrance outcome
and let an independent fair
bit $S$ select the protocol: $S=0$ selects the coherent passage, and $S=1$
selects the passage with unread context checks at every prescribed internal
boundary.  The final recorded outcome $X_b$ then has joint law
\begin{equation}
 \Pr(S=\sigma,X_{a-1}=x,X_b=y)
 =\frac12\pi(x)(\mathsf K_\sigma)_{yx},
 \qquad
 \mathsf K_0=\mathsf P_{b:a},\quad
 \mathsf K_1=\mathsf Q_{b:a}.
 \label{eq:schedule-information-experiment}
\end{equation}

For a discrete random variable \(A\) with law \(p_A\), write
\(H_2(A):=-\sum_a p_A(a)\log_2p_A(a)\), with \(0\log_2 0:=0\).  For
discrete (B,C), define
\(H_2(B\mid C):=\sum_c\Pr(C=c)H_2(B\mid C=c)\).  The base-two conditional
mutual information is
\[
 I(A;B\mid C):=H_2(B\mid C)-H_2(B\mid A,C).
\]
It measures how much learning \(A\) reduces the uncertainty of \(B\) once
\(C\) is known.  We also write
\(I(A;B):=H_2(B)-H_2(B\mid A)\).

Equivalently, the base-two Shannon entropy of a probability vector
$p$
\cite{app:Shannon1948} is
\begin{equation}
 H_2(p):=-\sum_y p_y\log_2p_y,
 \qquad 0\log_2 0:=0.
 \label{eq:shannon-entropy-app}
\end{equation}
Define the base-two Jensen--Shannon divergence
\cite{app:Lin1991} by
\begin{equation}
 \operatorname{JS}_2(p,q):=
 H_2\!\left(\frac{p+q}{2}\right)
 -\frac12H_2(p)-\frac12H_2(q).
 \label{eq:js-definition-app}
\end{equation}
The \emph{Born--Chapman--Kolmogorov schedule information} of the
passage is
\begin{equation}
 \mathcal I_{\C,\pi}^{\mathrm{BCK}}(b:a)
 :=I(S;X_b\mid X_{a-1}).
 \label{eq:schedule-information-definition-app}
\end{equation}
Here $X_{a-1}$ is the boundary outcome immediately before the
first included step $U_a$, while $X_b$ is the outcome immediately after the
last included step $U_b$.
\end{definition}

\begin{sloppypar}
\begin{proposition}[Conditional schedule information]
\label{prop:schedule-information}
For the experiment in Eq.~\eqref{eq:schedule-information-experiment},
\begin{equation}
 \mathcal I_{\C,\pi}^{\mathrm{BCK}}(b:a)
 =\sum_x\pi(x)\,
 \operatorname{JS}_2\!\left(
   (\mathsf P_{b:a})_{\cdot x},
   (\mathsf Q_{b:a})_{\cdot x}
 \right),
 \label{eq:schedule-information-identity}
\end{equation}
and $0\leq\mathcal I_{\C,\pi}^{\mathrm{BCK}}(b:a)\leq1$ bit per run.
If $\pi(x)>0$ for every $x$, then
\begin{equation}
 \mathcal I_{\C,\pi}^{\mathrm{BCK}}(b:a)=0
 \quad\Longleftrightarrow\quad
 \mathsf P_{b:a}=\mathsf Q_{b:a}.
 \label{eq:schedule-information-zero}
\end{equation}
Consequently, for any fixed strictly positive $\pi$, a prescribed sequence is
boundary-stable exactly when this conditional information vanishes on every
contiguous passage.  Prefix consistency requires the same condition only for
passages with $a=1$.

For the uniform entrance distribution,
\begin{equation}
 \mathcal I_{\C,\mathrm{unif}}^{\mathrm{BCK}}(b:a)
 \geq
 \frac{\|\mathsf P_{b:a}-\mathsf Q_{b:a}\|_{\mathrm F}^{2}}
      {8d\ln2}.
 \label{eq:schedule-information-pinsker}
\end{equation}
For two steps, with $U_1=U$ and $U_2=V$, this becomes
\begin{equation}
\mathcal I_{\C,\mathrm{unif}}^{\mathrm{BCK}}(2:1)
 \geq
 \frac{d-1}{8d\ln2}\nu_{\C}(V,U)^2.
 \label{eq:schedule-information-current-bound}
\end{equation}
\end{proposition}
\end{sloppypar}

\begin{proof}
Conditioned on $X_{a-1}=x$ but not on $S$, the endpoint law is
$m_x=[(\mathsf P_{b:a})_{\cdot x}+(\mathsf Q_{b:a})_{\cdot x}]/2$.
Expanding
$I(S;X_b\mid X_{a-1})
=H_2(X_b\mid X_{a-1})
-H_2(X_b\mid S,X_{a-1})$
gives Eq.~\eqref{eq:schedule-information-identity}.  Strict concavity of
Shannon entropy makes each Jensen--Shannon term nonnegative and zero exactly
when its two arguments agree.  Strict
positivity of $\pi$ therefore gives
Eq.~\eqref{eq:schedule-information-zero}, whose all-interval version is
Eq.~\eqref{eq:appendix-cue-law}.  Finally, Pinsker's inequality gives
\begin{equation}
 \operatorname{JS}_2(p,q)
 \geq\frac{\|p-q\|_1^2}{8\ln2}
 \geq\frac{\|p-q\|_2^2}{8\ln2}.
 \label{eq:js-pinsker-app}
\end{equation}
Averaging uniformly over the columns proves
Eq.~\eqref{eq:schedule-information-pinsker}; the two-step specialization uses
Eqs.~\eqref{eq:defect} and \eqref{eq:metric}.
\end{proof}

The conditioning on the entrance record is essential.  If a uniform entrance
is sampled and then discarded, bistochasticity makes both endpoint marginals
uniform, so $I(S;X_b)=0$ can hold even when
$\mathsf P_{b:a}\ne\mathsf Q_{b:a}$.  This is classical Shannon information
of the recorded variables
$(S,X_{a-1},X_b)$, not an entropy of the signed current,
a count of hidden physical bits, or a generic measure of quantum
non-Markovianity.

Coherence-mediated information return under classical reduction is
established~\cite{app:BenattiChruscinskiNichele}; coherent-versus-dephased
population discrepancies have also been quantified by trace
distance~\cite{app:GarciaDiazAccessible}.  The quantity above gives the
conditional information about the chosen checking schedule.

\Needspace{3\baselineskip}
The fair bit fixes the ordinary Jensen--Shannon normalization.  More
generally, if $\Pr(S=1)=\alpha\in(0,1)$, then the same proof replaces
Eq.~\eqref{eq:js-definition-app} by the weighted divergence
\begin{equation}
 \operatorname{JS}_{2,\alpha}(p,q)
 :=H_2((1-\alpha)p+\alpha q)
   -(1-\alpha)H_2(p)-\alpha H_2(q),
 \label{eq:weighted-js-app}
\end{equation}
without changing the zero criterion.  This biased form applies when a
quantum computation supplies the schedule bit in
Section~\ref{app:quantum-musical-prediction}.

Retaining the final quantum system $R$ before the context
measurement gives
\[
 I_c:=I(S;X_b\mid X_{a-1}),
 \qquad
 I_q:=I(S;R\mid X_{a-1}),
\]
where $I_q$ uses von Neumann entropies with base-two logarithms.
For the independent fair checking bit $S$,
\[
 \boxed{0\leq I_c\leq I_q\leq1\ \text{bit}.}
\]
These inequalities follow from nonnegativity, data processing under the
final measurement, and the entropy bound for a classical bit, respectively.

For the pair in Eq.~\eqref{eq:qubit-distinct-pair}, conditioned on
either entrance label, the coherent density operator is pure, the checked
density operator has eigenvalues $3/4,1/4$, and their equally weighted
average has eigenvalues $(4\pm\sqrt7)/8$.  Hence
\[
 I_q=h_2\!\left(\frac{4+\sqrt7}{8}\right)
 -\frac12h_2\!\left(\frac34\right)
 \simeq0.2504\ \text{bits},
\]
where $h_2(t)=-t\log_2t-(1-t)\log_2(1-t)$.
The common final unitary $V$ leaves $I_q$ unchanged.

\subsection[A sharp Frobenius bound]{A sharp Frobenius bound}
\label{app:bound}

\begin{definition}[Uniform mode and complex Hadamard unitaries]
\label{def:complex-hadamard-app}
Here, set
$d=|X|=|\C|=\dim\mathsf H$.  Define the normalized equal-amplitude context ket
\[
 |\boldsymbol{+}_d\rangle
 :=\frac{1}{\sqrt d}\sum_{x\in X}|x\rangle,
 \qquad
 J_d:=d|\boldsymbol{+}_d\rangle\langle\boldsymbol{+}_d|,
\]
so that $J_d$ is the $d\times d$ all-ones matrix.  The normalized Frobenius
current introduced in
Eq.~\eqref{eq:metric} has the sharp bound
\begin{equation}
 0\leq \nu_{\C}(V,U)\leq 1.
 \label{eq:metric-bound}
\end{equation}
It measures the aggregate signed difference of the endpoint-checked and
stepwise-checked kernels.  Its largest possible value is attained.

A unitary $F\in\U(d)$ is a \emph{complex Hadamard unitary} when every entry
has modulus $1/\sqrt d$.  Equivalently, its Born kernel is the uniform kernel
$J_d/d$.  With outcomes labeled $0,\ldots,d-1$, the Fourier matrix
\[
 F_d:=\frac1{\sqrt d}\sum_{x,y=0}^{d-1}
 e^{2\pi\ii xy/d}|x\rangle\langle y|
\]
is an example in every dimension.
\end{definition}

\begin{theorem}[Sharp Born-composition bound]
\label{thm:bound-main}
For every $d\ge2$,
\begin{equation}
 \sup_{U,V\in\U(d)}\nu_{\C}(V,U)=1.
 \label{eq:sharp-bound-main}
\end{equation}
Every \ inverse complex-Hadamard pair
$U=F$, $V=F^\dagger$ attains the bound.
\end{theorem}

\begin{proof}[Proof of Theorem~\ref{thm:bound-main}]
Equation~\eqref{eq:compression-identity} gives
\begin{equation}
 \mathcal J_{\C}(V,U)
 =\iota_{\C}^\dagger L_VQ_{\C}L_U\iota_{\C}.
 \label{eq:frobenius-compression}
\end{equation}
Here $\iota_{\C}$ is an isometry, $L_U$ and $L_V$ are Hilbert--Schmidt
unitaries, and $Q_{\C}$ is an orthogonal projection.

By Remark~\ref{rem:standard-notation-app}, these contractions have
operator norm one, and hence
\begin{equation}
 \left\|\mathcal J_{\C}(V,U)\right\|_{\mathrm{op}}\leq1.
 \label{eq:current-op-bound}
\end{equation}
Both $K_{\rm end}$ and $K_{\rm step}$ in
Eq.~\eqref{eq:endpoint-stepwise} are bistochastic.  Hence the normalized
equal-amplitude ket and its dual bra are respectively right and left null
vectors of their difference:
\begin{equation}
 \mathcal J_{\C}(V,U)|\boldsymbol{+}_d\rangle=0,
 \qquad
 \langle\boldsymbol{+}_d|\mathcal J_{\C}(V,U)=0.
 \label{eq:current-null-vectors}
\end{equation}
Consequently $\operatorname{rank}\mathcal J_{\C}(V,U)\leq d-1$.  Summing the
squares of the singular values gives the standard rank--norm inequality
\begin{equation}
 \left\|\mathcal J_{\C}(V,U)\right\|_{\mathrm F}^2
 \leq\operatorname{rank}\mathcal J_{\C}(V,U)
 \left\|\mathcal J_{\C}(V,U)\right\|_{\mathrm{op}}^2
 \leq d-1.
\end{equation}
This proves the upper bound.

Let $F$ be complex Hadamard and take $U=F$, $V=F^\dagger$.  Then
\begin{equation}
 \Born(V\circ U)=\openone,
 \qquad
 \Born(U)=\Born(V)=\frac{J_d}{d},
 \qquad
 \mathcal J_{\C}(V,U)=\openone-\frac{1}{d}J_d.
\end{equation}
\noindent
The final matrix is the orthogonal projection onto the $(d-1)$-dimensional
subspace orthogonal to $|\boldsymbol{+}_d\rangle$.  Its Frobenius norm is therefore
$\sqrt{d-1}$, so $\nu_{\C}(F^\dagger,F)=1$.
\end{proof}

\noindent\par
\subsection{Why a boundary-stable qubit sequence has only two genuinely mixing steps}
\noindent

We use the defined qubit term ``genuinely mixing'' throughout.

\begin{definition}[Bloch rotation and cumulative axes]
\label{def:cumulative-bloch-axes-app}
For a qubit unitary $U$, let $R_U\in\mathrm{SO}(3)$ be its Bloch-sphere
rotation, defined by
\begin{equation}
 U^\dagger(\boldsymbol a\boldsymbol\cdot\boldsymbol\sigma)U
 =(R_U^{\mathsf T}\boldsymbol a)\boldsymbol\cdot\boldsymbol\sigma,
 \qquad
 \boldsymbol\sigma=(\sigma_x,\sigma_y,\sigma_z),
 \qquad
 \boldsymbol z=(0,0,1)^{\mathsf T}.
 \label{eq:bloch-rotation-convention}
\end{equation}
Taking the expectation in $|0\rangle$ gives
\begin{equation}
 \beta(U)=\boldsymbol z\cdot R_U^{\mathsf T}\boldsymbol z,
 \label{eq:bloch-polarization}
\end{equation}
because $\langle0|\boldsymbol a\cdot\boldsymbol\sigma|0\rangle
=\boldsymbol z\cdot\boldsymbol a$ for every
$\boldsymbol a\in\mathbb R^3$.  Define the cumulative Bloch axes
\begin{equation}
 \boldsymbol b_0:=\boldsymbol z,
 \qquad
 \boldsymbol b_k:=R_{U_{k:1}}^{\mathsf T}\boldsymbol z\in\R^3,
 \quad 1\leq k\leq L.
 \label{eq:bloch-axis-vectors}
\end{equation}
Thus $\boldsymbol b_0$ is the initial context axis, and
$\boldsymbol b_k$ is that axis pulled back through the first $k$ steps.  For
$1\leq a\leq b\leq L$,
\begin{equation}
 \beta(U_{b:a})
 =\boldsymbol b_b\cdot\boldsymbol b_{a-1},
 \label{eq:segment-polarization}
\end{equation}
and
\begin{equation}
 \beta_k=
 \boldsymbol b_k\cdot\boldsymbol b_{k-1}.
\end{equation}
\end{definition}
Boundary stability is therefore equivalent to the Gram conditions
\begin{equation}
 \boldsymbol b_b\cdot\boldsymbol b_{a-1}
 =\prod_{j=a}^{b}\beta_j
 \qquad(1\leq a\leq b\leq L).
 \label{eq:markov-gram}
\end{equation}

\begin{definition}[Qubit frame innovation]
\label{def:qubit-frame-innovation-app}
For \(1\leq k\leq L\), the qubit frame innovation at step \(k\) is
\begin{equation}
 \boldsymbol e_k:=
 \boldsymbol b_k-\beta_k\boldsymbol b_{k-1}.
 \label{eq:qubit-frame-innovation-app}
\end{equation}
It is the component of the new cumulative Bloch axis not predicted by the
immediately preceding axis and its one-step polarization.
\end{definition}

\begin{theorem}[Qubit boundary-stability bound]
\label{thm:qubit-boundary-stability-app}
Every boundary-stable qubit sequence contains at most two genuinely mixing
steps.
\end{theorem}

\begin{proof}

For the innovations in
Definition~\ref{def:qubit-frame-innovation-app},
\begin{equation}
 \boldsymbol e_k=
 \boldsymbol b_k-\beta_k\boldsymbol b_{k-1}.
\end{equation}
For every $0\le r\le k-1$,
Eq.~\eqref{eq:markov-gram} gives
\begin{align}
 \boldsymbol e_k\cdot\boldsymbol b_r
 &=
   \boldsymbol b_k\cdot\boldsymbol b_r
   -\beta_k\boldsymbol b_{k-1}\cdot\boldsymbol b_r\\
 &=0.
\end{align}
Also
\begin{equation}
 \|\boldsymbol e_k\|^2=1-\beta_k^2.
\end{equation}
Thus every genuinely mixing step, characterized by $|\beta_k|<1$, adds a
new nonzero direction orthogonal to the span of all previous
$\boldsymbol b_r$.  Starting from the one-dimensional span of
$\boldsymbol b_0$ and remaining inside
$\R^3$, there is room for at most two
such innovations.  Hence at most two steps can genuinely mix.
\end{proof}

Equivalently, the Gram matrix in Eq.~\eqref{eq:markov-gram} is the correlation
matrix of the signed two-state population mode and has rank
\begin{equation}
 1+\left|\{k:|\beta_k|<1\}\right|,
\end{equation}
which cannot exceed three for Bloch vectors.

\subsection{Full-support capacity and prime-dimensional saturation}
\label{app:prime-count}

This subsection proves the full-support step-count bound; primality is
used only to attain it.
By Definition~\ref{def:support-mixing-app}, a step \(U_k\) has full
support precisely when
\begin{equation}
 \Born(U_k)_{yx}>0
 \qquad\text{for every }x,y.
 \label{eq:full-support-note}
\end{equation}
For qubits, full support is exactly the defined
genuinely mixing condition.
The upper bound holds in every finite dimension; primality enters only in
the sharp construction below.

\subsubsection{Context frames and dimension accounting}
\noindent

\begin{definition}[Population-deviation space and context frames]
\label{def:context-frames-app}
Let
\begin{equation}
 \mathcal R_d^0
 :=\left\{v\in\R^d:\sum_xv_x=0\right\}
 \label{eq:population-deviation-space}
\end{equation}
be the $(d-1)$-dimensional space of population deviations from the uniform
distribution.  \Needspace{5\baselineskip}
Embed it isometrically in the real Hilbert space
$\mathsf{Herm}_0(d)$ of traceless Hermitian operators by
\begin{equation}
\operatorname{Diag}_{\C}(v):=\sum_xv_x|x\rangle\langle x|.
 \label{eq:diagonal-embedding}
\end{equation}
Here $|x\rangle$ denotes the $x$th \ fixed context-basis vector, so
$|x\rangle\langle x|$ is the matrix unit with its only nonzero entry in
position $(x,x)$; thus $\operatorname{Diag}_{\C}(v)$ is simply the usual
diagonal matrix with diagonal $v$.
The inner products are Euclidean and Hilbert--Schmidt, respectively, and
\begin{equation}
 \dim_{\R}\mathsf{Herm}_0(d)=d^2-1.
 \label{eq:traceless-dimension}
\end{equation}

Define the context-frame isometries of the cumulative passages by
\begin{equation}
 \begin{aligned}
  &F_k:\mathcal R_d^0\longrightarrow\mathsf{Herm}_0(d),\qquad
    0\leq k\leq L,\\
  &F_0(v):=\operatorname{Diag}_{\C}(v),\\
  &F_k(v):=U_{k:1}^\dagger
    \operatorname{Diag}_{\C}(v)U_{k:1},\qquad
    1\leq k\leq L.
 \end{aligned}
 \label{eq:cumulative-context-frames}
\end{equation}
The map $\operatorname{Diag}_{\C}$ is an isometry and unitary conjugation
preserves the Hilbert--Schmidt inner product, so
$F_k$ identifies a zero-sum population deviation with the
corresponding fixed-context diagonal observable pulled back through the first
$k$ steps.
\end{definition}

For $1\leq a\leq b\leq L$, direct
evaluation of the Hilbert--Schmidt inner
product gives, for $v,w\in\mathcal R_d^0$,
\begin{equation}
 \begin{aligned}
  &\langle F_b(v),F_{a-1}(w)\rangle_{\HS}
    =\sum_{x,y}v_xw_y
    \left|\langle x|U_{b:a}|y\rangle\right|^2\\
  &=\langle v,\Born(U_{b:a})w\rangle_2.
 \end{aligned}
 \label{eq:frame-overlap-calculation}
\end{equation}
Equivalently,
\begin{equation}
 F_b^*\circ F_{a-1}
 =\left.\Born(U_{b:a})\right|_{\mathcal R_d^0}.
 \label{eq:frame-born-overlap}
\end{equation}
\begin{definition}[One-step frame transition]
\label{def:frame-transition-app}
Here $F_b^*$ is the adjoint between real
inner-product spaces.  Put
\begin{equation}
 T_k:=F_k^*\circ F_{k-1}
 =\left.\Born(U_k)\right|_{\mathcal R_d^0}.
 \label{eq:one-note-restriction}
\end{equation}
Boundary stability preserves the zero-sum population subspace because every Born
kernel is bistochastic.  Thus $T_k$ is the signed-population restriction of a
stochastic kernel; it is a linear contraction on $\mathcal R_d^0$, not itself
a probability transition on a simplex.
\end{definition}
Boundary stability is therefore precisely the block Gram rule
\begin{equation}
 F_b^*\circ F_{a-1}=T_b\circ T_{b-1}\circ\cdots\circ T_a
 \qquad(1\leq a\leq b\leq L).
 \label{eq:block-gram-rule}
\end{equation}

\begin{definition}[Context-frame innovation]
\label{def:context-frame-innovation-app}
The \emph{context-frame innovation} is the part of the next frame not
predicted by its immediate predecessor:
\begin{equation}
 E_k:=F_k-F_{k-1}\circ T_k^*.
 \label{eq:innovation-map}
\end{equation}
Thus $E_k(v)$ is the residual diagonal observable after the preceding
context frame has accounted for the linear prediction $T_k^*v$.  Its
\emph{innovation rank} is
\begin{equation}
 r_k^{\rm inn}:=\rank(E_k).
 \label{eq:innovation-rank-definition-app}
\end{equation}
\end{definition}

For every
$0\le r\le k-1$,
Eq.~\eqref{eq:block-gram-rule} yields
\begin{align}
 E_k^*\circ F_r
   =F_k^*\circ F_r-T_k\circ F_{k-1}^*\circ F_r\nonumber\\
 &=T_k\circ T_{k-1}\circ\cdots\circ T_{r+1}
   -T_k\circ T_{k-1}\circ\cdots\circ T_{r+1}=0.
 \label{eq:innovation-earlier-orthogonal}
\end{align}
Thus every innovation is orthogonal to all earlier context frames.  Since
\begin{equation*}
 \operatorname{ran}E_j\subseteq
 \operatorname{ran}F_j+\operatorname{ran}F_{j-1},
\end{equation*}
the ranges of
$E_1,\ldots,E_L$ are mutually orthogonal and are orthogonal to
$\operatorname{ran}F_0$.  Expanding
the definition and using
\begin{equation*}
 \begin{aligned}
  &F_k^*F_k=F_{k-1}^*F_{k-1}=\openone,\\
  &F_k^*F_{k-1}=T_k,
 \end{aligned}
\end{equation*}
gives
\begin{equation}
 E_k^*\circ E_k=\openone-T_k\circ T_k^*.
 \label{eq:innovation-gram}
\end{equation}

Every bistochastic matrix is a Euclidean contraction by
Eq.~\eqref{eq:bistochastic-contraction}.
Hence $\openone-T_k\circ T_k^*$ is positive semidefinite.
\hspace{0.25em}
Equation~\eqref{eq:innovation-gram} gives
\begin{equation}
 r_k^{\rm inn}=\rank(\openone-T_k\circ T_k^*)=\rank(E_k).
 \label{eq:innovation-rank}
\end{equation}
The original frame
$\operatorname{ran}F_0$ consumes
$d-1$ dimensions of
$\mathsf{Herm}_0(d)$, and its mutually orthogonal innovations consume
$\sum_kr_k^{\rm inn}$ more.  Therefore
\begin{equation}
 (d-1)+\sum_{k=1}^{L}r_k^{\rm inn}\le d^2-1,
 \label{eq:ambient-dimension-count}
\end{equation}
and hence
\begin{equation}
 \sum_{k=1}^{L}r_k^{\rm inn}\le d(d-1).
 \label{eq:weighted-note-count}
\end{equation}

The same innovation count has an information-theoretic form.
\begin{definition}[Mean conditional linear entropy]
\label{def:mean-linear-entropy-app}
For a bistochastic kernel $K$, define its mean conditional linear entropy by
\begin{equation}
 \ell_2(K)
 :=\frac1d\sum_x\left(1-\sum_yK_{yx}^2\right)
 =1-\frac{\|K\|_{\mathrm F}^2}{d}.
 \label{eq:kernel-linear-entropy-app}
\end{equation}
This is a linear, or Tsallis--2, entropy.  Equivalently, it is one minus the
collision probability, averaged over uniformly sampled entrances.  It is
dimensionless.
\end{definition}

\begin{sloppypar}
\begin{proposition}[Finite linear-entropy budget]
\label{prop:linear-entropy-budget-app}
Every boundary-stable $d$-dimensional sequence obeys
\begin{equation}
 \sum_{k=1}^{L}\ell_2\!\left(\Born(U_k)\right)\le d-1.
 \label{eq:linear-entropy-budget-app}
\end{equation}
The prime-dimensional construction in
Theorem~\ref{thm:sharp-prime-classification-app} attains equality.
\end{proposition}
\end{sloppypar}

\begin{proof}
The uniform mode contributes one to $\|\Born(U_k)\|_{\mathrm F}^2$, and its
restriction to the orthogonal zero-sum sector is $T_k$.  Therefore
Eqs.~\eqref{eq:innovation-gram} and \eqref{eq:kernel-linear-entropy-app} give
\begin{equation}
 d\,\ell_2\!\left(\Born(U_k)\right)
 =\Tr(\openone-T_kT_k^*)
 =\|E_k\|_{\mathrm F}^2
 \leq r_k^{\rm inn}.
 \label{eq:linear-entropy-innovation-app}
\end{equation}
The last inequality holds because the positive contraction
$\openone-T_kT_k^*$ has eigenvalues in $[0,1]$.  Summing and using
Eq.~\eqref{eq:weighted-note-count} proves the bound.  In the prime
construction, every one-step kernel is $J_d/d$, so every term equals
$1-1/d$ and the $d$ terms sum to $d-1$.
\end{proof}

If $K=\Born(U_k)$ has full support, then $KK^{\mathsf T}$ is a strictly
positive bistochastic matrix.  The Perron--Frobenius theorem makes its
eigenvalue $1$ simple, with the uniform vector as its eigendirection.  The
orthogonal complement of that vector is precisely $\mathcal R_d^0$, so every
singular value of
$T_k=K|_{\mathcal R_d^0}$ is strictly smaller than one.  It follows that
$\openone-T_k\circ T_k^*$ is positive definite on $\mathcal R_d^0$ and
\begin{equation}
 r_k^{\rm inn}=d-1.
 \label{eq:full-support-full-innovation}
\end{equation}
Each full-support step therefore consumes $d-1$ dimensions of the budget
in Eq.~\eqref{eq:weighted-note-count}; at most $d$ such steps can occur.

\Needspace{19\baselineskip}
\noindent\par
\subsubsection{Sharp full-support capacity theorem}

\noindent\par
\begin{theorem}[Sharp full-support capacity]
\label{thm:sharp-prime-classification-app}
Let $d\ge2$ and let
$U_1,\ldots,U_L$ be unitary operations on $\mathbb C^d$.  Suppose every
contiguous interval obeys the Born-kernel composition law:
\begin{equation}
 \begin{aligned}
  \Born(U_{b:a})
  &=\Born(U_b)\circ\cdots\circ\Born(U_a),\\[-1pt]
  &\hspace{1em}(1\leq a\leq b\leq L).
 \end{aligned}
 \label{eq:prime-classification-hypothesis}
\end{equation}
Then at most $d$ one-step kernels \(\Born(U_k)\) have full support.
For every prime $d$, this bound is attained by a boundary-stable sequence of
exactly $d$ full-support steps, for which every nonempty interval satisfies
\begin{equation}
 \begin{aligned}
  \Born(U_{b:a})
  &=\Born(U_b)\circ\cdots\circ\Born(U_a)\\[2pt]
  &=|\boldsymbol{+}_d\rangle\langle\boldsymbol{+}_d|,\\[-1pt]
  &\hspace{1em}(1\leq a\leq b\leq d).
 \end{aligned}
 \label{eq:prime-classification-conclusion}
\end{equation}
where $|\boldsymbol{+}_d\rangle$ is the normalized equal-amplitude context ket.
\end{theorem}
Here $d$ is the Hilbert-space dimension, whereas $L$ is the independent
number of prescribed steps.  In the saturating construction $L=d$, so its
$d$ steps have $d+1$ boundary basis frames.

\begin{proof}
If $m$
one-step kernels have full support, then
\(m(d-1)\le\sum_kr_k^{\rm inn}\le d(d-1)\), hence \(m\le d\).

For prime $d$, to attain the bound, choose a complete family of $d+1$
mutually unbiased
bases~\cite{app:WoottersFields}, beginning with the computational basis.  Write
$V_1,V_2,\ldots,V_{d+1}$ for their unitary basis maps, where
$V_a|x\rangle$ is the $x$th vector of basis $a$ and $V_1=\openone$.  Thus,
whenever $a\ne b$,
\begin{equation}
 \left|\langle x|V_a^\dagger V_b|y\rangle\right|^2=\frac1d
 \qquad\text{for every }x,y.
 \label{eq:mub-overlap}
\end{equation}
Use these bases as cumulative basis frames:
\begin{align}
 U_k
 &:=
 V_{k+1}^\dagger\circ V_k,\nonumber\\
 U_{k:1}
 &=V_{k+1}^\dagger,
 \qquad 1\le k\le d.
 \label{eq:mub-notes}
\end{align}
For every nonempty contiguous interval, the intermediate factors telescope,
\begin{equation}
 U_{b:a}=
 V_{b+1}^\dagger\circ V_a,
 \qquad
 1\leq a\leq b\leq d,
 \label{eq:mub-telescope}
\end{equation}
and the distinct endpoint bases are mutually unbiased.  Therefore
\begin{equation}
 \Born(U_{b:a})=|\boldsymbol{+}_d\rangle\langle\boldsymbol{+}_d|.
 \label{eq:mub-interval-kernel}
\end{equation}
Every one-step kernel is the same projector, which is idempotent:
\begin{equation}
 \left(|\boldsymbol{+}_d\rangle\langle\boldsymbol{+}_d|\right)^m
 =|\boldsymbol{+}_d\rangle\langle\boldsymbol{+}_d|
 \qquad(m\ge1).
 \label{eq:uniform-idempotence}
\end{equation}
The Born composition law consequently holds on every interval.  All $d$
steps have full support, so the upper bound is attained.
\end{proof}

For completeness, when $d$ is an odd prime, let $\mathbb F_d$ be the finite
field with $d$ elements and write
$\omega=e^{2\pi\ii/d}$.  Besides the computational basis, one may take the
$d$ bases
\begin{equation}
 |a,b\rangle
 :=\frac1{\sqrt d}\sum_{x\in\mathbb F_d}
   \omega^{ax^2+bx}|x\rangle,
 \qquad a,b\in\mathbb F_d.
 \label{eq:prime-mub-bases}
\end{equation}
For fixed $a$ these vectors are orthonormal, while distinct values of $a$
give quadratic Gauss sums of magnitude $\sqrt d$.  For $d=2$, the
eigenbases of $\sigma_z$, $\sigma_x$, and $\sigma_y$ supply the three mutually unbiased
bases and hence two steps.

\section{The analytic Walsh three-step circuit}
\label{app:walsh}

This section gives the full verification of
the analytic circuit and the composition identities
stated locally in Theorem~\ref{thm:walsh-composition-app}.
For self-contained checking, define
\begin{equation}
 H=\frac{\sigma_x+\sigma_z}{\sqrt2},
 \qquad W=H\otimes H,
 \qquad
 \Phi(a,b,c)=\operatorname{diag}
 \bigl(1,e^{\ii\pi a},e^{\ii\pi b},e^{\ii\pi c}\bigr),
 \label{eq:walsh-gates-app}
\end{equation}
and let $C_{\rm NOT}$ be CNOT with the first qubit controlling the second.
For real phase parameters $\xi,\eta,\zeta$, put
\begin{align}
 Q_1&=W\circ\Phi(\xi,0,\xi)\circ W
       \circ\Phi(1/2,0,1/2)\circ W,\nonumber\\
 Q_2&=W\circ\Phi(-1/2,1/2,-1)\circ W
       \circ\Phi(1/2-\xi,1/2,-\xi)\circ W,\nonumber\\
 Q_3&=W\circ\Phi(1/2,0,1/2)\circ W
       \circ\Phi(\eta,\zeta,\zeta-\eta)\circ W,\nonumber\\
 U_1&=C_{\rm NOT}\circ Q_1,\qquad
 U_2=C_{\rm NOT}\circ Q_2\circ C_{\rm NOT},\qquad
 U_3=C_{\rm NOT}\circ Q_3\circ C_{\rm NOT}.
 \label{eq:walsh-circuit-app}
\end{align}

These definitions make the construction self-contained.

This all-interval family supplies the cue-stable counterpart to the
schedule-selective audible construction in Section~\ref{sec:quantum-music}.
Here $W$, without a subscript, is the normalized Walsh transform and satisfies
$W^{\mathsf T}=W=W^{-1}$.
At each boundary the final CNOT
of one step is adjacent to the initial CNOT of the next and the pair squares
to the identity.  The gates are retained because the step boundaries and
every contiguous subinterval are part of the claim.

Figure~\ref{fig:walsh-circuit-app} expands the three steps at gate
level and compares the endpoint-only and stepwise-checking protocols.
\begin{figure}[H]
\centering
\begin{minipage}{0.96\linewidth}
\centering
\footnotesize
\textbf{(a) Compiled three-step passage}\par\vspace{0.2em}
$U_1=C_{\rm NOT}\circ Q_1$\par\vspace{-0.4em}
\begin{quantikz}[row sep=0.25cm,column sep=0.30cm]
 \lstick{$q_0$} & \gate{H} & \gate[2]{\Phi(1/2,0,1/2)} & \gate{H}
  & \gate[2]{\Phi(\xi,0,\xi)} & \gate{H} & \ctrl{1} & \qw \\
 \lstick{$q_1$} & \gate{H} & {} & \gate{H}
  & {} & \gate{H} & \targ{} & \qw
\end{quantikz}

\vspace{0.5em}
$U_2=C_{\rm NOT}\circ Q_2\circ C_{\rm NOT}$\par\vspace{-0.4em}
\begin{quantikz}[row sep=0.25cm,column sep=0.22cm]
 \lstick{$q_0$} & \ctrl{1} & \gate{H}
  & \gate[2]{\Phi(1/2-\xi,1/2,-\xi)} & \gate{H}
  & \gate[2]{\Phi(-1/2,1/2,-1)} & \gate{H} & \ctrl{1} & \qw \\
 \lstick{$q_1$} & \targ{} & \gate{H} & {} & \gate{H}
  & {} & \gate{H} & \targ{} & \qw
\end{quantikz}

\vspace{0.5em}
$U_3=C_{\rm NOT}\circ Q_3\circ C_{\rm NOT}$\par\vspace{-0.4em}
\begin{quantikz}[row sep=0.25cm,column sep=0.28cm]
 \lstick{$q_0$} & \ctrl{1} & \gate{H} & \gate[2]{\Phi(\eta,\zeta,\zeta-\eta)} & \gate{H}
  & \gate[2]{\Phi(1/2,0,1/2)} & \gate{H} & \ctrl{1} & \qw \\
 \lstick{$q_1$} & \targ{} & \gate{H} & {} & \gate{H}
  & {} & \gate{H} & \targ{} & \qw
\end{quantikz}

\vspace{0.65em}
\textbf{(b) Endpoint-only and stepwise checking protocols}\par\vspace{0.25em}
\resizebox{0.86\linewidth}{!}{%
\begin{minipage}{\linewidth}
\centering
\textit{Endpoint-only check:}\par\vspace{-0.25em}
\begin{quantikz}[row sep=0.27cm,column sep=0.34cm]
 \lstick{$q_0$} & \gate[2]{U_1} & \gate[2]{U_2} & \gate[2]{U_3}
  & \meter{} & \qw \\
 \lstick{$q_1$} & {} & {} & {} & \meter{} & \qw
\end{quantikz}

\vspace{0.4em}
\textit{Stepwise check:}\par\vspace{-0.25em}
\begin{quantikz}[row sep=0.27cm,column sep=0.28cm]
 \lstick{$q_0$} & \gate[2]{U_1} & \gate[2]{\Mc}
  & \gate[2]{U_2} & \gate[2]{\Mc}
  & \gate[2]{U_3} & \meter{} & \qw \\
 \lstick{$q_1$} & {} & {}
  & {} & {}
  & {} & \meter{} & \qw
\end{quantikz}
\end{minipage}%
}
\end{minipage}
\caption{\captionlead{Two-qubit circuits for the analytic Walsh passage.}
Panel (a) compiles the three steps in Eq.~\eqref{eq:walsh-circuit-app}; circuit
time runs from left to right.  Panel (b) compares the same passage under an
endpoint-only readout and unread context checks after the first two steps.
The next figure gives the record-register realization of each $\Mc$ block.
Theorem~\ref{thm:walsh-composition-app} proves equality of the two endpoint
laws on every contiguous interval; the dense specialization retains nonzero
pairwise path interference.}
\label{fig:walsh-circuit-app}
\end{figure}
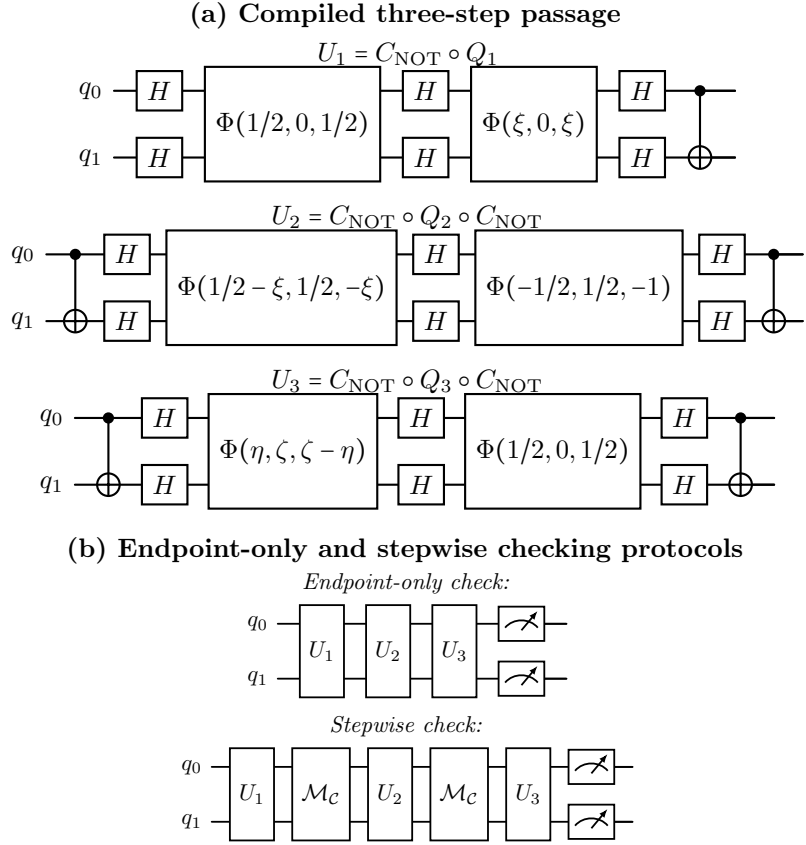

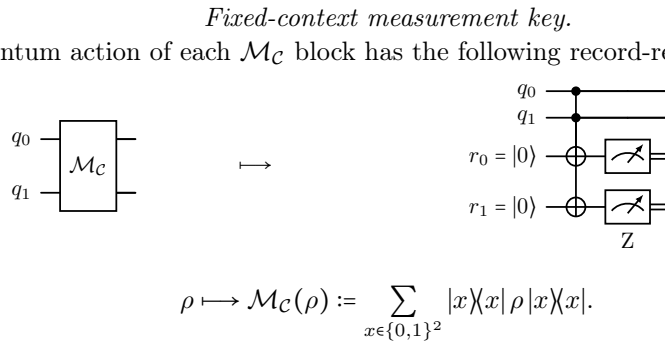
\begin{figure}[H]
\centering
\footnotesize
\textit{Fixed-context measurement key.}\quad

The quantum action of each $\Mc$ block has the following
record-register realization.\par
\vspace{0.35em}
\resizebox{0.84\linewidth}{!}{%
\begin{minipage}{\linewidth}
\centering
\begin{minipage}[c]{0.24\linewidth}
\centering
\begin{quantikz}[row sep=0.27cm,column sep=0.30cm]
 \lstick{$q_0$} & \gate[2]{\Mc} & \qw \\
 \lstick{$q_1$} & {} & \qw
\end{quantikz}
\end{minipage}%
\begin{minipage}[c]{0.12\linewidth}
\centering
\(\longmapsto\)
\end{minipage}%
\begin{minipage}[c]{0.50\linewidth}
\centering
\begin{quantikz}[row sep=0.27cm,column sep=0.30cm,classical gap=0.04cm]
 \lstick{$q_0$} & \ctrl{2} & \qw & \qw \\
 \lstick{$q_1$} & \ctrl{2} & \qw & \qw \\
 \lstick{$r_0=\ket0$} & \targ{} & \meter{} & \setwiretype{c} \\
 \lstick{$r_1=\ket0$} & \targ{} & \gate[ps=meter,priority label=below,disable auto height][2em][1.5em]{\mathrm Z} & \setwiretype{c}
\end{quantikz}
\end{minipage}
\end{minipage}%
}
\[
 \rho\longmapsto\Mc(\rho)
 :=\sum_{x\in\{0,1\}^2} |x\rangle\!\langle x|\,\rho\,|x\rangle\!\langle x|.
\]
\caption{\captionlead{Ancilla implementation of the
fixed-context check $\Mc$.}  Here $\C$ is the computational context.
The CNOTs copy its two binary labels into fresh record targets,
which are measured in the $Z$ basis; double wires denote classical records.
Retaining a record makes its outcome available for sonification;
marginalizing it implements the non-selective channel $\Mc$.}
\label{fig:walsh-readout-main}
\end{figure}

\subsection{Walsh-coordinate representation}
\label{app:walsh-shadows}
\noindent

\begin{definition}[Walsh-coordinate representation]
\label{def:walsh-coordinate-app}

For a classical \(4\times4\) matrix \(K\), define
\(\widetilde K:=WKW\).  The
tilde denotes
the same linear map in Walsh coordinates; in particular,
$\widetilde{\Born(U)}=W\Born(U)W$.
We call this coordinate matrix its
\emph{Walsh shadow}; it need not be stochastic or entrywise nonnegative.  The zeroth
Walsh coordinate is the constant mode and the other three coordinates are
the nonconstant binary parity modes.  If $E_{ab}$ is the matrix unit with a
single $1$ in row $a$, column $b$, then
\begin{equation}
 E_{ab}E_{cd}=\delta_{bc}E_{ad}.
 \label{eq:matrix-unit-law-app}
\end{equation}
\end{definition}

\begin{lemma}[Walsh-coordinate kernels]
\label{lem:walsh-shadows-app}
For the three-parameter family in

Eqs.~\eqref{eq:walsh-gates-app}--\eqref{eq:walsh-circuit-app},
\begin{align}
 \widetilde{\Born(U_1)}&=E_{00}+\sin(\pi\xi)E_{31},\nonumber\\
 \widetilde{\Born(U_2)}&=E_{00}+\sin(\pi\xi)E_{31},\nonumber\\
 \widetilde{\Born(U_3)}&=E_{00}+\sin(\pi\eta)E_{13},
 \label{eq:walsh-one-step-app}
\end{align}
and the coherently composed segments satisfy
\begin{align}
 \widetilde{\Born(U_2\circ U_1)}&=E_{00},\nonumber\\
 \widetilde{\Born(U_3\circ U_2)}
 &=E_{00}+\sin(\pi\xi)\sin(\pi\eta)E_{11},\nonumber\\
 \widetilde{\Born(U_3\circ U_2\circ U_1)}&=E_{00}.
 \label{eq:walsh-segments-app}
\end{align}
\end{lemma}

\Needspace{0.32\textheight}
\begin{proof}
Put $x=e^{\ii\pi\xi}$, $y=e^{\ii\pi\eta}$, and
$z=e^{\ii\pi\zeta}$.
Substitution of the six diagonal phase layers, multiplication of the
$4\times4$ matrices, entrywise squared moduli, and conjugation by $W$ leave
only the entries in Table~\ref{tab:walsh-symbolic-app}.
\begin{table}[H]
\centering
\caption{Nonzero entries of the Walsh shadows.}
\label{tab:walsh-symbolic-app}
\small
\begin{tabular}{@{}lll@{}}
\toprule
Segment & Constant entry & Additional entry\\
\midrule
$U_1$ & $(0,0)=1$ & $(3,1)=\ii(1-x^2)/(2x)=\sin(\pi\xi)$\\
$U_2$ & $(0,0)=1$ & $(3,1)=\ii(1-x^2)/(2x)=\sin(\pi\xi)$\\
$U_3$ & $(0,0)=1$ & $(1,3)=\ii(1-y^2)/(2y)=\sin(\pi\eta)$\\
$U_2\circ U_1$ & $(0,0)=1$ & none\\
$U_3\circ U_2$ & $(0,0)=1$ & $(1,1)=\sin(\pi\xi)\sin(\pi\eta)$\\
$U_3\circ U_2\circ U_1$ & $(0,0)=1$ & none\\
\bottomrule
\end{tabular}
\end{table}
The parameter $z$, hence $\zeta$, cancels from all six Born shadows.  The
calculation is symbolic and holds for arbitrary real $\xi,\eta,\zeta$.
\end{proof}

\begin{theorem}[Analytic three-step composition]
\label{thm:walsh-composition-app}
For the real-parameter family in
Eqs.~\eqref{eq:walsh-gates-app}--\eqref{eq:walsh-circuit-app}, every
nontrivial contiguous interval is operation-supplied compositional:
\begin{align}
 \Born(U_2\circ U_1)&=\Born(U_2)\circ\Born(U_1),\nonumber\\
 \Born(U_3\circ U_2)&=\Born(U_3)\circ\Born(U_2),\nonumber\\
 \Born(U_3\circ U_2\circ U_1)
 &=\Born(U_3)\circ\Born(U_2)\circ\Born(U_1).
 \label{eq:walsh-composition-app}
\end{align}
Hence the three-step passage is boundary-stable for every real
\(\xi,\eta,\zeta\).
\end{theorem}

\begin{proof}[Proof of Theorem~\ref{thm:walsh-composition-app}]
Put $a=\sin(\pi\xi)$ and $b=\sin(\pi\eta)$.  Since $W^2=\openone$, stochastic
kernel composition may be checked in Walsh coordinates.  By
Eq.~\eqref{eq:matrix-unit-law-app},
\begin{align}
 (E_{00}+aE_{31})^2&=E_{00},\nonumber\\
 (E_{00}+bE_{13})(E_{00}+aE_{31})
 &=E_{00}+abE_{11},\nonumber\\
 (E_{00}+abE_{11})(E_{00}+aE_{31})&=E_{00}.
 \label{eq:walsh-unit-products-app}
\end{align}
These are exactly the Walsh-coordinate matrices of the coherent endpoint
Born kernels in Eq.~\eqref{eq:walsh-segments-app}; conjugating back by $W$ proves every
identity in
Eq.~\eqref{eq:walsh-composition-app}.  The following subsections give the
explicit kernels and path-interference calculation at the dense specialization.
\end{proof}

\subsection{Complete dense specialization}
\label{app:walsh-kernels}

Set $\xi=\eta=1/6$ and $\zeta=0$.  Define
\begin{equation}
 K_{31}^{\rm d}:=W\left(E_{00}+\frac12E_{31}\right)W,
 \qquad
 K_{13}^{\rm d}:=W\left(E_{00}+\frac12E_{13}\right)W,
 \label{eq:dense-kernels-app}
\end{equation}
and
\begin{equation}
 K_{\rm unif}:=W E_{00}W
 =|\boldsymbol{+}_4\rangle\langle\boldsymbol{+}_4|=\frac14J_4,
 \qquad
 K_{11}^{\rm d}:=W\left(E_{00}+\frac14E_{11}\right)W.
 \label{eq:walsh-jl-app}
\end{equation}
Here $|\boldsymbol{+}_4\rangle$ is the $d=4$ case of the equal-amplitude ket
defined in Section~\ref{app:bound}.
For direct numerical checking, the two one-step kernels are
\begin{equation}
 K_{31}^{\rm d}=\frac18
 \begin{pmatrix}
  3&1&3&1\\1&3&1&3\\1&3&1&3\\3&1&3&1
 \end{pmatrix},
 \qquad
 K_{13}^{\rm d}=\frac18
 \begin{pmatrix}
  3&1&1&3\\1&3&3&1\\3&1&1&3\\1&3&3&1
 \end{pmatrix}.
 \label{eq:walsh-dense-matrices-app}
\end{equation}
The complete interval calculation is
\begin{equation}
 \Born(U_1)=K_{31}^{\rm d},
 \quad \Born(U_2)=K_{31}^{\rm d},
 \quad \Born(U_3)=K_{13}^{\rm d},
 \label{eq:walsh-one-kernels-app}
\end{equation}
and
\begin{align}
 \Born(U_2\circ U_1)&=K_{31}^{\rm d}\circ K_{31}^{\rm d}=K_{\rm unif},\nonumber\\
 \Born(U_3\circ U_2)&=K_{13}^{\rm d}\circ K_{31}^{\rm d}=K_{11}^{\rm d},\nonumber\\
 \Born(U_3\circ U_2\circ U_1)
 &=K_{13}^{\rm d}\circ K_{31}^{\rm d}\circ K_{31}^{\rm d}=K_{\rm unif}.
 \label{eq:walsh-all-kernels-app}
\end{align}

These equations give the complete interval algebra of the
three-step family.

\subsection{Multipath interference}

Consider the first two steps, input $0$, and output $1$.  The four amplitude
paths through the intermediate labels $m=0,1,2,3$ are
$z_m:=\langle1|U_2|m\rangle\langle m|U_1|0\rangle$, namely
\begin{equation}
 z_0=\frac{\sqrt3}{8},
 \qquad z_1=-\frac{\sqrt3}{8},
 \qquad z_2=\frac{\ii}{8},
 \qquad z_3=\frac{3\ii}{8}.
 \label{eq:walsh-paths-app}
\end{equation}
Every path is nonzero, and
\begin{equation}
 \left|\sum_{m=0}^{3}z_m\right|^2
 =\left|\frac{\ii}{2}\right|^2
 =\frac14
 =\sum_{m=0}^{3}|z_m|^2.
\end{equation}
The relation between the two probabilities is
\begin{equation}
 \left|\sum_mz_m\right|^2
 =\sum_m|z_m|^2
  +2\sum_{0\le r<s\le3}\operatorname{Re}(z_r\overline{z_s}).
\end{equation}
Only two real cross terms are nonzero:
\begin{equation}
 2\operatorname{Re}(z_0\overline{z_1})=-\frac3{32},
 \qquad
 2\operatorname{Re}(z_2\overline{z_3})=+\frac3{32}.
\end{equation}
They cancel.  The equality therefore comes from interference
balance, not from absent or separated amplitude paths.

\section{Musical realizations and prediction complexity}
\label{app:quantum-music-complexity}

We first give the two-pitch realization and the qutrit rhythm selector,
including their sound-output conventions. We then prove the transition-
prediction and sampling results. The prediction problem applies equally
when the outputs label states of a finite state machine.

\subsection{\texorpdfstring{Qubit Ostinato: an audible realization}{Qubit Ostinato: an audible realization}}
\label{app:qubit-ostinato}

\begin{example}[Qubit Ostinato]
\label{ex:qubit-ostinato-app}
The \emph{Qubit Ostinato} is the audible
 $K(1/2)$ member of the exact
constant-kernel family proved in Corollary~\ref{cor:constant-kernel-infinity}.
It maps the computational basis to the adjacent pitch space
\[
 A_{\mathrm{Ostinato}}=\{\mathrm E_2,\mathrm F_2\},
 \qquad \ell_{\rm snd}(|0\rangle)=\mathrm E_2,
 \qquad \ell_{\rm snd}(|1\rangle)=\mathrm F_2.
\]
Consequently,
\begin{equation}
 \begin{aligned}
 \Born(U_k)
 &=\;K(1/2)
 =
 \;\begin{pmatrix}\frac34&\frac14\\[2pt]\frac14&\frac34\end{pmatrix},\\
 \Born(U_k\circ\cdots\circ U_1)
 &=\;K(2^{-k}).
 \end{aligned}
 \label{eq:qubit-ostinato-score}
\end{equation}
Under readout after every operation, the pitch changes with probability
$1/4$ at each step, while every finite prefix has the same endpoint
distribution as the coherent passage. Figure~\ref{fig:qubit-ostinato-demo} makes
the first two updates explicit.
\end{example}

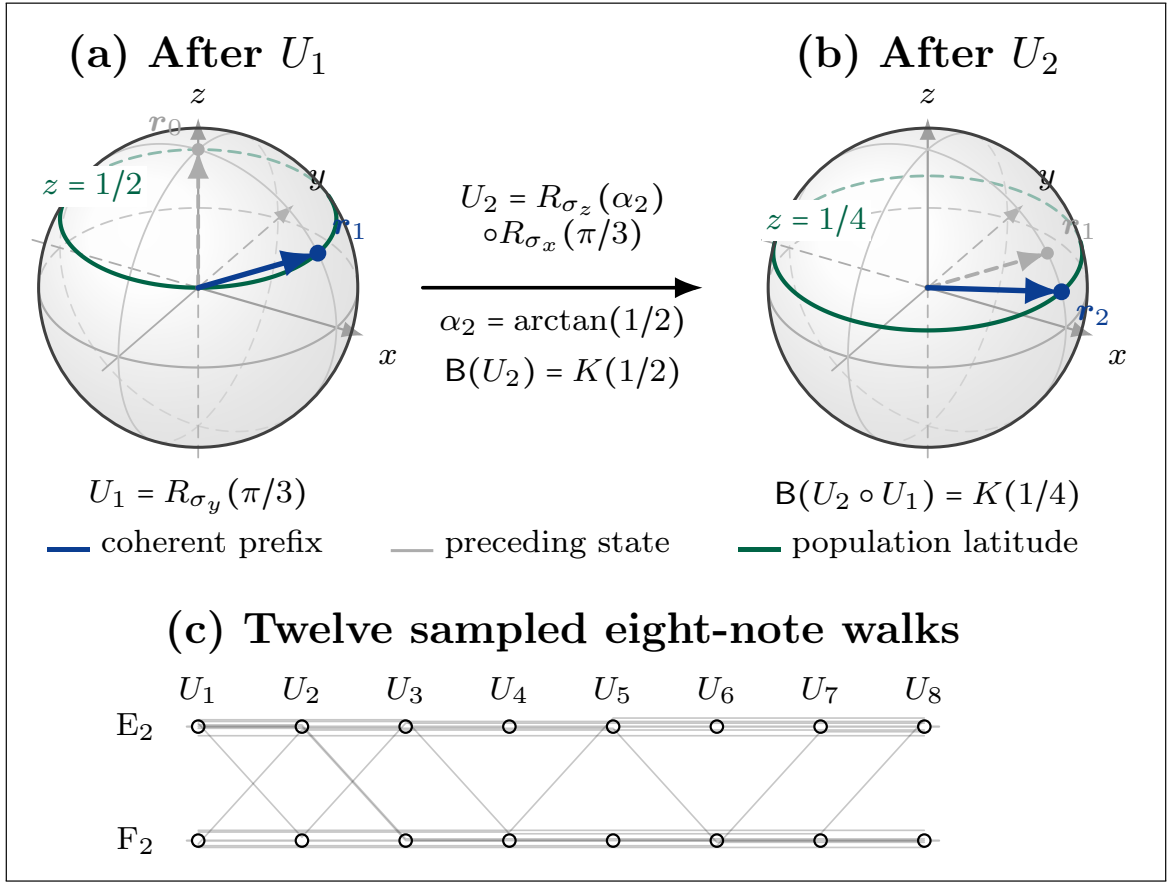
\begin{figure}[H]
\centering
\begingroup
\setlength{\fboxsep}{5pt}
\setlength{\fboxrule}{0.45pt}
\fbox{%
\begin{minipage}{0.94\linewidth}
\centering
\footnotesize
\vspace{0.15em}
\resizebox{0.98\linewidth}{!}{%
\begin{tikzpicture}[>=Latex,line cap=round,line join=round]
  \newcommand{\OstinatoBlochSphere}[1]{%
    \pgfmathsetmacro{\OstLatRadius}{sqrt(1-(#1)^2)}
    \pgfmathsetmacro{\OstLatCentre}{1.4*sqrt(3)*(#1)/2}
    \pgfmathsetmacro{\OstLatCut}{acos(-(#1)/(sqrt(3)*\OstLatRadius))}
    \shade[ball color=gray!12,opacity=0.18] (0,0) circle (1.40);
    \draw[gray!45,densely dashed,line width=0.45pt]
      plot[domain=90:270,samples=75,variable=\t]
      ({1.4*sin(\t)},{-0.7*cos(\t)});
    \draw[gray!35,densely dashed,line width=0.40pt]
      plot[domain=130.893395:310.893395,samples=75,variable=\t]
      ({0.7*sqrt(3)*sin(\t)},
       {1.4*(-sin(\t)/4+sqrt(3)*cos(\t)/2)});
    \draw[gray!35,densely dashed,line width=0.40pt]
      plot[domain=33.690068:213.690068,samples=75,variable=\t]
      ({0.7*sin(\t)},
       {1.4*(sqrt(3)*sin(\t)/4+sqrt(3)*cos(\t)/2)});
    \draw[pathone!45,densely dashed,line width=0.70pt]
      plot[domain=\OstLatCut:360-\OstLatCut,samples=75,variable=\t]
      ({1.4*\OstLatRadius*sin(\t)},
       {\OstLatCentre-0.7*\OstLatRadius*cos(\t)});
    \draw[gray!65,densely dashed,line width=0.45pt]
      ({-0.84*sqrt(3)},0.42) -- (0,0);
    \draw[->,gray!75,line width=0.55pt]
      (0,0) -- ({0.84*sqrt(3)},-0.42)
      node[below right,font=\scriptsize,text=black] {$x$};
    \draw[gray!65,line width=0.45pt]
      (-0.84,{-0.42*sqrt(3)}) -- (0,0);
    \draw[->,gray!60,densely dashed,line width=0.45pt]
      (0,0) -- (0.84,{0.42*sqrt(3)})
      node[above right,font=\scriptsize,text=black] {$y$};
    \draw[gray!60,densely dashed,line width=0.45pt]
      (0,-1.49) -- (0,0);
    \draw[->,gray!75,line width=0.55pt]
      (0,0) -- (0,1.49)
      node[above,font=\scriptsize,text=black] {$z$};
    \draw[gray!65,line width=0.55pt]
      plot[domain=-90:90,samples=75,variable=\t]
      ({1.4*sin(\t)},{-0.7*cos(\t)});
    \draw[gray!45,line width=0.45pt]
      plot[domain=-49.106605:130.893395,samples=75,variable=\t]
      ({0.7*sqrt(3)*sin(\t)},
       {1.4*(-sin(\t)/4+sqrt(3)*cos(\t)/2)});
    \draw[gray!45,line width=0.45pt]
      plot[domain=-146.309932:33.690068,samples=75,variable=\t]
      ({0.7*sin(\t)},
       {1.4*(sqrt(3)*sin(\t)/4+sqrt(3)*cos(\t)/2)});
    \draw[pathone,line width=1.05pt]
      plot[domain=-\OstLatCut:\OstLatCut,samples=85,variable=\t]
      ({1.4*\OstLatRadius*sin(\t)},
       {\OstLatCentre-0.7*\OstLatRadius*cos(\t)});
    \draw[black!75,line width=0.75pt] (0,0) circle (1.40);
    \fill[black!60] (0,0) circle (0.8pt);
  }
  \begin{scope}[shift={(0,0)}]
    \node[font=\small\bfseries] at (0,2.03) {(a) After $U_1$};
    \OstinatoBlochSphere{0.5}
    \node[font=\scriptsize,text=pathone,fill=white,inner sep=1pt]
      at (-0.93,0.89) {$z=1/2$};
    \draw[->,gray!65,densely dashed,line width=0.95pt]
      (0,0) -- (0,{0.7*sqrt(3)});
    \fill[gray!65] (0,{0.7*sqrt(3)}) circle (1.6pt);
    \node[font=\scriptsize,text=gray!75,above left]
      at (0,{0.7*sqrt(3)}) {$\boldsymbol r_0$};
    \draw[->,line width=1.35pt,deepnavy]
      (0,0) -- (1.05,{0.175*sqrt(3)});
    \fill[deepnavy] (1.05,{0.175*sqrt(3)}) circle (2.1pt);
    \node[font=\scriptsize,text=deepnavy,above right]
      at (1.05,{0.175*sqrt(3)}) {$\boldsymbol r_1$};
    \node[font=\scriptsize] at (0,-1.83)
      {$U_1=R_{\sigma_y}(\pi/3)$};
  \end{scope}

  \draw[->,line width=0.9pt] (1.97,0) -- (4.43,0);
  \node[font=\scriptsize,align=center] at (3.20,0.62)
    {$U_2=R_{\sigma_z}(\alpha_2)$\\[-1pt]
     $\circ R_{\sigma_x}(\pi/3)$};
  \node[font=\scriptsize] at (3.20,-0.30)
    {$\alpha_2=\arctan(1/2)$};
  \node[font=\scriptsize] at (3.20,-0.75)
    {$\Born(U_2)=K(1/2)$};

  \begin{scope}[shift={(6.40,0)}]
    \node[font=\small\bfseries] at (0,2.03) {(b) After $U_2$};
    \OstinatoBlochSphere{0.25}
    \node[font=\scriptsize,text=pathone,fill=white,inner sep=1pt]
      at (-0.98,0.59) {$z=1/4$};
    \draw[->,gray!65,densely dashed,line width=0.95pt]
      (0,0) -- (1.05,{0.175*sqrt(3)});
    \fill[gray!65] (1.05,{0.175*sqrt(3)}) circle (1.6pt);
    \node[font=\scriptsize,text=gray!75,above right]
      at (1.05,{0.175*sqrt(3)}) {$\boldsymbol r_1$};
    \draw[->,line width=1.35pt,deepnavy]
      (0,0) -- ({0.175*sqrt(45)},{0.175*sqrt(3)-0.0875*sqrt(15)});
    \fill[deepnavy]
      ({0.175*sqrt(45)},{0.175*sqrt(3)-0.0875*sqrt(15)}) circle (2.1pt);
    \node[font=\scriptsize,text=deepnavy,below right]
      at ({0.175*sqrt(45)},{0.175*sqrt(3)-0.0875*sqrt(15)})
      {$\boldsymbol r_2$};
    \node[font=\scriptsize] at (0,-1.83)
      {$\Born(U_2\circ U_1)=K(1/4)$};
  \end{scope}

  \node[font=\scriptsize] at (3.20,-2.25)
    {\textcolor{deepnavy}{\rule{1.25em}{1.15pt}} coherent prefix\qquad
     \textcolor{gray!65}{\rule{1.25em}{0.75pt}} preceding state\qquad
     \textcolor{pathone}{\rule{1.25em}{1.15pt}} population latitude};

  \begin{scope}[yshift=-3mm]
  \node[font=\small\bfseries] at (3.20,-2.72)
    {(c) Twelve sampled eight-note walks};
  \begin{scope}[yshift=-2mm]
  \draw[gray!48,line width=0.55pt] (-0.10,-3.35) -- (6.50,-3.35);
  \draw[gray!48,line width=0.55pt] (-0.10,-4.35) -- (6.50,-4.35);
  \node[font=\scriptsize\bfseries,anchor=east] at (-0.25,-3.35) {$\mathrm E_2$};
  \node[font=\scriptsize\bfseries,anchor=east] at (-0.25,-4.35) {$\mathrm F_2$};
  \foreach \x/\u in {0/U_1,0.91/U_2,1.82/U_3,2.73/U_4,
                       3.64/U_5,4.55/U_6,5.46/U_7,6.37/U_8} {
    \node[font=\scriptsize] at (\x,-3.03) {$\u$};
    \filldraw[fill=white,line width=0.55pt] (\x,-3.35) circle (1.55pt);
    \filldraw[fill=white,line width=0.55pt] (\x,-4.35) circle (1.55pt);
  }
  \newcommand{\CompactBornPath}[9]{%
    \draw[black,opacity=0.22,line width=0.42pt]
      (0,{-3.85+0.50*(#2)+#1})--
      (0.91,{-3.85+0.50*(#3)+#1})--
      (1.82,{-3.85+0.50*(#4)+#1})--
      (2.73,{-3.85+0.50*(#5)+#1})--
      (3.64,{-3.85+0.50*(#6)+#1})--
      (4.55,{-3.85+0.50*(#7)+#1})--
      (5.46,{-3.85+0.50*(#8)+#1})--
      (6.37,{-3.85+0.50*(#9)+#1});%
  }
  \CompactBornPath{-0.08}{1}{1}{1}{1}{1}{1}{1}{1}
  \CompactBornPath{-0.06}{-1}{-1}{-1}{-1}{-1}{-1}{-1}{-1}
  \CompactBornPath{-0.045}{-1}{-1}{-1}{-1}{-1}{-1}{1}{1}
  \CompactBornPath{-0.03}{-1}{1}{1}{1}{1}{1}{1}{1}
  \CompactBornPath{-0.015}{1}{1}{1}{1}{1}{-1}{-1}{-1}
  \CompactBornPath{0}{1}{1}{-1}{-1}{-1}{-1}{-1}{-1}
  \CompactBornPath{0.015}{1}{1}{-1}{-1}{-1}{-1}{-1}{-1}
  \CompactBornPath{0.03}{1}{-1}{1}{1}{1}{1}{1}{1}
  \CompactBornPath{0.045}{1}{1}{1}{1}{1}{1}{1}{1}
  \CompactBornPath{0.06}{1}{1}{1}{-1}{-1}{-1}{-1}{1}
  \CompactBornPath{0.075}{-1}{-1}{-1}{-1}{1}{1}{1}{1}
  \CompactBornPath{0.09}{-1}{-1}{-1}{-1}{-1}{-1}{-1}{-1}
  \end{scope}
  \end{scope}
\end{tikzpicture}%
}
\end{minipage}%
}
\endgroup
\caption{\captionlead{Qubit Ostinato in the two-note space $\{\mathrm E_2,\mathrm F_2\}$.}
 Here $\Pr(\mathrm{change})=1/4$, with
$|0\rangle\mapsto\mathrm E_2$ and $|1\rangle\mapsto\mathrm F_2$.
Panels (a) and (b) are orthographic views of the Bloch sphere after the
first two steps of the constant-kernel sequence in
Eqs.~\eqref{eq:signed-infinity-notes}
and~\eqref{eq:signed-infinity-phase}.  The coherent Bloch vectors are
$\boldsymbol r_1=(\sqrt3/2,0,1/2)$ and
$\boldsymbol r_2=(\sqrt{15}/4,0,1/4)$.
Gray arrows show the preceding states, beginning at
$\boldsymbol r_0=(0,0,1)$.  Green latitude circles mark the population
polarizations $z=1/2$ and $z=1/4$; rear arcs are lighter and dashed.
Panel (c) shows twelve reproducible eight-note samples of the
associated two-state stochastic walk, beginning with the first audible
outcome after $U_1$.}
\label{fig:qubit-ostinato-demo}
\end{figure}
The Qubit Ostinato is an audible realization of the separation
between initial- and internal-cue tests: composition from the initial cue does
not by itself make the score cue-stable.

\subsection{Qutrit rhythm selector}
\label{app:qutrit-rhythm-selector}

Let $J_3$ denote the $3\times3$ all-ones matrix. The qutrit selector used in
Figure~\ref{fig:emulated-quantum-score} is
\[
 \begin{gathered}
 P_3|s\rangle=|(s+1)\bmod3\rangle,\qquad
 a=\frac{-1+\ii\sqrt{55}}8,\\
 V^{\rm sel}=aP_3+\frac14(J_3-P_3).
 \end{gathered}
\]
Its Born kernel is given in Eq.~\eqref{eq:score-selector-law-main}.

\paragraph{Score and sound output.}
\label{app:score-sound-output}
A score prescribes an ordered sequence of unitary operations in a fixed
context, together with a classical sound assignment; its readout schedule is
specified separately. In the musical dictionary, each operation is a note
and each temporal boundary is a cue. At an enabled boundary $b$, outcome $x$
prepares the post-measurement state $|x\rangle$ and the classical map
$\ell_{{\rm snd},b}:X\to\mathsf A_{\rm snd}$ assigns a sound. The boundary
dependence of this map changes the voicing, while $x$ retains its role in the
subsequent quantum evolution.

The entrance sound is specified separately from the entrance quantum state
and is held until the first scheduled readout. Every enabled readout starts
its assigned sound anew, including when the pitch repeats; between readouts
the sound is held. Thus the spacings between successive readouts determine
the rhythm, with an externally chosen beat duration $\Delta$.

\paragraph{Displayed performance.}
In the two-qubit realization, each data pass starts in $|00\rangle$ and the
entrance sound carries over from the preceding endpoint. The selector starts
in state $S=0$ and thereafter is prepared in its preceding observed state
before $V^{\rm sel}$ and readout. Its outcomes $0,1,2$ choose $(2,1)$, $(3)$,
and $(1,1,1)$, respectively. The sound map is the boundary-dependent pitch
table in Section~\ref{sec:quantum-music}.

The displayed seed-1766 window begins at an observed endpoint and contains
five complete three-beat data passes, with selector outcomes $1,2,0,1,2$.
The mandatory final readout is at $t=15\Delta$. Its assigned pitch
$\mathrm C_4$ is sounded until $t=16\Delta$, followed by one silent beat
ending at $t=17\Delta$. This terminal one-beat hold is a playback convention
after the measured passage and adds no quantum operation or readout.
The main-text circuit in Eq.~\eqref{eq:qutrit-schedule-control-main} and
Figure~\ref{fig:emulated-quantum-score} give the complete realization and the
256 eight-readout comparison excerpts. Their sampled outcomes and durations
are independent of this displayed terminal hold.

\subsection{Quantum musical-transition prediction}
\label{app:quantum-musical-prediction}

\begin{definition}[Complexity conventions]
\label{def:complexity-conventions-app}
Fix the standard finite universal gate set
$\mathcal G=\{H,T,C_{\rm NOT}\}$, where
\[
 H=\frac1{\sqrt2}\begin{pmatrix}1&1\\1&-1\end{pmatrix},
 \qquad
 T=\operatorname{diag}(1,e^{\ii\pi/4}),
 \qquad
 C_{\rm NOT}|c,r\rangle=|c,r\mathbin\oplus c\rangle.
\]
A promise problem is a pair
$\mathcal L=(\mathcal L_{\rm yes},
\mathcal L_{\rm no})$ of disjoint languages.  We write
\PromiseBQP{} for the class of promise problems decided by a polynomial-time
uniform family of polynomial-size quantum circuits over $\mathcal G$.
There are polynomial-time computable rational thresholds
$0\leq\alpha(\ell)<\beta(\ell)\leq1$, separated by an inverse polynomial in the
classical input length $\ell$, such that acceptance probability is at least
$\beta(\ell)$ on $\mathcal L_{\rm yes}$ and at most $\alpha(\ell)$ on
$\mathcal L_{\rm no}$. Repetition on fresh registers and reversible
comparison of the observed frequency with the thresholds' midpoint give the usual
bounded-error definition with polynomial overhead
~\cite[Sec.~IV.2]{app:WatrousComplexity}.
\end{definition}

\begin{remark}[History presentation]
\label{rem:qmp-history-presentation-app}
For the feedback architecture in the main text, fix a finite audible alphabet
$\mathsf A_{\rm snd}$ and a literal polynomial-time injective encoding
\(\operatorname{enc}\) of circuit descriptions into finite symbol strings.
Choose this encoding so that the symbol string is at least as long as the
circuit's bit description.
A history $m=(a_1,\ldots,a_t)\in\mathsf A_{\rm snd}^*$ consists only of
symbols that have already been played.  The special prelude
\begin{equation}
 m_Q:=\operatorname{enc}(\langle Q\rangle)
 \label{eq:qmp-history-encoding-app}
\end{equation}
is decoded by the score compiler before the next passage, giving
$U(m_Q)=Q$.  Conversely, the decoded circuit supplies the next transition
prediction.  Encoding and decoding are classical polynomial-time operations,
so this history presentation and the explicit circuit presentation below
define polynomial-time equivalent presentations when the thresholds $a,b$
are passed as additional input. Exact identity gate pairs can pad the
decoded circuit, or its encoded history, to ensure that the target input
length is at least the source length, preserving the inverse-polynomial gap
condition. Symbols produced by a
passage are appended only for a later compiler invocation; no same-run
quantum feedback is assumed.
\end{remark}

\begingroup
\setlength{\emergencystretch}{2em}
We
now formulate the
prediction problem and prove its completeness locally.  Fix two distinct terminal transition labels
$\mathsf t_0,\mathsf t_1$; these are classical score outputs, not circuit gates
or time indices. An instance supplies a classical description $\langle Q\rangle$
of an $n$-qubit circuit over $\mathcal G$, initialized in $|0^n\rangle$,
and rational thresholds $a,b$.
We use an explicit circuit encoding whose length polynomially bounds
both the number of qubits and the number of gates.
Measuring its distinguished musical output qubit $M$ gives
$z\in\{0,1\}$ and selects $\mathsf t_z$; the remaining qubits form an
unobserved register $\mathsf J$, with outcome $u\in\{0,1\}^{n-1}$.  In the
computational atomic context, set
\par\endgroup
\begin{equation}
 \ell_{\rm tr}(z,u):=\mathsf t_z.
 \label{eq:musical-transition-label-app}
\end{equation}
The probability of the designated musical transition $\mathsf t_1$ is
therefore
\begin{equation}
p_Q:=\Pr_Q(M=1)
=\sum_{u\in\{0,1\}^{n-1}}
  [\Born_{\C}(Q)]_{(1,u),0^n}.
 \label{eq:qmp-probability-app}
\end{equation}
The explicit context subscript records that this is the canonical
computational-context Born kernel.

\begin{definition}[Quantum musical-transition prediction
  (\textup{\textsc{QMTP}})]
\label{def:qmp-app}
Fix an integer $c\geq1$. An instance is $\langle Q,a,b\rangle$, with
rational thresholds $0\leq a<b\leq1$ satisfying
$b-a\geq\ell^{-c}$, where $\ell:=|\langle Q,a,b\rangle|$ is the total
input bit length. The thresholds are part of the input; $c$ is fixed for
the problem. Among these instances, the yes- and no-instances are
\begin{equation}
 \begin{aligned}
 \textsc{QMTP}_{\rm yes}
   &:=\{\langle Q,a,b\rangle:p_Q\geq b\},\\
 \textsc{QMTP}_{\rm no}
   &:=\{\langle Q,a,b\rangle:p_Q\leq a\}.
 \end{aligned}
 \label{eq:qmp-promise-app}
\end{equation}
Inputs with $a<p_Q<b$ lie outside the promise. The promise gap is $b-a$.
\end{definition}

\begin{theorem}[\PromiseBQP-completeness]
\label{thm:qmp-promisebqp-app}
\mbox{}\par\noindent
\begin{sloppypar}
\noindent

For every fixed $c\geq1$ in
Definition~\ref{def:qmp-app}, quantum musical-transition prediction
(\textup{\textsc{QMTP}}) is \PromiseBQP-complete under deterministic
classical polynomial-time many-one reductions.
\end{sloppypar}
\end{theorem}

\begin{proof}

For membership, run $Q$ independently $N$ times,
measure $M$ in each run, and compare the observed frequency with
$(a+b)/2$. Hoeffding's inequality
bounds the error by $\exp[-N(b-a)^2/2]$, so
$N=O(\ell^{2c})$ repetitions
give bounded error.
Reversible counting and rational-threshold comparison implement this procedure as a uniform
polynomial-size quantum circuit.

\begin{sloppypar}

For hardness, let
$\mathcal L=(\mathcal L_{\rm yes},
\mathcal L_{\rm no})\in\PromiseBQP$. Its uniform circuit family produces in
classical polynomial time an
$n=n(|w|)$-qubit verifier $U(w)$, with both qubit and gate counts
polynomially bounded in $|w|$, over $\mathcal G$, including the
computational-basis preparation of $w$, with
acceptance qubit $s$,
recorded outcome $S\in\{0,1\}$, and
remaining work register $\mathsf W$. Let $\alpha_w<\beta_w$ be its
polynomial-time computable rational thresholds, with
$\Delta_w:=\beta_w-\alpha_w\geq1/r(|w|)$ for some fixed polynomial $r$.
Then
\end{sloppypar}
\begin{equation}
 w\in \mathcal L_{\rm yes}
 \Longrightarrow\Pr(S=1)\geq\beta_w,
 \qquad
 w\in \mathcal L_{\rm no}
 \Longrightarrow\Pr(S=1)\leq\alpha_w.
 \label{eq:qmp-verifier-promise-app}
\end{equation}
Add a fresh qubit $M$ in $|0\rangle$ and define
\begin{equation}
 Q(w)
 :=C_{\rm NOT}^{s\to M}
 \circ(U(w)\otimes\openone_M).
 \label{eq:qmp-reduction-circuit-app}
\end{equation}
Here $Q(w)$ has
$n+1$ qubits, and the unobserved register above is
$\mathsf J=(\mathsf W,s)$.
Writing the verifier state before the final CNOT
as
\begin{equation}
 U(w)|0^n\rangle
 =|\phi_{w,0}\rangle_{\mathsf W}
   |0\rangle_{s}
  +|\phi_{w,1}\rangle_{\mathsf W}
   |1\rangle_{s},
 \label{eq:qmp-verifier-branches-app}
\end{equation}
with unnormalized branch vectors, gives
$p_{Q(w)}=\|\phi_{w,1}\|^2=\Pr_{U(w)}(S=1)$.
Set $a(w)=\alpha_w$ and $b(w)=\beta_w$. Pad the explicit circuit
description with $H\circ H=\openone$ gate pairs if needed so that
$\ell(w):=|\langle Q(w),a(w),b(w)\rangle|\geq\lceil\Delta_w^{-1}\rceil$.
This preserves the unitary and its acceptance probability exactly and uses
only polynomially many gates, since $\Delta_w^{-1}\leq r(|w|)$.
For $c\geq1$, it gives
$\ell(w)^{-c}\leq\ell(w)^{-1}\leq\Delta_w$.
Thus $w\mapsto\langle Q(w),a(w),b(w)\rangle$ is a classical
polynomial-time reduction satisfying the \textup{\textsc{QMTP}} promise.
Composing it with
$Q(w)\mapsto m_{Q(w)}=\operatorname{enc}(\langle Q(w)\rangle)$, retaining
$a(w),b(w)$ and padding if needed, proves the
same hardness statement when
\textup{\textsc{QMTP}} is presented as an already played musical
history.
\end{proof}

Thus, in the standard promise-class sense, deciding a designated transition in
this circuit-universal quantum score model is \PromiseBQP-complete
\cite{app:JanzingWocjan2007}.  This neither makes every fixed transition hard,
nor proves that the restricted Born-kernel scores here are universal or that
deciding cue stability is \PromiseBQP-hard.

\subsection[A verifier-controlled stepwise check with an unconditional endpoint]{A verifier-controlled stepwise check with an unconditional endpoint}
\label{app:verifier-controlled-check}

In Figure~\ref{fig:bqp-controlled-check-app}, $\mathsf W$ is the
verifier work register, $D$ is the data register, and $M_{\rm step}$ and
$M_{\rm end}$ retain the intermediate and endpoint outcomes
$Y_{\rm step}$ and $Y_{\rm end}:=Y$.  The fixed history argument $m$ is
suppressed in the gate labels.

The verifier output can also choose a measurement schedule.  Let $m$ be a
fixed already played history, and let its score compiler produce a selector
circuit $U(m)$ with acceptance qubit $s$ and
$p(m):=\Pr_{U(m)}(S=1)$.  The same compiler selects a pair
$U_1(m),U_2(m)$ of $k$-qubit data notes.  Let
$D=(q_0,\ldots,q_{k-1})$ be a $k$-qubit data register initialized in
$|x\rangle$, with
$x\in X=\{0,\ldots,2^k-1\}$ labeling a computational-context outcome,
and let the two data notes act on $D$ in that order.  Initialize
the meter registers $M_{\rm step}$ and $M_{\rm end}$
in $|0^k\rangle$.

\begingroup
\color{black}
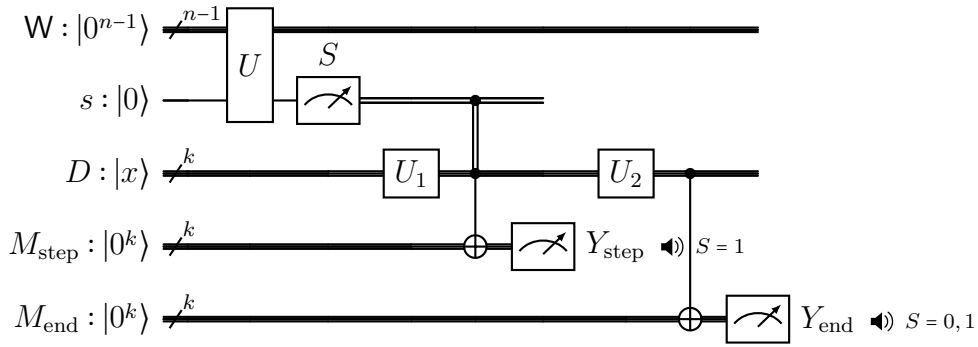
\begin{figure}[H]
\centering
\begingroup
\newsavebox{\scoreSpeakerBox}
\sbox{\scoreSpeakerBox}{%
 \begin{tikzpicture}[x=1.2ex,y=1.2ex,baseline=-0.45ex,
   line width=0.5pt,line cap=round,line join=round]
  \fill (0,0.22)--(0.28,0.22)--(0.68,0.55)--(0.68,-0.55)
    --(0.28,-0.22)--(0,-0.22)--cycle;
  \draw (0.86,0.34) to[out=-35,in=35] (0.86,-0.34);
  \draw (1.05,0.56) to[out=-35,in=35] (1.05,-0.56);
 \end{tikzpicture}%
}
\begin{quantikz}[
 row sep=0.34cm,
 column sep=0.32cm,
 ampersand replacement=\&,
 wire types={b,q,b,b,b}
]
 \lstick{$\mathsf W:|0^{n-1}\rangle$}
  \& \qwbundle{n-1}\qw \&[0.20cm] \gate[2]{U}
  \& \qw \& \qw \& \qw \& \qw \& \qw \& \qw \& \qw \\
 \lstick{$s:|0\rangle$}
  \& \qw \& {} \& \meter{S} \& \setwiretype{c}\cw
  \& \ctrl[vertical wire=c]{1}\cw \& \cw \\
 \lstick{$D:|x\rangle$}
  \& \qwbundle{k}\qw \& \qw \& \qw \& \gate{U_1} \& \ctrl{1} \& \qw
  \& \gate{U_2} \& \ctrl{2} \& \qw \\
 \lstick{$M_{\rm step}:|0^k\rangle$}
  \& \qwbundle{k}\qw \& \qw \& \qw \& \qw \& \targ{}
  \& \meter{}\rstick{$Y_{\rm step}$\enspace\usebox{\scoreSpeakerBox}\ {\scriptsize$S=1$}} \\
 \lstick{$M_{\rm end}:|0^k\rangle$}
  \& \qwbundle{k}\qw \& \qw \& \qw \& \qw \& \qw \& \qw \& \qw
  \& \targ{} \& \meter{}\rstick{$Y_{\rm end}$\enspace\usebox{\scoreSpeakerBox}\ {\scriptsize$S=0,1$}}
\end{quantikz}
\endgroup
\caption{\captionlead{Verifier-controlled checking schedule.}  Reading $S$ before the data
passage classically enables the step record bank only for $S=1$; its retained
outcome $Y_{\rm step}$ supplies the optional intermediate sound.  The endpoint
bank records and sounds $Y_{\rm end}$ on both branches.  The intermediate
outcome is not used as feed-forward within the passage, although its
measurement backaction on $D$ is retained.
Bundled wires denote multiqubit registers.}
\label{fig:bqp-controlled-check-app}
\end{figure}
\endgroup

Write $\operatorname{CCX}_{c_1,c_2\to r}$ for the Toffoli gate
\[
 |c_1,c_2,r\rangle\longmapsto
 |c_1,c_2,r\mathbin\oplus(c_1c_2)\rangle.
\]
It has an exact constant-size decomposition over $\mathcal G$
($T^\dagger=T^7$).  Apply
\begin{equation}
 \begin{aligned}
 C_{\rm step}^{(s)}
 &:=\prod_{j=0}^{k-1}\operatorname{CCX}_{s,q_j\to
   (M_{\rm step})_j},\\
 C_{\rm end}
 &:=\prod_{j=0}^{k-1}C_{\rm NOT}^{q_j\to
   (M_{\rm end})_j}.
 \end{aligned}
 \label{eq:bqp-controlled-record-banks-app}
\end{equation}
Place $C_{\rm step}^{(s)}$
after $U_1(m)$ and the unconditional $C_{\rm end}$
after $U_2(m)$.
If $S$ is the recorded value of $s$,
$X=x$ is the data entrance, and $Y=y$
is the endpoint record, then summing over the
step record gives
\begin{equation}
 \begin{aligned}
 K_0(m)&=\Born_{\C}(U_2(m)\circ U_1(m)),\\
 K_1(m)&=\Born_{\C}(U_2(m))\circ\Born_{\C}(U_1(m)),
 \end{aligned}
 \label{eq:bqp-schedule-kernels-app}
\end{equation}
\begin{equation}
\begin{aligned}
 \Pr(S=0,Y=y\mid X=x,m)
 &=(1-p(m))(K_0(m))_{yx},\\
 \Pr(S=1,Y=y\mid X=x,m)
 &=p(m)(K_1(m))_{yx}.
 \end{aligned}
 \label{eq:bqp-schedule-joint-law-app}
\end{equation}
\Needspace{3\baselineskip}
Thus the endpoint check is present on both branches, whereas the
stepwise check
is present only on the accepting branch.  The
step record may be read and
forgotten at the boundary, or its measurement may be deferred to the end:
because no later gate acts on that record, both implementations give the same
endpoint statistics.  \ Figure~\ref{fig:bqp-controlled-check-app} shows the
equivalent early-readout implementation for live playback.

For live playback, measure
$s$ first and
classically enable the step CNOT bank; this implements the same conditional
endpoint law.
Enable the step speaker only for $S=1$ and the endpoint speaker
for both values.  The emitted symbols are appended to $m$ only after this
passage; marginalizing the step record therefore leaves
Eq.~\eqref{eq:bqp-schedule-joint-law-app} unchanged within the current run.

\begin{corollary}[Endpoint-only conditional-check prediction]
\label{cor:qmp-endpoint-check-app}
Specialize Fig.~\ref{fig:bqp-controlled-check-app} to $k=1$, initialize
$D$ in $|0\rangle$, and take $U_1(m)=U_2(m)=H$. Fix an integer $c\geq1$.
An instance is $\langle m,a,b\rangle$, where $m$ uses the circuit-universal
history presentation of Remark~\ref{rem:qmp-history-presentation-app}
and its compiler produces a verifier circuit $U(m)$ over $\mathcal G$
with designated output qubit $s$. The rational input thresholds satisfy
$0\leq a<b\leq1$ and $b-a\geq\ell^{-c}$, where
$\ell:=|\langle m,a,b\rangle|$ is the total input length.
The task is to decide whether
\begin{equation}
 \Pr(Y=1)\geq
 \frac b2
 \qquad\text{or}\qquad
 \Pr(Y=1)\leq
 \frac a2,
 \label{eq:qmp-endpoint-promise-app}
\end{equation}
It is promised that one alternative holds. This problem is
\PromiseBQP-complete under deterministic classical polynomial-time
many-one reductions.
\end{corollary}

\begin{proof}
On the branch $S=0$, no
step record is created and
$H\circ H|0\rangle=|0\rangle$.  On the branch
$S=1$, the step record
dephases $H|0\rangle=|+\rangle$ in the computational context; the probe is
then entangled with the record, so its reduced state is maximally mixed.  The
second Hadamard leaves that reduced state maximally mixed.
Hence
\begin{equation}
 \Pr(Y=1\mid S=0)=0,
 \qquad
 \Pr(Y=1\mid S=1)=\frac12,
 \qquad
 \Pr(Y=1)=\frac12p(m).
 \label{eq:qmp-endpoint-probability-app}
\end{equation}
\setlength{\emergencystretch}{2em}
The gadget sends the verifier thresholds $a,b$ to $a/2,b/2$ and halves
their separation, leaving an inverse-polynomial gap. For hardness, encode
the verifier and pass its thresholds exactly as in the proof of
Theorem~\ref{thm:qmp-promisebqp-app}, padding with identity gate pairs when
needed for the total history-instance length. Equation~\eqref{eq:qmp-endpoint-probability-app}
then gives the required endpoint promise.

For membership, sample $N$ independent runs and compare the endpoint
frequency with $(a+b)/4$. The error is at most
$\exp[-N(b-a)^2/8]$, so
$N=O(\ell^{2c})$ gives bounded error.
\end{proof}

\Needspace{6\baselineskip}
If the schedule bit $S$ is retained,
this $H$--$H$ specialization also gives
the explicit information value
\begin{equation}
 I(S;Y)
 =H_2\!\left(1-\frac{p(m)}{2},
   \frac{p(m)}{2}\right)
   -p(m)
 \quad\text{bits per run},
 \label{eq:qmp-schedule-information-app}
\end{equation}
which is the biased-schedule specialization of
Eq.~\eqref{eq:weighted-js-app}.  This identity links the two constructions
operationally; the completeness proof itself is the circuit reduction above.

\Needspace{7\baselineskip}
\subsection{\texorpdfstring{Exact sampling of musical outputs}{Exact sampling of musical outputs}}
\label{app:music-sampling}

\begin{remark}[Complexity-class terminology]
\label{rem:complexity-classes-app}
\(\mathrm{BPP}\) and \(\mathrm{BQP}\) are the bounded-error classical
randomized and quantum polynomial-time language classes.  An IQP circuit is a
polynomial-size circuit of mutually commuting gates diagonal in the Pauli-
\(X\) basis, followed by computational-basis measurement.  The polynomial
hierarchy \(\mathrm{PH}\) is the union of its alternating quantifier levels;
its third level is \(\Sigma_3^{\mathrm P}\cup\Pi_3^{\mathrm P}\).  These
standard notions are recalled only to state the sampling consequence below.
\end{remark}

\begin{proposition}[Exact musical-output sampling]
\label{prop:music-exact-sampling-app}
Fix a polynomial-time injective encoding of full computational-basis output
strings into polynomial-length musical strings, with a polynomial-time
inverse.  If a classical randomized polynomial-time algorithm, given the
circuit description, samples these musical outputs exactly for every circuit
in the uniform finite-gate IQP families of Bremner, Jozsa, and
Shepherd~\cite{app:BremnerJozsaShepherd2011}, then the polynomial hierarchy
collapses to its third level.
\end{proposition}

\begin{proof}
Applying the inverse encoding to each sample gives an exact classical sampler
for the full IQP output distribution.  The circuit families needed
in~\cite{app:BremnerJozsaShepherd2011} use only a finite set of one- and
two-qubit gates.  Their exact-sampling consequence therefore applies directly.
\end{proof}

This concerns sampling full output strings.  By contrast, an efficient
classical randomized decider for the
\textsc{QMTP} promise would imply
$\mathrm{BQP}\subseteq\mathrm{BPP}$; its decision-completeness theorem alone does not
establish a collapse of the polynomial hierarchy.

\Needspace{7\baselineskip}

\clearpage
\begingroup
\renewcommand{\refname}{Appendix references}

\label{app:last-page}
\endgroup

\end{document}